\documentclass[11pt, letterpaper]{article}
\RequirePackage{amsthm,amsmath,amsfonts,amssymb}
\usepackage[backend=biber,
style=authoryear,
natbib=true]{biblatex}
\usepackage{booktabs}
\usepackage{tikz}
\usetikzlibrary{fadings}
\usetikzlibrary{patterns}
\usetikzlibrary{shadows.blur}
\usetikzlibrary{shapes}
\usepackage{mathtools}
\usepackage{hyperref}
\usepackage{pdflscape}
\usepackage{setspace}

\theoremstyle{plain}
\newtheorem{theorem}{Theorem}
\newtheorem{corollary}{Corollary}

\theoremstyle{remark}
\newtheorem{assumption}{Assumption}
\newtheorem{definition}{Definition}

\newtheorem{lemma}{Lemma}

\newcommand\cP{\mathcal{P}}
\newcommand\bX{\mathbf{X}}
\newcommand\iid{\stackrel{\rm i.i.d.}{\sim}}
\newcommand\cS{\mathcal{S}}
\newcommand\E{\mathbb{E}}
\newcommand\bx{\mathbf{x}}
\newcommand\R{\mathbb{R}}
\newcommand{\cip}{\stackrel{p}{\rightarrow}}
\newcommand{\cid}{\stackrel{d}{\rightarrow}}
\newcommand\var{{\rm Var}}
\newcommand\asyvar{{\rm AVar}}
\newcommand\cov{{\rm Cov}}
\newcommand{\argmin}{\operatornamewithlimits{argmin}}

\newcommand\spacingset[1]{\renewcommand{\baselinestretch}%
  {#1}\small\normalsize}
\spacingset{1.1}

\begin{document}

\newrefsection
\title{Leveraging Population Outcomes to Improve \\ the Generalization of Experimental Results\footnote{The authors would like to thank Nicole Pashley, Dustin Tingley, Tara Slough, the Miratrix CARES Lab, and the UCLA Causal Inference reading group. Melody Huang is supported by the National Science Foundation Graduate Research Fellowship under Grant No. 2146752. Any opinion, findings, and conclusions or recommendations expressed in this material are those of the authors(s) and do not necessarily reflect the views of the National Science Foundation.
}}
\author{Melody Huang\thanks{Ph.D. Candidate, Department of Statistics,
    University of California, Berkeley, CA 94720. Email:
    \href{mailto:melodyyhuang@berkeley.edu}{melodyyhuang@berkeley.edu}, URL:
    \url{https://melodyyhuang.github.io}}
  \hspace{0.2in}
  Naoki Egami\thanks{Assistant Professor, Department of
    Political Science, Columbia University, New York NY 10027. Email:
    \href{mailto:naoki.egami@columbia.edu}{naoki.egami@columbia.edu}, URL:
    \url{https://naokiegami.com}}
  \hspace{0.2in}
  Erin Hartman\thanks{Assistant Professor, Department of Statistics
    and of Political Science, University of California, Berkeley, CA 94720. Email:
    \href{mailto:ekhartman@berkeley.edu}{ekhartman@berkeley.edu}, URL:
    \url{www.erinhartman.com}}
  \hspace{0.2in}
  Luke Miratrix\thanks{Associate Professor, Graduate School of
    Education, Harvard University, Cambridge MA 02138. Email:
    \href{mailto:lmiratrix@gse.harvard.edu}{ lmiratrix@gse.harvard.edu}, URL:
    \url{https://scholar.harvard.edu/lmiratrix/home}}}
\date{\today}
\maketitle

\pdfbookmark[1]{Title Page}{Title Page}
\vspace{-0.35in}
\thispagestyle{empty} 
\setcounter{page}{0} 
\begin{abstract}
 Generalizing causal estimates in randomized experiments to a broader target population is essential for guiding decisions by policymakers and practitioners in the social and biomedical sciences. While recent papers developed various weighting estimators for the population average treatment effect (PATE), many of these methods result in large variance because the experimental sample often differs substantially from the target population, and estimated sampling weights are extreme. To improve efficiency in practice, we propose post-residualized weighting in which we use the outcome measured in the observational population data to build a flexible predictive model (e.g., machine learning methods) and residualize the outcome in the experimental data before using conventional weighting methods. We show that the proposed PATE estimator is consistent  under  the  same  assumptions  required  for  existing  weighting methods, importantly without assuming the correct specification of the predictive model. We demonstrate the efficiency gains from this approach through simulations and our application based on a set of job training experiments. 
\end{abstract}

\noindent
\vspace{0.2in}

\vfill

\newpage

\spacingset{1.6}

\section{Introduction} 

The ``credibility revolution'' has elevated the role of randomized, controlled trials (RCTs), which are praised for their strong internal validity \citep{banerjee2009experimental, falk2009lab, baldassarri2017field}.  Control over the design of the RCT allows researchers to draw causal inferences about treatment effects, within the experimental sample, while imposing minimal assumptions.  This focus on credibility is not without controversy though, with some arguing that the emphasis on causality has led researchers to narrow the scope of their inquiry \citep{huber_2013, deaton2018understanding}.  This debate has revealed a pressing need for methods that allow researchers to generalize the causal impact of treatments, and resulting policy implications, beyond the experimental setting.

In light of this need, a robust literature on methods for generalizing
experimental results to broader populations of interest has emerged.
Often, in practice, the cost of a controlled environment is that the
experiment cannot be conducted on a representative sample of the
target population of interest.  Recent work has outlined the
necessary assumptions for generalizing an experiment to identify the
population average treatment effect (PATE), i.e., the effect of the
experimental treatment in a clearly defined target population that
differs from the experimental sample \citep{cole2010generalizing,
  hartman2015sample, bareinboim2016fusion}. In practice, the most
common approach first models the experimental sample inclusion
probability, with the PATE then estimated using inverse probability
weighted estimators \citep{stuart2011use, tipton2013improving,
  buchanan2018generalizing}. Alternative estimators focus on modeling
treatment effect heterogeneity \citep{kern2016assessing,
  nguyen2017sensitivity} or doubly robust estimation
\citep{dahabreh2019generalizing}. 

Despite these theoretical advances in methods for estimating the PATE,
in practice, weighted estimators are often
imprecise, especially when the experimental sample differs substantially from the target population.  This makes it difficult for policymakers and practitioners to draw conclusions about the impact of treatment in the target population to guide their policy recommendations.  Indeed, researchers
empirically find that weighted  estimators
often increase the mean squared error for the PATE compared to an
estimator ignoring sampling weights because inverse probability
weighting estimators have much larger standard errors, even though they have
smaller bias \citep{miratrix_2018}. More
generally, considering the bias-variance tradeoff,
the cost of large precision loss associated with the conventional
weighting methods makes it unclear if it is
``worth weighting'' and questions the applicability of these weighting
methods commonly used by empirical researchers. 

In practice, these weighting methods often leave a valuable resource on the table --- outcome
data measured in the population.  While inverse probability weighting
methods leverage population data about pre-treatment covariates when
modeling the sampling weights, use of outcome data has primarily
been limited to use in placebo tests \citep{cole2010generalizing,
hartman2015sample}. Recently, the data fusion literature proposed using experimental data to help aid the estimation of causal effects in observational studies (e.g., see \cite{athey2020combining, athey2019surrogate, kallus2020role}). Our proposed approach aims to incorporate observational population data to reduce noise in generalizing experimental results. Population data often have larger sample sizes and therefore provide an opportunity to model covariate-outcome relationships with more flexible modeling approaches.  It is this opportunity --- to incorporate large population data sets that contain outcome data to improve precision --- that serves as the foundation of our method. 

We propose post-residualized weighting to leverage outcome data
measured in the population to improve precision in estimation of the
PATE.  We begin by constructing a predictive model of the outcome
using the population data.  We then use this to residualize the
experimental outcome data, and these residuals replace the
experimental outcome in the standard inverse probability weighting
estimators used for generalization. Identification of the PATE
proceeds under the same assumptions required for existing inverse
probability weighting methods, namely that the sampling weights are
correctly specified.  We show that this estimator is consistent,
regardless of the residualizing model constructed in the population
data. Therefore, we can safely use machine learning methods to build a predictive model. We then establish under what conditions the proposed post-residualized weighting estimator is more efficient than existing methods. 

We also extend our estimator to the weighted least squares framework, which has three advantages: (1) it incorporates the
well-known benefits of stabilized weighting estimators
(i.e. H\`{a}jek estimators), (2) it allows for additional precision
gains from prognostic variables measured only within the experiment,
and (3) it addresses concerns about
scaling differences between the outcomes measured in the experiment
and the population data.  Importantly, we provide a diagnostic that
allows researchers to assess when the post-residualized weighting method is
likely to result in efficiency gains.

The paper proceeds by introducing our empirical application evaluating generalizability of site-specific results for trials conducted under the Job Training Partnership Act, described below.  We then introduce notation and existing methods for
estimating the population average treatment effect from experimental
data in Section \ref{sec:background}. In Section
\ref{sec:recycle_ipw} we introduce post-residualized weighting, prove its statistical properties, and introduce a diagnostic to assess whether researchers should expect efficiency gains in their applications.  We extend these results to
weighted least squares estimation in Section \ref{sec:recycle_ipw_cov}, and discuss a special case in which we include the
predicted outcome as a covariate in Section \ref{subsec:proxy}.  Finally, we
provide simulation evidence supporting the performance of
post-residualized weighting estimators and diagnostic tools in Section
\ref{sec:sims} and apply them to an empirical application evaluating
the Job Training Partnership Act in Section \ref{sec:empirical}. 

\subsection{Background and Data}
To motivate our method, we re-evaluate a foundational experiment that assessed the impact of a job training program. The Job Training Partnership Act (JTPA) was introduced by the U.S. Congress in 1982 to help provide employment and training programs to economically disadvantaged adults and youths.  To assess its effectiveness, the national JTPA study evaluated the impact of the program across a diverse set of 16 experimental sites between 1987 and 1989. The experimental units were individuals who were interviewed and deemed eligible to receive JTPA services. Individuals assigned to treatment were given access to the JTPA services, while those assigned to control were told that the services were not available. The treatment to control ratio was set at 2:1. A follow-up survey 18 months later was then conducted to measure outcomes, such as earnings and employment \citep{bloom1993national}.

We use the same 16 experimental sites from the national JTPA study as
the basis for our analysis. While the original study focused on
four target groups: adult women and men (categorized formally as ages
22 and older), and female and male out-of-school youths (ages 16-21),
we focus our analysis on adult women, the largest target group within the JTPA study.\footnote{The estimated impact of JTPA for the other target groups were not found to be statistically significant in the original study.}  We consider two different outcomes: employment status (binary outcome) and total earnings (zero-inflated, continuous
outcome).  Across the 16 sites, the average effect on earnings was \$1240 and employment was 1.63\%, but point estimates across sites ranged from -\$5210 in Butte, MT to \$3030 in Providence, RI for earnings and -7\% in Butte, MT and Marion, OH to 7\% in Heartland, FL and Providence, RI.  Had a policymaker only run their experiment in Providence, RI, they may have concluded that the treatment was effective, but not so in Butte, MT.  Weighted estimators can adjust for demographic differences across sites, but many of the sites, such as Butte, MT, contain few units, emphasizing the need for precise estimators when generalizing results to other populations.

Unlike the original study, which evaluated the overall effectiveness, our focus is on generalizing the effect.  The multisite design of this experiment serves as an ideal test bed for our method.  We generalize the results of each site individually to a target population defined by the units in the other 15 sites, allowing us to benchmark our estimator against the experimentally identified causal estimate of the excluded sites and evaluate precision gains from post-residualized weighting.  Ultimately, we find between a 5\% and 21\% reduction in variance where our methods are applicable.  A summary of the JTPA experimental set up is provided in Supplementary Materials Table
\ref{tbl:jtpa_summary}.

\section{Existing Estimators for Generalization}
\label{sec:background}
\subsection{Setup}
\label{subsec:notation}
We begin by defining the target population as an infinite
super-population $\cP$ with probability distribution $F$ and probability
density $dF$, for which we wish to infer the effectiveness of
treatment. Following \citet{buchanan2018generalizing}, suppose we
observe $n$ units as the ``experimental sample,'' but, as with most experiments in practice, the selection into the experiment from the target population is biased . Let $\mathcal{S}$ represent the random set of indices for the units in the experimental sample.

Let $T_i$ be the binary treatment variable, 
where $T_i = 1$ for units assigned to treatment, and $T_i = 0$ for
control. Using the potential outcomes framework \citep{neyman1923,
  rubin1974causal}, we define $Y_i(t)$ to be the potential outcome of
unit $i$ that would realize if unit $i$ receives the treatment $T_i =
t$, where $t \in \{0, 1\}$. For each unit in the experiment, only one
of the potential outcome variables can be observed, and the realized
outcome variable for unit $i$ is denoted by $Y_i = T_i Y_i(1) + (1 -
T_i) Y_i(0).$ We also observe pre-treatment covariates $\bX_i$ for
units in the experiment. We use $\tilde{F}$ to represent the sampling distribution for the
experimental sample, i.e., $\{Y_i(1), Y_i(0), T_i, \bX_i\}_{i=1}^n
\iid \tilde{F}$ with density $d\tilde{F}$. Because we consider
settings where the selection into the experiment from the target
population $\cP$ is biased, $F \neq \tilde{F}.$ 

We assume that the treatment assignment is randomized within the experiment.
\begin{assumption}[Randomization within Experiment]
  \label{assum-exp}
  \begin{equation}
d\tilde{F}(Y_i(1), Y_i(0), T_i, \bX_i) = d\tilde{F}(Y_i(1),
    Y_i(0), \bX_i) \cdot d\tilde{F}(T_i)
  \end{equation}
\end{assumption} 
\noindent Under this assumption, it is well known that the sample average treatment effect (SATE) can be estimated without bias using a difference-in-means estimator: 
\begin{equation}
  \widehat{\tau}_\cS = \frac{1}{\sum_{i \in \cS} T_i } \sum_{i \in \cS} T_i Y_i - \frac{1}{\sum_{i \in \cS} (1-T_i)} \sum_{i \in \cS} (1-T_i) Y_i.
\end{equation}

This within-experiment estimand, the SATE, is important for evaluating the effectiveness of treatment. However, researchers often want to know to what extent the findings are externally valid to the target population \citep{cole2010generalizing, miratrix_2018, egami2020elements}. This population level estimand, the population average treatment effect (PATE), is our primary causal quantity of interest and is formally defined as:
\begin{equation}
\tau \coloneqq \E_{F}\{Y_i(1) - Y_i(0)\},
\end{equation}
where the expectation is taken over the target population distribution
$F$. When the experimental sample is randomly drawn from the target
population $F = \tilde{F}$, $\widehat{\tau}_{\cS}$ can be used as an
unbiased estimator for $\tau$. However, in most settings, experimental
units are not randomly drawn from the target population with equal
probability.  

To estimate the PATE, we also assume we observe an $i.i.d.$ sample of 
$N$ units from the target super-population $\cP$ as the ``population
data,'' which is separate from the experimental sample. This design is
most common in the social sciences, and is called the non-nested
design in that the experimental sample is not a subset of the population data \citep{colnet2020causal}.\footnote{While we
  focus on the non-nested design in this paper, the same proposed
  approach is useful for the nested design where the experimental sample
  is a subset of the population data. The main difference arises in the
  analytical expressions of the efficiency gain from our proposed approach.} Typically, the size of the population data is much larger than
the experimental data, i.e., $N \gg n$. In the conventional setup, researchers only observe
pre-treatment covariates $\bX_i$ for each unit $i$ in the population
data. In the next subsection, we review assumptions and estimators for the PATE under this
conventional setup. In Section~\ref{sec:recycle_ipw}, we then consider
our setting in which researchers also observe an outcome measure in
addition to pre-treatment covariates in the population
data. Importantly, because the treatment is not randomized in the
population data, we cannot identify the PATE just using the population
data. 

\subsection{Assumptions}
\label{subsec:assumptions}
We make the standard assumptions of no interference and that treatments are identically administered across all units (i.e., SUTVA, defined in \citet{rubin1980SUTVA}). In order to identify the PATE using experimental data,
we require additional assumptions about sampling of the experimental units. First, we assume that, conditional on a set of pre-treatment covariates $\bX_i$, the sample selection mechanism is ignorable. More formally,  
\begin{assumption}[Ignorability of Sampling and Potential Outcomes]\mbox{}\\
\label{assum-ind}
\begin{equation}
 dF(Y_i(1), Y_i(0) \mid \bX_i = \bx) =   d\tilde{F}(Y_i(1), Y_i(0) \mid \bX_i = \bx)
\end{equation}

\end{assumption} 

Assumption \ref{assum-ind} states that, conditional on $\bX_i$, the
distribution of the potential outcomes $\{Y_i(1), Y_i(0)\}$ is the same across the experimental sample and the target
population \citep{stuart2011use, Pearl:2014hb, kern2016assessing}.\footnote{For identification of the PATE, a weaker assumption of conditional ignorability of sampling and treatment effect heterogeneity may be invoked instead. However, the variance derivations rely on the conditional ignorability of sampling and potential outcomes.} 
We also assume that given the pre-treatment covariates $\bX_i$, there is a positive probability of being included in the experimental sample \citep{westreich2010invited}. 

\clearpage
\begin{assumption}[Positivity]\mbox{}\\ For all $\bx$ with $dF(\bX_i = \bx)  > 0,$ we have
  \label{assum-pos}
  \begin{equation} 
  dF(\bX_i = \bx)  > 0 \Rightarrow d \tilde F(\bX_i = \bx) > 0.
  \end{equation} 
\end{assumption}

\subsection{Estimation of PATE}
\label{subsec:ipw}

There is a robust, and growing, literature on methods for estimating
the PATE.  The most common approach is the inverse probability
weighting estimator (IPW) \citep{cole2010generalizing}. The IPW
estimator relies on sampling weights usually defined as an inverse of
the probability of being sampled into the experiment. In our case,
given the infinite superpopulation defined by $F$, we first define a
relative density as follows.
\begin{equation} 
  \pi(\bX_i)  =  \frac{d\tilde{F}(\bX_i)}{dF(\bX_i)}. \label{eq:r-d}
\end{equation} 
Sampling weights are proportional to the inverse of this relative density. For each unit $i$, 
\begin{equation*} 
  w_i  \propto  \frac{1}{\pi(\bX_i)}.
\end{equation*}

Weights are typically estimated using a binary outcome model,
such as logistic regression \citep{stuart2011use, o2014generalizing, buchanan2018generalizing} by
exploiting the fact that weights are
proportional to the relative probability of being in the observed population data to the probability of being in the experimental
sample, conditional on being in either set:
\begin{equation*} 
  w_i  \propto \frac{\Pr(S_i = 0 \mid \bX_i)}{\Pr(S_i = 1 \mid \bX_i)},
\end{equation*} 
where $S_i$ takes on a value of $1$ if the unit belongs to the experimental sample, and $0$ if the unit belongs to the observed population data. Researchers can estimate $\Pr(S_i = 1 \mid
\bX_i)$ and $\Pr(S_i = 0 \mid \bX_i)$ using a binary
outcome model where we model $S_i = 1$ with $\bX_i$ by
stacking the experimental data and population data \citep{stuart2011use,
  buchanan2018generalizing, egami2019covariate}. Alternatively, researchers can use balancing methods,
such as entropy balancing, which estimates weights such that weighted
moments (e.g., means of each pre-treatment covariate $\bX_i$) of the experimental data is equal to moments of the observed
population data \citep{deville1992, hainmueller2012entropy, hartman2015sample}.

Once researchers have estimated the sampling weights, the PATE can be
estimated using a weighted estimator, also known as the H\`{a}jek
estimator: 
\begin{equation}
  \hat \tau_{W} \coloneqq \frac{\sum_{i \in \cS} \hat w_i T_i
    Y_i}{\sum_{i \in \cS} \hat w_i T_i}- \frac{\sum_{i \in \cS} \hat w_i
    (1-T_i) Y_i}{\sum_{i \in \cS} \hat w_i (1-T_i)}. \label{eq:std}
\end{equation}

This weighted estimator can be computed using a weighted least squares regression of the outcome on an intercept and the treatment indicator, with the estimated weights $\hat w_i$. The weighted estimator is equivalent to the estimated coefficient of the treatment indicator.

Under Assumption~\ref{assum-exp}--\ref{assum-pos} and the consistent estimation of the
sampling weights, the weighted estimator $\hat \tau_W$ is a consistent
estimator of the PATE \citep{buchanan2018generalizing}. However, in
practice, weighted estimators can suffer from large variance due to
extreme weights. This problem has been highlighted in the
observational causal inference literature with respect to inverse
propensity score weighted estimators, in which large imbalances between
treatment and control groups can result in extreme weights
\citep{kang2007demystifying, stuart2010matching}. This issue is often
exacerbated in the generalization setting, where imbalances between a
convenience experimental sample and target population can be
relatively large.  As a result, losses in precision from weighting can be challenging to overcome when
generalizing from the SATE to the PATE \citep{miratrix_2018}. 

\section{Post-Residualized Weighting}\label{sec:recycle_ipw}
Existing methods, such as the weighted estimator $\hat \tau_W$
described above, require pre-treatment covariate data, measured in both
the experimental sample and target population, for estimating the
sampling weights. However, researchers often have access to an
outcome variable in the observational population data as well. Our proposed method, \textit{post-residualized
weighting}, aims to improve precision in the estimation of the PATE by
leveraging this outcome variable measured in the observational population data.  

\subsection{Setup}
In contrast to the conventional setup, we consider settings where
researchers observe an outcome variable in addition to pre-treatment
covariates in the population data. See Figure~\ref{fig:data} for a
visualization of the difference in settings from conventional methods.  Below we describe two canonical social science examples that motivate the data settings that underpin our method and to which we return.  We also come back to our benchmark analysis of the JTPA data in Section~\ref{sec:empirical}.

\paragraph*{Example: Get-Out-the-Vote (GOTV) Experiments}
Political scientists have conducted a number of field experiments to
evaluate the impact of canvassing efforts, including door-to-door,
phone, and mail, on voter turnout. Such GOTV experiments typically rely on administrative data to
measure the outcome, namely voter turnout data from the Secretary of State. These experiments are often
conducted in a small geographic region (e.g., New Haven, Connecticut
in \cite{gerber_green_2000}), but scholars are often interested in
generalizing the effect to broader populations, such as for a statewide
election. Importantly, when considering generalization, the outcome variable of voter turnout is available 
not only for the experimental data but also for the broader target
population of interest. In our framework, we use this information
about voter turnout measured in the observational population data to
improve precision in the estimation of the PATE.

\paragraph*{Example: Education Experiments}
Education research also relies on experiments to evaluate the
performance of classroom interventions, such as the impact of smaller
class size on curriculum-based and standardized tests \citep[e.g.,][]{word1990state}. These experiments are often done in
partnership with school systems. For example, the Tennessee STAR
experiment was conducted in classrooms across Tennessee. However,
researchers are interested in the broader impact of such
interventions.  For example, a researcher may ask what the long term
impact of small class sizes in primary school is on standardized test
scores, such as the SAT, for all public schools in the United
States. To estimate the PATE, existing methods use demographic
variables from a random sample of public school students to construct sampling weights. In our framework, we can additionally use SAT scores measured for a random sample of
public school students, which improves estimation accuracy. 

\paragraph*{Remark} For simplicity of exposition, this section focuses on settings where we observe the same outcome measure in the experimental and population data. The outcomes measured in the population may be a mix of treatment and control outcomes. However, our proposed method can also accommodate scenarios in which we only observe a proxy outcome variable (rather than the same outcome
measure) in the population data. We consider this case in
Section~\ref{subsec:proxy}. \qed

\begin{figure}[!t]
\centering 
\includegraphics[trim={4cm 9cm 4cm 10cm},clip,scale=0.7]{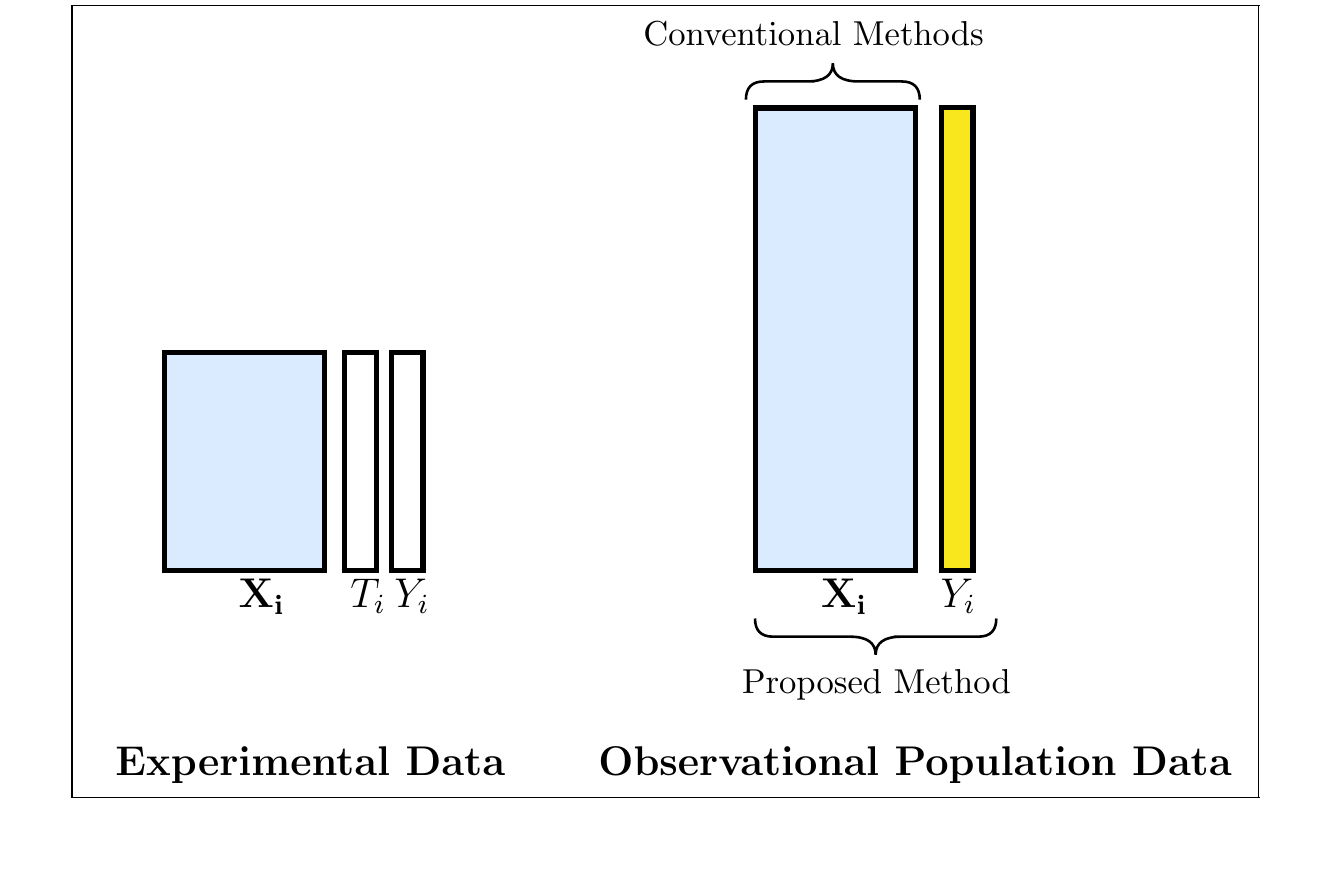}
\caption{Data Requirements. Conventional estimation methods only use
  the covariate data $\bX_i$ (highlighted in blue, above). Our
  proposed approach leverages the outcome data, in addition to the
  covariate data at the population level (highlighted in
  yellow).}\label{fig:data}
\end{figure} 

\subsection{Post-residualized Weighted Estimators}\label{sec:post_res_weighted}
Our proposed post-residualized weighting approach exploits the outcome
measured in the population data to improve precision in the estimation
of the PATE. The key idea is that we estimate a predictive model
with the outcome measured in the population data and then use this
estimated predictive model to \textit{residualize} outcomes in the
experimental data, before using conventional weighting estimators for
the PATE. 

In total, post-residualized weighting has four steps. The first step is to estimate sampling
weights $w_i$, which is the same as the conventional weighting
approach. In the second step, we fit a flexible model in the population data to predict the outcome
variable $Y_i$ using pre-treatment $\bX_i$. We refer to this
predictive model fitted in the population data as a
\textit{residualizing model}, and formally denote it as $g(\bX_i)$: $\mathcal{X}
\to \R$ where $\mathcal{X}$ is the support of $\bX_i$. In the third
step, we use the estimated residualizing model to predict outcomes $\widehat{Y}_i$ in the experimental data, which is separate
from the population data used to estimate the residualizing
model. In the fourth and final step, we apply the weighted estimator
(equation~\eqref{eq:std}) using the residuals from this prediction, (denoted by $\hat e_i = Y_i -\hat Y_i$) as outcomes (instead of $Y_i$ used in the conventional weighting approach).

We summarize our proposed approach in Table~\ref{tab:summary}. In the following section, we directly extend the weighted estimator discussed in Section~\ref{sec:background}. We then consider how post-residualizing can improve a more general weighted least squares estimator that includes further covariate adjustment in Section~\ref{sec:recycle_ipw_cov}. 

\begin{table}[!t]
\small
  \begin{center}
    \begin{tabular}{r|p{5in}}
      \toprule
      \multicolumn{2}{l}{\textbf{Post-residualized Weighting for the PATE estimation:} } \\
      \midrule
      Step 1: & Estimate sampling weights, $w_i$, for units in the experimental sample. \\[8pt]
      Step 2: & Choose a residualizing model $g(\bX_i)$: $\mathcal{X}
                \to \R$, where $\mathcal{X}$ is the support of
                $\bX_i$. Using the population data, estimate $\hat
                g(\bX_i)$ that predict the population outcomes using pre-treatment covariates $\bX_i$.
      \\[8pt] 
      Step 3: & Predict $\hat Y_i = \hat{g}(\bX_i)$ for each unit in the
                experimental data, and compute residual $\hat e_i =
                Y_i - \hat Y_i$ for units in the experimental sample. \\[8pt]
      Step 4: & Estimate the PATE using residuals $\hat
                e_i$ and estimated sampling weights $\hat w_i$.
      \\[3pt] 
              & \emph{No covariate adjustment within the experimental
                data} (Section~\ref{sec:recycle_ipw})
      \\
              & $\qquad$\rotatebox[origin=c]{180}{$\mathbf{\Lsh}$}
                See post-residualized weighted estimator
                $\hat \tau_W^{res}$ (Definition~\ref{prw}).
      \\[5pt] 
              & \emph{With covariate adjustment within the experimental
                data} (Section~\ref{sec:recycle_ipw_cov})
      \\
              & $\qquad$\rotatebox[origin=c]{180}{$\mathbf{\Lsh}$} See
                post-residualized weighted least squares estimator
                $\hat \tau_{wLS}^{res}$ (Definition~\ref{pr-wls}).\\
      \bottomrule
    \end{tabular}
    \end{center}
    \caption{Summary of Post-residualized
      Weighting.}\label{tab:summary}
\end{table}

\begin{definition}[Post-residualized Weighted Estimator] 
  \label{prw}
  Let $\hat w_i$ be estimated sampling weights. 
  Define $\hat e_i$ to be residuals from the residualizing model prediction (i.e., $\hat e_i = Y_i - \hat Y_i$). The post-residualized weighted estimator is defined as: 
  \begin{equation}
    \hat \tau_W^{res} \coloneqq \frac{\sum_{i \in \cS} \hat w_i T_i \hat
      e_i}{\sum_{i \in \cS} \hat w_i T_i} - \frac{\sum_{i \in \cS} \hat w_i
      (1-T_i) \hat e_i}{\sum_{i \in \cS} \hat w_i (1-T_i)}. \label{eq:prw}
  \end{equation}
\end{definition}
We summarize several key aspects of the post-residualized weighted estimator here and formally discuss each point in the subsequent sections.
First, the identification of the PATE is obtained under the same assumptions required for existing weighted estimators, and we do not make any additional assumptions (Section~\ref{subsec:IPW_consistency}). Most importantly, our proposed estimator is consistent for the PATE, regardless of the choice of the residualizing model. That is, we do not require the correct specification of the residualizing model $g (\bX_i)$ to guarantee  consistency of the proposed estimator. Therefore, akin with \citet{rosenbaum2002covariance} and \citet{sales2018rebar}, the residualizing model $g (\bX_i)$ can be seen as an ``algorithmic model'' in that the goal is to predict outcomes, rather than substantively explain an underlying probabilistic process. 

Second, the proposed post-residualized weighted estimator, $\hat \tau_{W}^{res}$, can achieve significant improvements in precision over the traditional weighted
estimator (equation~\eqref{eq:std}) when the residualizing model can predict outcomes in the experiment well (Section~\ref{subsec:IPW_eff_gains}). We will show in
Section~\ref{subsec:IPW_eff_gains} that while we maintain consistency regardless, how much efficiency gain we achieve depends on the predictive performance of the fitted residualizing model $\widehat{g}(\bX_i)$. As such, researchers should, when possible, use not only simple models, such as ordinary least squares, but also more flexible machine learning models, such as random forests or other ensemble learning methods \citep{breiman2001random, polley2010super} as the residualizing models to improve precision of the PATE estimation.

Finally, we derive a diagnostic measure that researchers can use to determine whether residualizing will likely lead to precision gains when estimating the PATE (Section~\ref{subsec:IPW_diagnostics}). As emphasized in the second point above, when the residualizing model can predict outcomes in the experiment well, we can expect efficiency gains. However, when the residualizing model fails to predict outcome measures in the experimental data, it is possible for post-residualizing
to increase uncertainty of the PATE estimation. Our diagnostic measure helps researchers to estimate the expected efficiency gain, thereby deciding whether residualizing is beneficial in their applications. 

\paragraph*{Remark}
Our proposed post-residualized weighting estimator is closely
connected to the augmented inverse probability weighted estimators (AIPW) \citep{robins1994} developed for the
PATE \citep{dahabreh2019generalizing} in that both estimators combine
weighting and outcome-modeling. The process of estimating weights for both the post-residualized weighting estimators and AIPW is the same. However, the key difference between two approaches is that the AIPW estimates the outcome model using only the experimental data, thereby not exploiting the outcome variable available in the population data. In
contrast, our post-residualized weighting estimator explicitly
uses the outcome information available in the population data to
estimate the residualizing model and improve precision. Furthermore, post-residualized weighting does not attempt to model both the treatment and control outcomes separately, and therefore, does not have the double robustness that the AIPW has. \qed

\subsection{Consistency}\label{subsec:IPW_consistency}
In this section, we show that the post-residualized weighted estimator is a consistent estimator of the PATE regardless of the choice of the residualizing model $g(\bX_i)$. This emphasizes the point that $g(\bX_i)$ need not be a correct specification of the underlying data generating process, but merely a function that predicts outcomes measured in the population. 

\begin{theorem}[Consistency of Post-residualized Weighted Estimators]\label{thm:consistency_ipw} Assume that sampling weights $\hat w_i$ are consistently estimated and Assumptions~\ref{assum-exp}--\ref{assum-pos} hold with pre-treatment covariates $\bX_i$. Then, the post-residualized weighted estimator, using any residualizing model $g(\bX_i)$ built on the population data, is a consistent estimator for the PATE:  
  \begin{equation*}
    \hat{\tau}_{W}^{res} \cip \tau, 
  \end{equation*}
  where $\cip$ denotes the convergence in probability.
\end{theorem} 
The proof of Theorem~\ref{thm:consistency_ipw} can be found in
Supplementary Materials~\ref{app:proofs}. 
This property allows for a large degree of
flexibility in building the residualizing model, since consistency is
guaranteed \textit{regardless} of model specification or performance
of $g(\bX_i)$. We can obtain the consistency even for a misspecified
residualizing model $g(\bX_i)$ because the predicted experimental
outcome $\hat{Y_i} = \hat{g}(\bX_i)$ is only a function of the
pre-treatment covariates $\bX_i$, and thus, with randomized treatments
(Assumption~\ref{assum-exp}), its distribution is the same across
treatment and control units on average for any sample size. As such,
residualizing preserves the consistency of the original weighted
estimator without requiring any additional assumptions.   

While consistency is guaranteed, efficiency gains from residualizing \textit{do} depend on the ability of the residualizing model to predict outcome measures in the experimental data. Theorem~\ref{thm:consistency_ipw} allows for researchers to leverage complex, ``black box'' approaches (such as ensemble methods) to maximize the predictive accuracy, as interpretability of the residualizing model is secondary to being able to fit the data well. In the next section, we will formalize the criteria for variance reduction from residualizing. 

\subsection{Efficiency Gains} \label{subsec:IPW_eff_gains}
The post-residualized weighted estimator allows researchers to include information from the observational population data about the relationship between the pre-treatment covariates and the population outcomes into the
estimation process.  Whether or not we obtain precision gains, and the magnitude of these precision gains, will depend on the nature of the residualizing model. In general, the better researchers are able to explain the outcomes measured in the experiment using the residualizing model, the greater the efficiency gains. 

To make these gains more explicit, we first define the
\textit{weighted variance} and \textit{weighted covariance} as follows. 
\begin{eqnarray}
  \var_w(A_i) & = & \int \frac{1}{\pi(\bX_i)^2} \cdot (A_i - \bar A)^2  d\tilde F(\bX_i, A_i), \label{eq:w-var}\\
  \cov_w(A_i, B_i) & = & \int \frac{1}{\pi(\bX_i)^2} \cdot (A_i - \bar A)(B_i - \bar B)
                         d\tilde F(\bX_i,
                         A_i, B_i), \label{eq:w-cov}
\end{eqnarray}
where $\bar A = \E_F(A_i)$ and $\bar B = \E_F(B_i)$. The
efficiency gain for the post-residualized weighted estimator is formalized as follows. 
\begin{theorem}[Efficiency Gain for Post-residualized Weighted Estimators]\mbox{}\\
  The difference between the asymptotic variance of $\hat
  \tau_W^{res}$ and that of $\hat \tau_W$ is: 
  \begin{align} 
\asyvar_{\tilde F}(\hat \tau_W)&-\asyvar_{\tilde F}(\hat \tau_W^{res})  \nonumber \\
              &= -\frac{1}{p(1-p)}\var_w(\hat Y_i) +\frac{2}{p} \cov_w(Y_i(1), \hat Y_i) + \frac{2}{1-p} \cov_w(Y_i(0), \hat Y_i),                
  \end{align} 
  where $\asyvar_{\tilde F}(Z)$ denotes the scaled asymptotic variance
  of random variable $Z$ over the sampling distribution $\tilde F$,
  i.e., $\asyvar_{\tilde F}(Z) = \lim_{n \rightarrow \infty}
  \var_{\tilde F}(\sqrt{n}Z).$ $p$ is the probability of being treated
  within the experiment, i.e., $p = \Pr_{\tilde{F}}(T_i = 1).$
\label{thm:var_compare}
\end{theorem}

The proof of Theorem~\ref{thm:var_compare} can be found in
Supplementary Materials~\ref{app:proofs}. 
Theorem \ref{thm:var_compare} decomposes the efficiency gain from
post-residualized weighting into two components: (1) the variance of
the predicted experimental outcomes $\var_w(\hat Y_i)$, and (2) how
related the predicted outcomes are to the actual outcomes in the
experimental samples (represented by $\cov_w(Y_i(1), \hat Y_i)$ and
$\cov_w(Y_i(0), \hat Y_i)$). If the covariance between the predicted
outcomes and actual outcomes in the experimental sample is greater
than the variance of the predicted outcomes, we expect precision
gains. In other words, the gains to precision from residualizing
depend on how well outcome measures in the experiment are explained by
the residualizing model fitted to the population data.\footnote{We note that the efficiency gain expression does not include uncertainty associated with estimating the residualizing model. This is because the chosen $\hat g(\bX_i)$ is a dimension reducing function of the fixed pre-treatment covariates.} As such, researchers should leverage the large amounts of data available at the population level to apply flexible modeling strategies in order to maximize the variation explained by the residualizing model.

In the following subsection, we will describe a diagnostic measure that can help researchers determine whether or not they should expect precision gains from residualizing.

\subsection{Diagnostics} \label{subsec:IPW_diagnostics}
As discussed above, while post-residualized weighting 
stands to greatly improve precision in estimation of the PATE, this is not guaranteed. To address this concern, we
derive a diagnostic that evaluates when researchers should expect precision gains from residualizing. 

We define a pseudo-$R^2$ measure as: 
\begin{equation} 
R_0^2 \coloneqq 1 - \frac{\var_{w}(\hat e_i(0))}{\var_{w}(Y_i(0))},
\label{eqn:r2} 
\end{equation}
where we define $\hat{e}_i(t) = Y_i(t) - \hat{Y}_i$ for $t \in \{0,1\}.$

$R^2_0$ can be interpreted as the weighted goodness-of-fit of the residualizing model for the potential outcomes under control for units in the experiment. Researchers can estimate $R^2_0$ using the estimated residuals across the control units in the experiment. When $R_0^2 > 0$, we expect an improvement in precision across the control units from residualizing. 

In line with Rubin's ``locked box'' approach \citep{rubin2008objective}, we do not suggest estimating the analogous $R_0^2$ among treated units. However, if the variation in the control outcomes is greater than the overall treatment effect heterogeneity, then checking if $R^2_0$ is greater or less than zero is an effective diagnostic for whether or not we expect precision gains from residualizing. We formalize this in the following corollary, where we write the relative reduction from residualizing as a function of this proposed $R^2_0$ measure. 

\clearpage
\begin{corollary}[Relative Reduction from Residualizing] \mbox{}\\
  \label{cor:rel_reduc} 
  With $R^2_0$ defined as in equation~\eqref{eqn:r2}, define $R^2_1$ as the weighted goodness-of-fit of the residualizing model for the potential outcomes under treatment. Let $\xi = R^2_0 - R^2_1$, such that: 
  $$R^2_1 \coloneqq 1 - \frac{\var_{w}(\hat e_i(1))}{\var_{w}(Y_i(1))} = R_0^2 - \xi.$$
 Furthermore, define the ratio $f = p \var_{w}(Y_i(0)) / (1-p) \var_{w}(Y_i(1))$.
  Then the relative reduction in variance from residualizing is given by: 
  $$\text{Relative Reduction} \coloneqq \frac{\asyvar_{\tilde F}(\hat \tau_W) - \asyvar_{\tilde F}(\hat
    \tau_W^{res})}{\asyvar_{\tilde F}(\hat \tau_W)} = R^2_0 - \frac{1}{1+f} \cdot
  \xi$$ 
\end{corollary}

Corollary~\ref{cor:rel_reduc}, proof available in Supplementary Materials~\ref{app:proofs}, 
decomposes the overall relative
reduction in variance of the weighted estimator from residualizing
into two components: (1) our proposed diagnostic measure $R^2_0$ and
(2) a factor, represented by $\xi$, that measures the difference in
prediction error between the experimental control and experimental
treated potential outcomes. If the residualizing model explains similar amounts
of variation across both the treated and control potential outcomes, then $R^2_1
\approx R^2_0$ and $\xi \approx 0$. In that scenario, $R^2_0$ will be
roughly indicative of the expected relative reduction.  When $R^2_0$
takes on a negative value, this is a strong indication that
residualizing is unlikely to result in precision gains, since it is
unlikely the prediction error will be significantly lower for treated
units.

To summarize, $R_0^2$ can diagnose when one should expect improvements in precision from
residualizing. When $R^2_0$ takes on negative values, researchers should not proceed with residualizing, as
it is likely to result in precision loss. 

\section{Post-residualized Weighting with Covariate Adjustment}
\label{sec:recycle_ipw_cov}
In Section~\ref{sec:recycle_ipw}, we showed that the post-residualized
weighted estimator is a consistent estimator of the PATE, regardless
of the residualizing model that researchers use.  However,
researchers often rely on covariate adjustment to experimental data
to improve precision in estimation.  We now extend our
post-residualized weighted estimator to include a standard covariate
adjustment for the experimental data. 

As with estimation of the SATE, including covariate adjustment when estimating
the PATE can combat the precision loss associated with the weighted estimators.
Additionally, while estimation of weights requires covariates to be measured across both the population
and the experimental data, covariate adjustment can leverage covariates that are only measured in the
experimental data, where researchers can measure prognostic variables
\citep{stuart2017generalizing}. The weighted least squares estimator $\hat \tau_{wLS}$ for the PATE is estimated using a weighted regression, regressing the outcomes on the treatment indicator and pre-treatment covariates. More formally,
\begin{equation}
  (\hat{\tau}_{wLS}, \hat{\alpha}, \hat{\gamma}) =
  \argmin_{\tau, \alpha, \gamma}
  \frac{1}{n} \sum_{i \in \cS}\hat{w}_i \left(Y_i - (\tau T_i +
  \alpha + \widetilde{\bX}^\top_i \gamma)\right)^2 \label{eq:wls}
\end{equation}
where $\widetilde{\bX}_i$ are experimental pre-treatment covariates
included in the covariate adjustment.  Covariates
$\widetilde{\bX}_i$ can differ from the $\bX_i$ required for
Assumptions~\ref{assum-ind}--\ref{assum-pos}. This weighted least
squares estimator $\hat{\tau}_{wLS}$ is consistent for the PATE
under Assumptions~\ref{assum-exp}--\ref{assum-pos} as long as the
sampling weights $\hat{w}_i$ is consistently estimated \citep{dahabreh2019generalizing}.

By extending the weighted least squares estimator
(equation~\eqref{eq:wls}), we formally define the post-residualized weighted least squares estimator as follows: 
\begin{definition}[Post-Residualized Weighted Least Squares
  Estimator]\label{pr-wls} $ $\\
  Given a residualizing model estimated as $\hat{g}(\cdot)$, the
  post-residualized weighted least squares estimator $\hat{\tau}^{res}_{wLS}$ for the PATE is defined as, 
  \begin{equation}
    (\hat{\tau}^{res}_{wLS}, \hat{\alpha}^{res}, \hat{\gamma}^{res}) =
    \argmin_{\tau, \alpha^{res}, \gamma^{res}}
    \frac{1}{n} \sum_{i \in \cS}\hat{w}_i (\hat e_i - \tau T_i -
    \alpha^{res} - \widetilde{\bX}^\top_i \gamma^{res})^2 \label{eq:prw-wls}
  \end{equation}
  where $\hat{e}_i = Y_i - \hat{g}(\bX_i)$ and $\widetilde{\bX}_i$ are experimental pre-treatment covariates included in the covariate adjustment.  We allow $\widetilde{\bX}_i$ to differ from $\bX_i$ used to calculate $\hat{g}(\bX_i)$.
\end{definition}

In practice, the post-residualized weighted least squares estimator
can be estimated by running a weighted regression, where the estimated
residualized values $\hat e_i$ is regressed on the treatment indicator
$T_i$ and covariates $\widetilde{\bX}_i$, and using the sampling
weights $\hat{w}_i$ as the weights. The coefficient of the treatment
indicator is the post-residualized weighted least squares
estimate for the PATE. If no pre-treatment
covariates are included, the post-residualized weighted least squares
estimator is equivalent to the post-residualized weighted estimator
(equation~\eqref{eq:prw}) discussed in Section~\ref{sec:recycle_ipw}.

There are two advantages to the post-residualized weighted least
squares estimator. First, it can leverage precision gains from
pre-treatment covariates that are measured in the experimental data
but not in the population data. That is, $\widetilde{\bX}_i$
can include more covariates than $\bX_i$. Second, $\hat
\tau^{res}_{wLS}$ provides additional robustness over the 
post-residualized weighted estimator $\hat \tau^{res}_{W}$. More
specifically, without further covariate adjustment, residualizing can
be sensitive to differences between the population and experimental
units in the covariate-outcome relationships. As illustrated
in Section~\ref{subsec:IPW_eff_gains}, when this difference is large,
residualizing can result in efficiency loss. However, by
performing covariate adjustment on the residualized outcomes in the
experimental data, we have an opportunity to correct for the difference in the
covariate-outcome relationships between the experimental data and the
population data. In other words, the
post-residualized weighted least squares estimator, $\hat
\tau^{res}_{wLS}$, gives researchers two opportunities to combat the
precision loss of weighting: once from using the
population data in the residualizing process, and a second from
adjusting for covariates in the experimental data.   

\subsection{Consistency} 
We show that, much like the post-residualized weighted estimator $\hat \tau^{res}_{W}$, the post-residualized weighted least squares
estimator $\hat \tau^{res}_{wLS}$ is also a consistent estimator of
the PATE, regardless of the chosen residualizing model $g(\bX_i)$ and
pre-treatment covariates $\widetilde{\bX}_i$ that researchers adjust for in the weighted least squares estimator. 
\begin{theorem}[Consistency of Post-residualized Weighted Least Squares Estimators] 
  \label{thm:wLS_consistency}
  Assume that sampling weights $\hat w_i$ are consistently estimated
  and Assumptions~\ref{assum-exp}--\ref{assum-pos} hold with pre-treatment covariates
  $\bX_i$. Then, the post-residualized weighted least squares
  estimator that adjusts for pre-treatment covariates
  $\widetilde{\bX}_i$ $($equation~\eqref{eq:prw-wls}$)$ is a consistent estimator 
  $$\hat \tau_{wLS}^{res} \cip \tau,$$
  with any residualizing model $g(\bX_i)$ and any pre-treatment covariates
  $\widetilde{\bX}_i$.
\end{theorem} 
This theorem follows closely from Theorem~\ref{thm:consistency_ipw}, with a proof available in Supplementary Materials~\ref{app:proofs}.
As before, we highlight that no additional assumptions are needed to establish consistent estimation
of the PATE.    

A potential concern with covariate adjustment is that performing
covariate adjustment within the experimental data can result in worsened asymptotic precision and
invalid measures of uncertainty \citep{freedman2008}. An
alternative approach is to include interaction terms between the treatment
indicator and covariates \citep{lin2013}. Regardless, we can compute valid
standard errors with the Huber–White sandwich standard error estimator. 

\subsection{Efficiency Gain} \label{subsec:eff_gain_wls} 
In this section, we discuss the efficiency gains that can be obtained
from residualizing. Like the
weighted estimator case, we expect that when outcomes measured in the
experiment are better predicted by the residualizing model $\hat
g(\bX_i)$, the efficiency gains from residualizing is larger. However,
because there is an additional opportunity for covariate adjustment
using covariates measured in the experiment $\widetilde{\bX}_i$, the residualizing model must explain
variation in outcomes that cannot be explained using covariates $\widetilde{\bX}_i$ in the weighted least squares regression. We formalize this below.  

\begin{theorem}[Efficiency Gain for Post-Residualized Weighted Least
  Squares Estimators]
  \label{thm:eff_gain_wls} 
 The difference between the asymptotic variance of $\hat \tau_{wLS}$ and that of $\hat \tau_{wLS}^{res}$ is: 
\begin{align}
  & \asyvar_{\tilde F}(\hat \tau_{wLS}) - \asyvar_{\tilde F}(\hat \tau_{wLS}^{res}) \nonumber \\
\ = \ & \frac{1}{p} \left\{\var_w(Y_i(1) - \tilde \bX_i^\top \gamma_\ast) - \var_w(Y_i(1) - \hat{g}(\bX_i)) \right\} \nonumber \\
 & \underbrace{\qquad + \ \frac{1}{1-p} \left\{ \var_w(Y_i(0) - \tilde \bX_i^\top \gamma_\ast)-
    \var_w(Y_i(0) - \hat{g}(\bX_i)) \right\}}_{\text{\normalfont{\small (a) Explanatory power of residualizing model over linear regression}}}\nonumber \\
 & \ + \ \underbrace{\frac{2}{p} \cov_w(\hat e_i(1), \tilde \bX_i^\top
   \gamma_\ast^{res}) + \frac{2}{1-p} \cov_w(\hat e_i(0), \tilde \bX_i^\top
   \gamma_\ast^{res}) - \frac{1}{p(1-p)} \var_w(\tilde \bX_i^\top
   \gamma_\ast^{res})}_{\mathclap{\text{\normalfont{\small (b) Remaining variation in
   residualized outcomes explained by linear regression on $\tilde \bX_i$}}}}, \label{eqn:wLS_diff}
\end{align} 
where ${\gamma}_\ast$ and $\gamma^{res}_\ast$ are the true coefficients\footnote{We define the true coefficients as the coefficients that would be estimated as the experimental sample size $n \to \infty$. See Supplementary Materials for more information.} associated with the pre-treatment covariates, $\widetilde{\bX}_i$
defined in the weighted least squares regression
  $($equation~\eqref{eq:wls}$)$ and the post-residualized
  weighted least squares regression $($equation~\eqref{eq:prw-wls}$)$,
  respectively. 
\end{theorem}

When we include covariate adjustment to the experimental data, the gains to precision depend on two factors. The first factor, $(a)$, compares the explanatory power of the residualizing model with the linear regression. More specifically, if $\hat{g}(\bX_i)$ is able to explain more variation than the linear combination of $\tilde \bX_i$, then we expect the first term to be positive. The second term, $(b)$, represents the amount of variation in the residualized outcomes that can be explained by the pre-treatment covariates $\tilde \bX_i$. Thus, the magnitude of the precision gain will depend on (1) how much variation the residualizing model can explain in outcomes across the experimental sample, and (2) how much of the variation the covariates $\widetilde{\bX}_i$ are able to explain in the residualized outcomes $\hat e_i$. 

A natural question is why not directly adjust for covariates within the experimental sample instead of using a residualizing model?  One advantage to using the post-residualized weighting over directly adjusting for covariates within the experimental sample arises from the fact that there is typically a larger amount of data available in the population data (i.e. $N \gg n$). While researchers could choose to use a flexible model within the experimental data to perform covariate adjustment, there is a greater restriction with respect to degrees-of-freedom to what type of model can be fit. The availability of large amounts of population data can be leveraged in the residualizing process to better estimate covariate-outcome relationships.  Additionally, by using population data to build and tune the residualizing model, we protect the fidelity of inferences using the experimental data since it is only used for estimation of the PATE. 

When data generating processes are not identical in the experimental and population data, we see concerns similar to the weighted estimator case. When there are large differences between outcomes
measured in the experiment and outcomes measured in the
population data, there is a
risk that we may lose precision from residualizing. However, in the
context of weighted least squares, the additional step of
covariate adjustment can help mitigate potential efficiency losses that arise from a
residualizing model that poorly predicts outcomes measured in the experiment. 

\subsection{Diagnostic}
\label{subsec:wLS_diagnostics}
We now extend the proposed diagnostic from
Section~\ref{subsec:IPW_diagnostics} to the post-residualized weighted
least squares estimator. More formally, we define $R^2_{0, wLS}$ as:  
$$R^2_{0, wLS} = 1 - \frac{\var_{w}(\hat e_i(0) - \widetilde{\bX}_i^\top
  \gamma_{\ast}^{res})}{\var_{w}(Y_i(0)- \widetilde{\bX}_i^\top \gamma_{\ast})},$$
where we now include covariate adjustments from weighted least squares
regression in our diagnostic. $\hat e_i(0) - \widetilde{\bX}_i^\top \gamma_{\ast}^{res}$ are the
residuals that arise from regressing the residualized outcomes under
control on the pre-treatment covariates in the weighted
regression. Similarly, the quantity $Y_i(0) - \widetilde{\bX}_i^\top \gamma_{\ast}$ are the residuals from regressing the outcomes under control on the pre-treatment covariates. In this way, we are directly comparing the
variance of the outcomes, following covariate adjustment, across the control units.

The interpretation of this value is identical to that of the pseudo-$R^2$
value in the weighted estimator case. A negative estimated $R^2_{0, wLS}$ indicates that residualizing may result
in a loss in efficiency. When the estimated $R^2_{0, wLS}$ is
positive, we expect there to be improvements.  

\section{Extension: Using the Predicted Outcomes as a Covariate}
\label{subsec:proxy}

Thus far, we have discussed residualizing, or directly subtracting the predicted outcome values from the outcomes measured in the experimental sample. An alternative approach is to regress the outcomes measured in the experimental sample on the predicted outcomes $\hat Y_i$ from our residualizing model. In particular, consider including the $\hat Y_i$ as a covariate in a weighted linear regression: 
\begin{equation*}
\left( \hat \tau^{cov}_{W}, \hat \beta, \hat \alpha \right) =
\argmin_{\tau, \beta, \alpha} \frac{1}{n} \sum_{i \in \cS}
\hat w_i \left( Y_i - (\tau T_i + \beta \hat Y_i + \alpha )\right)^2.
\end{equation*}
We can extend this approach to also include pre-treatment covariates:
\begin{equation*}
  \left( \hat \tau^{cov}_{wLS}, \hat \beta, \hat \gamma, \hat \alpha
\right) = \argmin_{\tau, \beta, \gamma, \alpha} \frac{1}{n} \sum_{i
  \in \cS} \hat w_i \left( Y_i - (\tau T_i + \beta \hat Y_i +
  \widetilde{\bX}_i^\top \gamma + \alpha )\right)^2.
\end{equation*}
The residualizing methods we discussed in Sections \ref{sec:recycle_ipw} and \ref{sec:recycle_ipw_cov} can be seen as special cases of these methods when
we set $\beta = 1$. 

Residualizing by directly including $\hat Y_i$ as a covariate in the weighted least squares has many advantages. The primary advantage is that this approach allows researchers to flexibly use proxy outcomes measured in the target population. When the outcome of interest is not measured at the population level, or if the outcomes are measured in different ways across the experimental sample and the observed population data, researchers can estimate the residualizing model $g(\bX_i)$ using alternative proxy outcomes $\tilde{Y}_i$ related to the outcome of interest.  However, this can lead to scaling issues that limit the ability of the weighted and weighted least squares methods for post-residualizing to achieve efficiency gains.  We show how including $\hat Y_i$ as a covariate addresses these concerns.

Additionally, as with our post-residualized estimators $\hat \tau_W^{res}$ and $\hat
\tau_{wLS}^{res}$ discussed in Sections \ref{sec:recycle_ipw} and \ref{sec:recycle_ipw_cov}, both $\hat \tau_W^{cov}$ and $\hat
\tau_{wLS}^{cov}$ are consistent estimators of the PATE
(Section~\ref{subsec:cov_consistent}). Finally, including the predicted outcome $\hat Y_i$ as a covariate protects
against efficiency loss, unlike $\hat \tau_W^{res}$ and $\hat
\tau_{wLS}^{res}$ in the previous sections. This is true whether
researchers rely on a proxy outcome $\tilde{Y}_i$, or if they build
the residualizing model on $Y_i$. 

\subsection{Proxy Outcomes in the Population Data}
There are many settings in which researchers may rely on a proxy
outcome $\tilde Y_i$. First, an outcome measure used to estimate the
residualizing model in the population data may differ from the outcome
measure in the experiment. Second, even when the outcome measure used
to estimate the residualizing model in the population data is in
principle the same measure as the outcome of interest in the
experimental data, there can be differences between $\tilde Y_i$ and
$Y_i$ that may arise due to differences in how the outcomes are
measured or operationalized across the experimental sample and the
population, or when the potential outcomes depend on context.  For
example, this might occur if the population is a mix of  both treatment and control conditions with non-random treatment selection.

\paragraph*{Example: Get-Out-the-Vote (GOTV) Experiments}
Consider Get-Out-the-Vote experiments, again, where we are interested
in the causal effect of a randomized GOTV message on voter turnout,
which is measured by administrative voter files in the United States
\citep[e.g.,][]{gerber_green_2000}.  Imagine, however, that we do not have administrative data available on our population, such as for all voters in the United States, but rather, we have a nationally representative survey. For many nationally representative surveys, it is infeasible to link administrative individual-level voting history data due to privacy issues and data constraints; as such, we do not have access to voter turnout. Instead, surveys often ask voters an ``intent-to-vote'' question, which can proxy for actual voter turnout. Our proposed method can use this ``intent-to-vote'' variable to build a residualizing model.

\paragraph*{Example: Education Experiments}
Imagine that researchers are primarily interested in the causal effect of small class sizes not on long term standardized outcomes such as the SAT, but rather a curriculum-based test score specific to a state collected during a given academic year.  In this case, researchers may not have access to this curriculum-based measure in the state-level population data, but may have
access to related standardized testing scores.  These standardized test scores may be used as a proxy to the curriculum-based test score of interest that is measured in the experimental data when constructing the residualizing model. \\ 

When using proxy outcomes to estimate the residualizing model, efficiency gain will be impacted by how similar the proxy outcomes are to the actual outcomes of interest. More formally, consider the following decomposition of the residuals $\hat e_i$: 
\begin{equation} 
\hat e_i \ \ = \underbrace{Y_i - \tilde Y_i}_{\substack{{\text{(a)}}\\{\text{Difference between}}\\{\text{Outcomes in Experiment}}\\{\text{and Proxy Outcome}}}} + \underbrace{\tilde Y_i - \hat Y_i}_{\substack{{\text{(b)}}\\\text{Prediction Error}\\\text{for Proxy Outcome}}},
\label{eqn:error_decomp} 
\end{equation} 
where we define $\tilde Y_i$ as the proxy outcome. Conceptually,
$\tilde Y_i$ represents the proxy outcome, had it been measured for the experimental data. For example, in the GOTV experiment, $\tilde Y_i$ could represent the ``intent-to-vote'' variable, had it been measured for the experimental sample.

Equation~\eqref{eqn:error_decomp} decomposes the residual term into
two components.  The second component (b) is the model prediction
error. This is driven by how well the chosen residualizing model
$g(\bX_i)$ fits proxy outcomes measured in the population data. The
first component (a) is how similar the proxy outcomes measured in the
population data are to the outcome measures used in the experimental
data. If the proxy outcomes differ substantially from the outcomes
measured in the experimental data, while the post-residualized
weighted estimators will still be consistent (see
Theorem~\ref{thm:consistency_ipw}), there may be losses in efficiency
from residualizing, regardless of how much we are able to minimize the
prediction error in the second term (b). 

\subsection{Consistency}\label{subsec:cov_consistent}
Like the previously proposed post-residualized weighted estimators
$\hat \tau_W^{res}$ and $\hat \tau_{wLS}^{res}$, both $\hat
\tau_{W}^{cov}$ and $\hat \tau_{wLS}^{cov}$ will be consistent
estimators of the PATE.   This follows from the fact that $\hat{Y}_i =
\widehat{g}(\bX_i)$ is just a function of pre-treatment covariates
$\bX_i$. In this sense, we can think of $\hat \tau_W^{cov}$ and $\hat
\tau_{wLS}^{cov}$ as extensions of the weighted least squares
estimator, where $\hat Y_i$ is an additional pre-treatment covariate included in the
weighted linear regression. Thus, as shown in Section \ref{sec:recycle_ipw_cov}, both $\hat \tau_{W}^{cov}$ and $\hat \tau_{wLS}^{cov}$ are consistent estimators of PATE.

\subsection{Efficiency Gain and Diagnostics}
There are two advantages to using $\hat Y_i$ as
an additional covariate. First, because $\hat Y_i$ is treated as a
covariate in a weighted regression, the estimated coefficient (i.e.,
$\hat \beta$) can capture any potential scaling differences between
the proxy outcomes and the actual outcomes of interest. For example,
returning to the Get-Out-the-Vote experiments, intent-to vote is often
measured on a Likert scale, while voter turnout is simply a binary
variable of whether the individual voted or not. In such a scenario,
residualizing directly on $\hat Y_i$ can lead to efficiency loss, despite the fact that intent-to-vote is correlated to voter turnout. 

Second, treating $\hat Y_i$ as a covariate protects against precision loss when the proxy outcomes are significantly different from the outcomes of interest. At worst, $\hat Y_i$ is unrelated to $Y_i$, and we expect the coefficient in front of $\hat Y_i$ to be near zero. When this occurs, we expect the variance of the post-residualized weighted estimator when using $\hat Y_i$ as a covariate to be similar to the variance of a conventional estimator that does not include population-level outcome information. More formally: 
 
\begin{corollary} 
The post-residualized weighted estimator using $\hat Y_i$ as a covariate will be at least as asymptotically efficient as the standard weighted estimator: 
$$\asyvar(\hat \tau_W) - \asyvar(\hat \tau_W^{cov}) \geq 0$$
$$\asyvar(\hat \tau_{wLS}) -\asyvar(\hat \tau_{wLS}^{cov}) \geq 0,$$
This result follows from from \citet{Ding2021Paradox}, who shows that the variance of an estimator that accounts for pre-treatment covariates will be asymptotically less than or equal to the variance of an estimator that does not account for pre-treatment covariates. 
\end{corollary}

To account for whether or not the re-scaled predicted outcomes sufficiently explain enough of the variation in the experimental sample, we extend our previously proposed diagnostic measures to the proxy outcome setting. To do so, we propose using sample splitting across the control units in the experimental sample. We regress $\hat Y_i$ on the control outcomes $Y_i$ across one subset of the sample. This allows us to estimate $\hat \beta$. Then using $\hat \beta$, we can estimate residuals, accounting for the scaling factor (i.e., $Y_i - \hat \beta \hat Y_i$), across the held out sample, and calculate the $\hat R^2_0$ and $\hat R^2_{0, wLS}$ diagnostics from before. We finally conduct cross-fitting, i.e., repeating the same procedure by flipping the role of training and test data and then averaging diagnostics from both sample splitting.

\subsection{When to worry about external validity} 
When diagnostic measures indicate that post-residualized weighting is
unsuitable for the data at hand, it is important to understand why. In
particular, Equation~\eqref{eqn:error_decomp} shows that
efficiency loss could occur from (1) the residualizing model's
prediction error, and (2) the difference between the outcomes in the
population and the outcomes measured in the experimental sample.  Low
diagnostic values indicate that post-residualizing methods may not
provide efficiency gains, however it may also be indicative of
contextual differences in the potential outcomes, which affect the
validity of the PATE estimate.

The residualizing model's prediction error, from
equation~\eqref{eqn:error_decomp}-(b), can be estimated through cross
validation using the population-level data. Researchers can hold out random subsets of the population-level data when estimating the residualizing model and calculate the prediction error across the held out sample. If the
cross validated error is large, there will likely be little to no efficiency gains from using post-residualized weighting due to poor prediction, even if the true outcome $Y_i$ is used to estimate $\hat g$.  The difference between the outcomes $Y_i$ and the proxy outcome $\tilde Y_i$, from equation~\eqref{eqn:error_decomp}-(a), can be estimated when the proxy outcome is also measured in the experimental sample.  For example, in the Get-Out-the-Vote experiments, researchers may have voters' intent-to-vote in the experimental sample. Alternatively, in the education experiments, researchers could measure both the curriculum-based test score and the standardized test score in the experimental sample.

In settings where $\tilde{Y}_i$ is not measured in the experimental data, researchers can still use the proposed diagnostic measures to determine if there are concerns about generalizability. For example, if the cross validated prediction error is low, but the diagnostics indicate that post-residualized weighting will not improve efficiency, then this indicates that the residualizing model predicts the population outcomes well, but does not predict outcomes measured in the experiment well. This could be due to two problems. First, if the population outcome is a proxy measure of the outcome measured in the experimental sample, then it could be that the measure used in the population data is not a good proxy for the experimental outcome. Alternatively, if researchers believe that the experimental and population outcomes are measured in the same way, then a low or negative $R^2_0$ measure, in conjunction with low cross validated prediction error, would indicate that the outcome-covariate relationships in the population are considerably different from the outcome-covariate relationships in the experimental sample. In this case, there may be limited external validity of the experiment due to a failure of the consistency of parallel studies assumption, since the potential outcomes may depend on context (see \cite{egami2020elements} for more discussion).

\section{Simulation} \label{sec:sims}
We now run a series of simulations to empirically examine the proposed
post-residualizing method. In total, we consider four different
data-generating scenarios, based on the following model for the
potential outcomes under control: 
\begin{align*}
Y_i(0) \ = \ & \beta_1 X_{1i} + \beta_2 X_{2i} + \gamma_1 X_{1i}^2 + \gamma_2 \sqrt{|X_{2i}|} + \gamma_3 \big(X_{1i} \cdot X_{2i}) \\ 
& \ + \ \beta_S \cdot (1-S_i) \cdot (\alpha + \beta_3 X_{1i} + \gamma_4 X_{1i} \cdot X_{2i}) + \varepsilon_i,
\end{align*}
where $(X_{1i}, X_{2i})$ are observed pre-treatment covariates, and
$S_i \in \{0,1\}$ is a binary indicator variable, taking the value of
one when unit $i$ is in the experimental data, and taking the value of
zero when unit $i$ is in the population data. $\beta_S$ controls for
differences between the experimental sample and population data
outcomes, and the
$\gamma$ terms dictate the nonlinearity of the data generating processes.

We then define the treatment effect model as follows: 
$$\tau_i = \alpha_\tau + X_{\tau,i},$$ 
where $X_{\tau, i}$ is an observed pre-treatment covariate that
governs treatment effect heterogeneity. Therefore, the observed
outcomes take on the following form: $Y_i = Y_i(0) + \tau_i \cdot
T_i$. We provide additional details, including the sampling model and
distributions of observed covariates in Supplementary Materials~\ref{sec:appendix_sims}. 

The first two scenarios test the method when the outcome measures for
both the experimental sample and the population data are drawn from
the same underlying data generating process, to explore a setting
where the outcome is measured identically across the experiment and
target population (i.e., $\beta_S = 0$). The third and fourth
scenarios use different data generating processes to
simulate the experimental sample and population data outcome measures
(i.e., $\beta_S \neq 0$). This mimics the setting in which researchers
use a proxy 
outcome. For both of these settings, we consider a version of the data
generating processes that is linear in the included covariates, and a
second version that contains nonlinearities. Table
\ref{tbl:sim_scenarios} provides a summary of the different scenarios.  

\begin{table}[t]
  \centering
  \caption{\textbf{Summary of Different Simulation Scenarios}} 
  \label{tbl:sim_scenarios} 
  \begin{tabular}{lcc} \toprule
    & Proxy and Experimental Sample Outcomes & DGP Type   \\ \midrule
    Scenario 1 & Identical DGP ($\beta_S = 0$)                              & Linear ($\gamma_{\circ} = 0$) \\
    Scenario 2 & Identical DGP ($\beta_S = 0$)                              & Nonlinear ($\gamma_{\circ} \neq 0$) \\
    Scenario 3 & Different DGP ($\beta_S \neq 0$)               & Linear     ($\gamma_{\circ} = 0$) \\
    Scenario 4 & Different DGP ($\beta_S \neq 0$)               & Nonlinear ($\gamma_{\circ} \neq 0$) \\ \bottomrule
  \end{tabular}
\end{table}

We compare conventional and post-residualized versions of two sets of
estimators in each simulation. We perform post-residualizing in two
different ways: the first directly residualizes the outcomes in the
experimental sample by subtracting the predicted outcomes, and the
second treats the predicted outcomes as a covariate in a weighted
regression. Therefore, we compare a total of six different estimators:
(1) the weighted estimators $\hat \tau_W$, $\hat \tau_W^{res}$, $\hat
\tau_W^{cov}$, and (2) weighted least squares (wLS) $\hat \tau_{wLS}$,
$\tau_{wLS}^{res}$, and $\hat \tau_{wLS}^{cov}$. The
difference-in-means estimator (DiM) is also provided as a baseline with no
weighting adjustment.

The underlying sampling process is governed by a logit model. At each
iteration of the simulation we draw both a biased experimental sample and a
random sample of a larger target population, the population data. The population data is used to estimate the residualizing model
and sampling weights.  We use entropy balancing to estimate the
sampling weights $\hat w_i$ for each simulation. Our residualizing
model is a regression that contains all the pair-wise interactions of
the included covariates. The weighted least squares regression
includes covariates additively without any interactions.\footnote{It is possible in
  practice to include nonlinear transformation of pre-treatment
  covariates in the
  regression adjustment step. However, we have omitted it to
  illustrate the efficiency gains that can be obtained from accounting
  for nonlinearities through the residualizing step. This mimics how,
  in practice, we are able to fit more complex models to more data.}

\begin{figure}[!ht]
\centering
\includegraphics[width=5in]{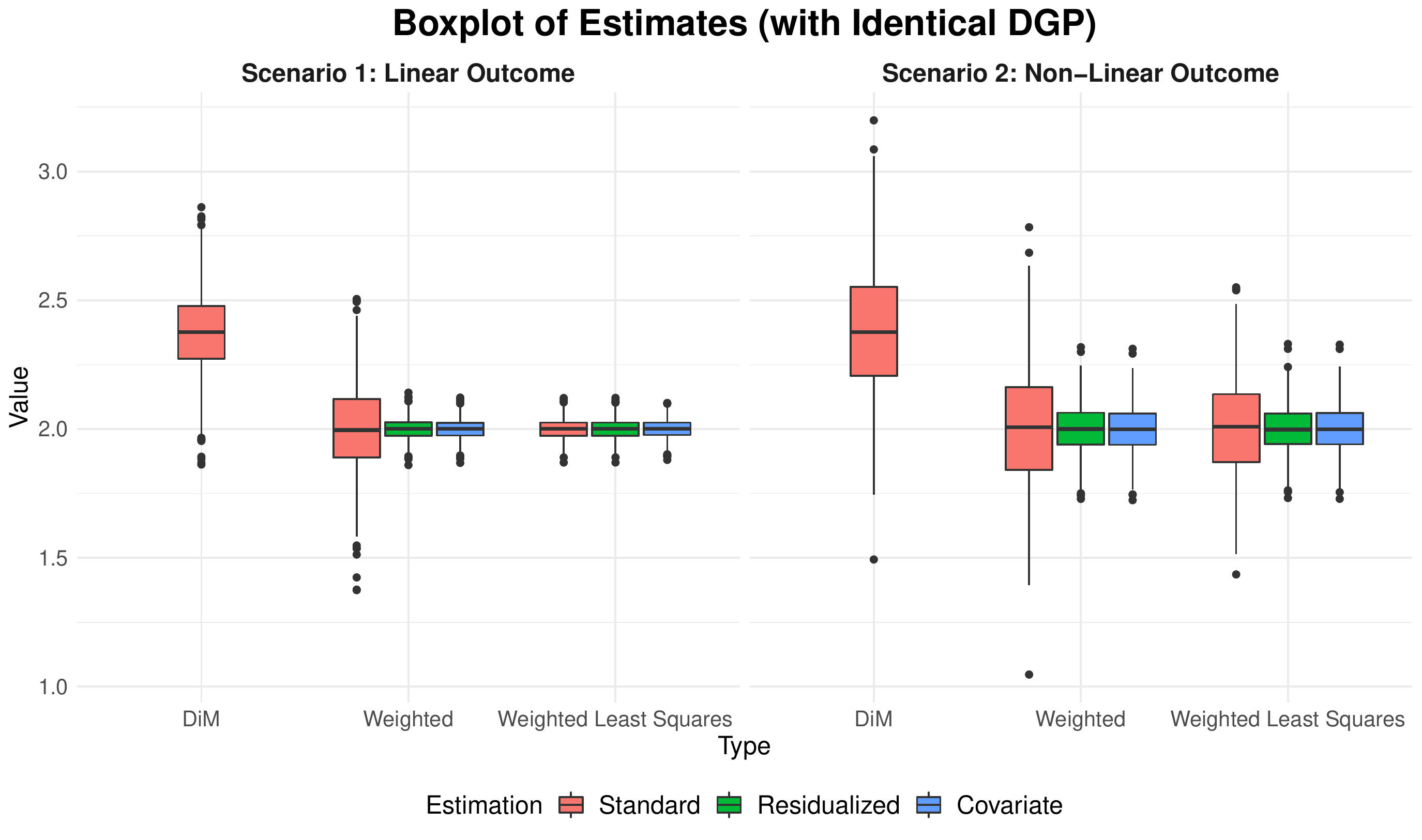}
\caption{Summary of estimates across 1,000 simulations for Scenarios 1 and 2, in which the experimental sample and population outcomes are drawn from the same data generating process. The dotted line represents the super-population PATE.}
\label{fig:boxplot}
\end{figure}

Overall, we find that when the underlying outcome model is complex and contains nonlinear terms, our post-residualizing method exhibits large precision gains compared to conventional methods. When there is no difference between the population-level outcomes and the outcomes in the experimental sample, seen in Figure~\ref{fig:boxplot}, direct residualizing and including $\hat Y_i$ as a covariate performs identically. 

\paragraph*{Scenario 1} When we consider a linear DGP, residualizing
results in substantial precision gains for the weighted
estimator. However, for the weighted least squares estimator,
residualizing does not result in precision gains, because the
covariate adjustment taking place in the weighted regression already
includes the linear terms in the data generating process, and thus,
the residualizing step does not model anything in the outcomes that is
not already accounted for in the wLS regression. 

\paragraph*{Scenario 2} When we include nonlinear terms into the data
generating process, residualizing results in precision gains for all
of the estimators, because the residualizing model is able to account for
some of the nonlinearities that the wLS regression does not account
for. It is worth noting that the estimated residualizing model is not
a correct specification of the underlying outcome model for the
population data.  However, because we have included the pairwise interactions between the covariates, the residualizing model is able to significantly reduce the variance for both estimators, even without accounting for all of the nonlinear terms in the underlying data generating process.

\paragraph*{Scenarios 3 and 4} Next we consider a difference in the
underlying data generating process between the experimental and
population outcomes, presented in Figure~\ref{fig:rmse_beta_s}.  We
operationalize this by including an interaction between treatment, the
sampling indicator, and covariates. The degree to which the two
processes differ is varied across different simulations using a single
parameter, $\beta_S$. When the difference is relatively small
(i.e. small $|\beta_S|$), the two methods used to residualize the
experimental sample outcomes perform identically. This is evident by a
lower RMSE when $|\beta_S| < 2$ for the post-residualized weighted
estimators. When the difference in the DGP are large (i.e., $|\beta_S| > 2$), residualizing by directly subtracting the outcomes from the predicted outcomes results in precision loss, evident by a larger RMSE for the
post-residualized weighted estimator $\hat{\tau}_W^{res}$,
and for the post-residualized weighted least square estimator $\hat{\tau}_{wLS}^{res}$ when the true DGP is
nonlinear. However, treating the predicted outcomes as a covariate in
a weighted linear regression $\hat{\tau}_W^{cov}$ and
$\hat{\tau}_{wLS}^{cov}$ allows for precision gain, even in these
settings. We see that at worst, the covariate-based residualizing
approach performs equivalently to the conventional
estimators. 

It is important to highlight that regardless of the degree of divergence between the population and experimental sample DGP’s, post-residualized weighting is able to maintain nominal coverage. Furthermore, our proposed diagnostic measures adequately capture when we expect to gain or lose precision from residualizing. We provide coverage results and a summary of the diagnostic performance in the Supplementary Materials~\ref{sec:appendix_sims}.

\begin{figure}[t]
\centering
\includegraphics[width=4in]{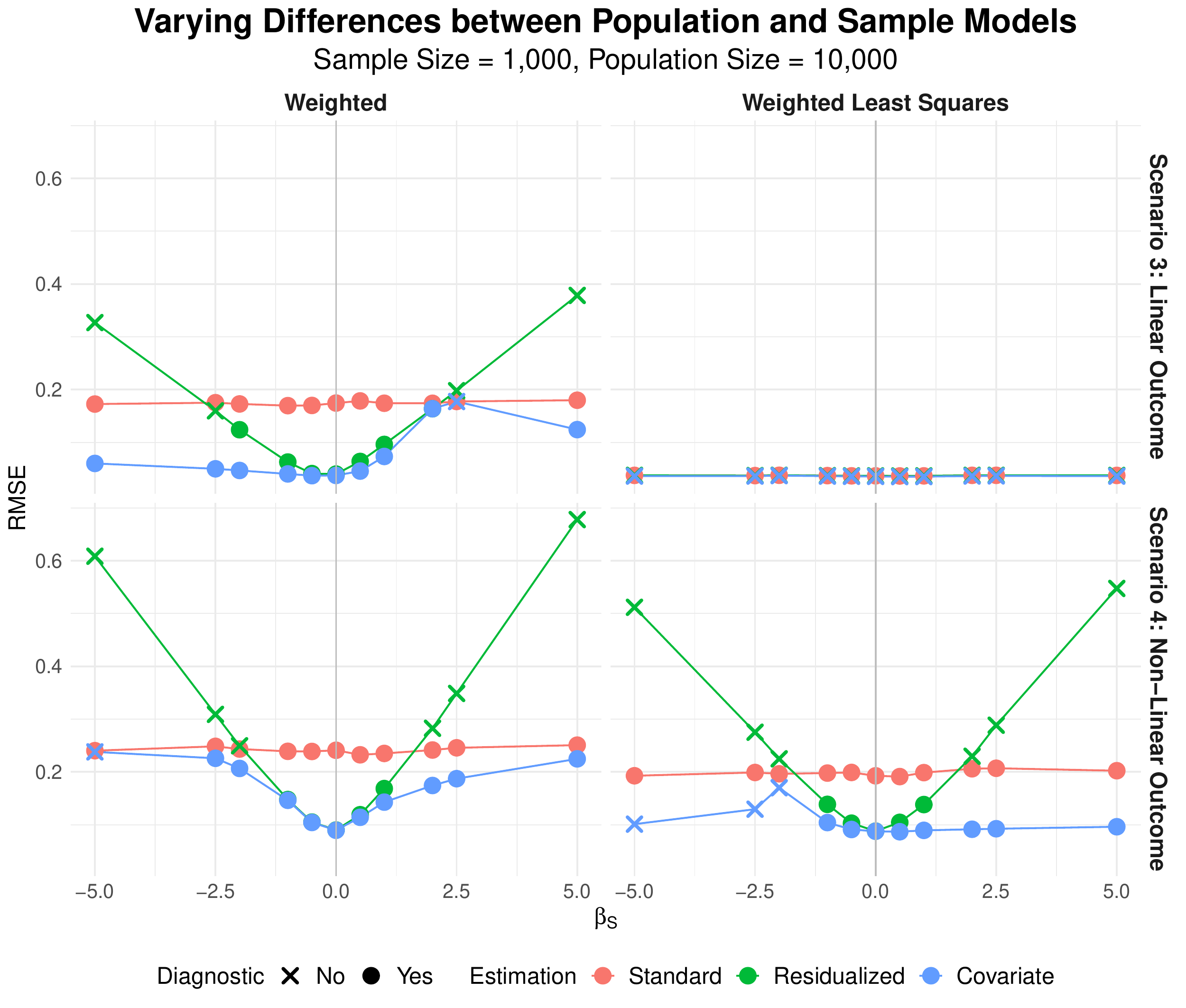}
\caption{Plot of RMSE of the different estimators for Scenarios 3 and
  4, in which the experimental sample and population outcomes are
  drawn from different data generating processes. $\beta_S$ controls for how different the two processes are (i.e., the larger $|\beta_S|$ is, the
  larger the difference is between the two processes).  The standard estimators are presented in red and the residualized estimators in blue and green. We label all the points for which the diagnostic measure estimates a loss in efficiency ($\times$) or gain ($\bullet$) from residualizing more than 50\% of the time in the 1,000 iterations.} 
\label{fig:rmse_beta_s}
\end{figure} 

\section{Empirical Evaluation: Job Training Partnership Act} \label{sec:empirical} 
To evaluate and benchmark how our proposed post-residualizing method
may work in practice, we now turn to an empirical application. Recall that, while the original study evaluated the overall impact of JTPA, our focus is on generalizing the effect of each site individually to the other 15 sites. More specifically, in our leave-one-out analysis for each site, we define the PATE as the average treatment effect among units in the remaining 15 sites.  We then generalize the experimental results from
one site to the population defined by the pooled remaining sites. This
allows us to benchmark our method's performance by comparing our PATE
estimators to the pooled experimental benchmark in the remaining
sites. We evaluate generalizability for two outcomes: employment status
(binary outcome) and total earnings (zero-inflated, continuous
outcome).

\subsection{Post-Residualized Weighting}
\subsubsection{Residualizing model} We include baseline covariates measured in the interview stage of the JTPA study. The covariates include measures of age, previous earnings, marital status, household composition, public assistance history, education and employment history, access to transportation, and ethnicity. More details about the pre-treatment covariates can be found in Supplementary Materials~\ref{app:empirical}. 

We construct our residualizing model using an ensemble method, the
\emph{SuperLearner} \citep{van2007super}. The ensemble model contains
the Random Forest, with varying hyperparameters, and the LASSO, with
hyperparameters chosen using cross validation. This allows us to
capture nonlinearities in the data through the Random Forest, as well
as linear relationships using the LASSO \citep{van2007super}.  We build separate models for the probability of employment and total earnings.  We fit our residualizing model on the control units from the target population. Details can be found in Supplementary Materials~\ref{app:empirical}.

\subsubsection{Estimators} We estimate the PATE using two different
estimators: the weighted estimator and the weighted least squares
estimator (wLS). For each estimator, we consider the conventional
estimators ($\widehat{\tau}_W$ and $\widehat{\tau}_{wLS}$), the post-residualized estimators directly subtracting the
predicted outcomes from the outcomes in the experimental sample
($\widehat{\tau}^{res}_W $ and $\widehat{\tau}^{res}_{wLS}$), and
the post-residualized estimators using the predicted outcomes as a
covariate ($\widehat{\tau}^{cov}_{W}$ and $\widehat{\tau}^{cov}_{wLS}$). Sampling weights are estimated using entropy balancing in
which we match main margins for age, education, previous earnings,
race, and marital status \citep{hainmueller2012entropy}. Our weighted
least squares (wLS) estimators include age, education level, and
marital status as controls. Standard errors are estimated using
heteroskedastic-consistent standard errors (HC2). 

\subsubsection{Diagnostics} For each site, we compute the pseudo-$R^2$ diagnostics. This can be done directly for the post-residualized weighted and weighted least squares estimators. When treating $\hat Y_i$ as a covariate, we use sample splitting to estimate the pseudo-$R^2$ values. Because some of the experimental sites comprise of relative few units (i.e., the experimental site of Montana contains only 38 units total), we perform repeated sample splitting, taking the average of the diagnostic across the repeated splits \citep{jacob2020cross, chernozhukov2018generic}.

\subsection{Results}
\subsubsection{Bias}\label{sec:empirical_bias}
Because the conventional estimators and our proposed approach rely on the same identification assumptions, we first want to verify that the overall bias in the PATE estimation is not affected by the post-residualized weighting step. Across all 16 sites, the point estimates from post-residualized weighting do not change substantially from standard estimation approaches. Even in experimental sites in which it may not be advantageous to perform post-residualized weighting for efficiency gains, point estimates from post-residualized weighting methods are close to those from the conventional weighting estimators. We report the mean absolute error for all 16 sites in Supplementary Materials Table~\ref{tbl:jtpa_coverage}. 

\subsubsection{Diagnostics}\label{sec:empirical_diag}
To evaluate whether the post-residualized weighting estimators provide efficiency gains over conventional approaches, we estimate our diagnostics. Supplementary Materials Table~\ref{tab:diag} 
summarizes the performance of the diagnostic measures across all 16 sites for both earnings and employment. 

On average, we see that the proposed diagnostic measures are able to adequately capture when researchers should expect precision gains from residualizing. The $\hat R^2_0$ diagnostic has a high true positive rate for both directly residualizing and using $\hat Y_i$ as a covariate. As such, when the diagnostic measures indicate that researchers should residualize, residualizing results in precision gains. In the case when we are directly residualizing, the diagnostic measure also has a relatively high true negative rate, which implies that when $\hat R^2_0 < 0$, there is a loss in precision from directly residualizing. In the case of including $\hat Y_i$ as a covariate, there is a greater false negative rate, as the diagnostic tends to be more conservative in this setting. This is especially noticeable when employment is the outcome. Many of the false negatives here correspond to estimated $\hat R^2_0$ values that are negative, but very close to zero.

\subsubsection{Efficiency Gain}
Results on the efficiency gains to post-residualized weighting are summarized in
Table~\ref{tbl:jtpa_estimator_perf}, and graphically displayed in Figure~\ref{fig:int_comparison}. Restricting our attention to the sites for which the $\hat R^2_0$ values are greater than zero, there is a large reduction in variance overall from residualizing. When directly residualizing, for earnings, residualizing results in a 21\% reduction in estimated variance for the weighted estimator and a 12\% reduction for the weighted least squares estimator. For employment, directly residualizing leads to a 10\% reduction in estimated variance for the weighted estimator and a 5\% reduction for the weighted least squares estimator. 

\begin{table}[t]
\centering 
\textbf{Summary of Standard Errors across Experimental Sites Subset by Diagnostic} \\   \vspace{2mm} 
\resizebox{\textwidth}{!}{%
\begin{tabular}{l|ccc|ccc} \toprule
& \multicolumn{3}{c|}{\underline{Earnings}} & \multicolumn{3}{c}{\underline{Employment}} \\ 
& \begin{tabular}{@{}c@{}}Number\\of Sites\end{tabular} & Conventional& \begin{tabular}{@{}c@{}}Post-Resid. \\ Weighting\end{tabular}
& \begin{tabular}{@{}c@{}}Number\\of Sites\end{tabular} & Conventional& \begin{tabular}{@{}c@{}}Post-Resid. \\ Weighting\end{tabular}\\\midrule
\hspace{3mm} \underline{Weighted} & & & &\\
$\qquad$ Direct Residualizing & 10 & 2.42 & 2.13 & 11  & 8.33 & 7.81  \\ 
$\qquad$ $\hat Y_i$ as Covariate & 7  & 2.17 & 1.86 & 1  & 5.58 & 5.01\\ [0.5em]
\hspace{3mm} \underline{Weighted Least Squares}  & & & &\\ 
$\qquad$ Direct Residualizing & 12 & 2.71 & 2.56 & 11  & 7.88 & 7.64\\ 
$\qquad$ $\hat Y_i$ as Covariate & 7 & 1.87 & 1.71 & 1 & 5.56 & 5.45\\ \midrule
\end{tabular}
}
\caption{Summary of gains to post-residualized weighting.  Columns 1 and 4 give the number of sites for which the diagnostic measure indicates gains to post-residualized weighting.  The average standard error among selected sites are presented for the conventional estimators (columns 2 and 5) and post-residualized estimators (columns 3 and 6).}
\label{tbl:jtpa_estimator_perf}
\end{table}

When using $\hat Y_i$ as a covariate, we see that including the predicted outcomes as a covariate results in a 25\% reduction in variance for the weighted estimator and 16\% reduction for weighted least squares when earnings is the outcome. For employment, adjusting for the predicted outcomes results in a 9\% reduction in variance for the weighted estimator, and a 4\% reduction for the weighted least squares. 

\begin{figure}[!t]
\centering
\includegraphics[width=5in]{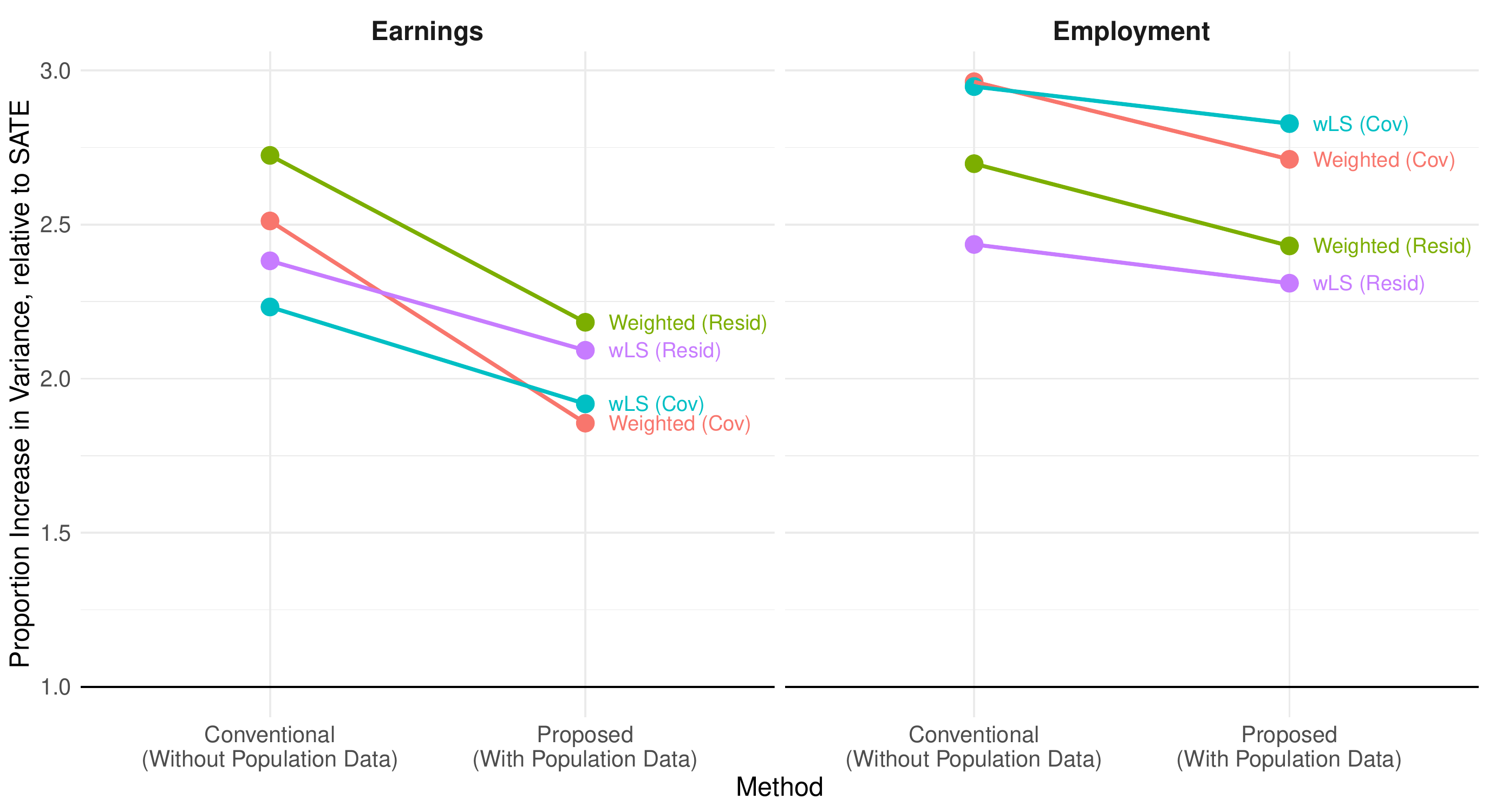}
\caption{Reduction in Variance from using post-residualized weighting. This is
  restricted to the set of sites for which the diagnostic measure
  indicated that we should expect precision gains from
  residualizing. We calculate the variance of the estimators, relative to the variance of the difference-in-means (DiM) estimator. We can interpret the $y$-axis as the amount of variance inflation that is incurred from generalization, and see that using the proposed method of incorporating population data can allow us to offset some of the precision loss incurred from re-weighting.}
\label{fig:int_comparison}
\end{figure} 

There are several takeaways to highlight. First, we see that directly residualizing the outcomes can result in significant precision gain. In particular, the reduction in variance in the post-residualized weighted least squares demonstrates the advantage residualizing has over just using regression adjustment. Second, the larger reduction in variance from using $\hat Y_i$ as a covariate underscores the value of being able to capture the scaled relationships between the outcomes in the population data and in the experimental sample. 

Figure~\ref{fig:int_comparison} shows the relative variance of the PATE estimators to the unweighted SATE.  It is well known that PATE estimators typically have higher variance than the SATE
\citep{miratrix_2018}, however we see that with the post-residualized method, some of the precision loss incurred from the weighted PATE estimators can be offset. Table \ref{tbl:jtpa_estimator_perf} provides a summary of the standard errors of the PATE estimators, relative to the difference-in-means estimators. 

\section{Conclusion} 

Ever since researchers raised concerns over external validity of
experiments \citep{campbell1963experimental}, researchers have also
worked on how to estimate population causal effects 
with weighting estimators \citep{cole2010generalizing,
  buchanan2018generalizing, hartman2015sample,
  dahabreh2019generalizing}.  These estimators, while unbiased under
Assumptions \ref{assum-exp} -- \ref{assum-pos}, typically have high
variance, especially if some sampling weights are extreme \citep{miratrix_2018}, making it difficult for policymakers and practitioners to draw conclusions about the impact of treatment in the target population.

In this paper we introduce post-residualized weighting, which
solves an important problem for practitioners by improving precision in estimation of population treatment effects.  To do this, we leverage outcome data measured in the target population, valuable information not incorporated by current methods.  Our proposed
method first builds a flexible model using population outcome and
covariate data, which is then used to residualize the experimental
outcome data.  We show that post-residualized weighting estimators,
which rely on residualized outcomes, are consistent for the PATE under
the same identifying assumptions as current methods. However, by
utilizing residualized outcomes, the post-residualized weighting
estimators can obtain large precision gains over conventional
approaches.  We propose three classes of post-residualized weighting
estimators: a weighting estimator using the residualized experimental
outcomes; a weighted least squares estimator based on the residualized
experimental outcomes; and
an extension of weighted least squares in which the predicted values
of the residualizing model are included as a covariate. 

Our proposed framework has many advantages.  As discussed in
Section~\ref{sec:post_res_weighted}, the residualizing model,
$g(\mathbf{X}_i)$, is an ``algorithmic model,'' which merely needs to
adequately predict the outcomes measured in the experiment, but does
not need to be correctly specified.  This allows researchers a great
deal of flexibility in constructing it. In Section \ref{subsec:proxy}
we discuss how researchers can leverage proxy outcomes that are
correlated with, but different from, the outcome measured in the
experimental setting.  Finally, we provide diagnostic measures, based
on the outcomes measured among experimental controls, that allow
researchers to determine whether post-residualized weighting will
likely improve precision in estimating the PATE.

We evaluate our three post-residualized estimators through simulation studies and an empirical application.  Our simulations show significant precision gains from post-residualized weighting, and confirm the performance of the diagnostic measure to differentiate when researchers should expect precision gains from post-residualized weighting.  We also find that including the predicted outcomes as a covariate ensures that post-residualized weighting does not hurt precision.

In our re-evaluation of the impact of the Job Training Partnership Act (JTPA), we use the multi-site nature of the experiment to benchmark the performance of our estimators relative to common methods using a within study comparison approach.  We evaluate two outcomes, employment and earnings.  We find that the post-residualized methods result in a 5-25\% average reduction in variance, and that confidence intervals maintain nominal coverage.  In particular, we achieve the most significant gains from including the predicted outcomes as a covariate, underscoring the value of this method when scaling issues may be present in the relationship between the outcomes in the population data and in the experimental sample.  Finally, our diagnostic measures accurately capture when the post-residualized estimators result in precision gains in estimation of the PATE.

\printbibliography

\newpage
\appendix
\newrefsection
\setcounter{page}{1}
\begin{center}
    \singlespacing
    \LARGE
    \textbf{Supplementary Materials:} \\ Leveraging Population Outcomes to Improve \\ the Generalization of Experimental Results
\end{center}

\setcounter{table}{0}
\renewcommand{\thetable}{A\arabic{table}}

\setcounter{figure}{0}
\renewcommand{\thefigure}{A\arabic{figure}}

\setcounter{equation}{0}
\renewcommand{\theequation}{A\arabic{equation}}

\section{Proofs and Derivations} \label{app:proofs} 

\subsection{Derivation of Variance Terms} 

Consider a countably infinite population of $(\bX_i, Y_i(t)) \sim F$, where
$t\in\{0,1\}$, with density $dF(\bX_i, Y_i(t))$. This is our target
population. We define the sampling distribution for the experimental
data to be $(\bX_i, Y_i(t)) \sim \tilde{F}$ with density
$d\tilde{F}(\bX_i,Y_i(t))$. Because we consider settings where the
selection into the experiment from the target population is biased, $F
\neq \tilde{F}.$ Let $\mathcal{S}$ be the set of all indices for all units sampled in the experimental sample. As we can consider the treatment and control groups to be independent samples from an infinite population, we will focus below on one potential outcome $Y_i(t).$

We defined a relative density in equation~\eqref{eq:r-d}
as follows.
\begin{equation*}
  \pi(\bX_i)  =  \frac{d\tilde{F}(\bX_i)}{dF(\bX_i)}.
\end{equation*}
over the support of $F$, where $dF(\bX_i) > 0$. The $\pi(\bX_i)$ is our infinite analog to the sampling propensity score.  It scales our distribution. We further assume that $\pi(\bX_i) >  0$ (this is an overlap assumption, saying our realized sampling distribution is not missing parts of the underlying distribution). $\pi(\bX_i)$ captures the relative density of our realized distribution to the real population. Smaller $\pi(\bX_i)$ correspond to areas where there is a lot more in the target population than in our sample. Larger $\pi(\bX_i)$ are where we are over-sampling. \\

\noindent We assume known weights for any unit, dependent on $\bX_i$, with $w_i = \kappa / \pi( \bX_i )$ (the $\kappa$ is a fixed constant allowing our weights to be normalized on some arbitrary scale).

For the remainder of the Supplementary Materials, the distribution over which a
quantity is computed will be denoted by subscript. For example, the
expectation over the realized sampling distribution will be written as
$\E_{\tilde F}(\cdot)$, while the expectation over the target
population will be written as $\E_F(\cdot)$.  

\begin{lemma}[Variance of a H\'{a}jek estimator] \label{lem:var_hajeck} 
Define $\hat \mu$ as a H\'{a}jek estimator:
$$\hat \mu_t = \frac{\sum_{i \in \cS} w_i Y_i(t)}{\sum_{i \in \cS} w_i},$$ 
where consistent with before, $w_i = \kappa/\pi(\bX_i)$, and $(\bX_i,
Y_i(t)) \sim \tilde F$. The approximate asymptotic variance of a
H\'{a}jek estimator is: 
$$\asyvar_{\tilde F}(\hat \mu_t) \approx \int \frac{1}{\pi(\bX_i)^2} (Y_i(t) - \mu_t)^2 d \tilde F(\bX_i, Y_i(t)),$$
where the asymptotic variance is being taken with respect to the realized sampling distribution, and $\mu_t = \E_{F}(Y_i(t))$ (i.e., the expected value of $Y_i(t)$ over the target population). 
\end{lemma} 
\begin{proof} 
To begin, we write the H\'{a}jek estimator as a ratio estimator of the following form: 
\begin{align*} 
\hat{\mu}_t &=  \frac{\sum_{i \in \cS} w_i Y_i(t)}{\sum_{i \in \cS} w_i}\\ 
            &= \frac{\frac{1}{n} \sum_{i \in \cS} w_i Y_i(t) }{ \frac{1}{n} \sum_{i \in \cS} w_i }
\end{align*}
where we define $n$ to be the sample size, i.e., $n = |\cS|$.

We then define $\hat{A} = \frac{1}{n} \sum_{i \in \cS} w_i
Y_i(t)$ and $\hat{B} = \frac{1}{n} \sum_{i \in \cS} w_i$ for
notational simplicity. If we define $A = \E_{\tilde F}
(\hat{A}),$ $A = \kappa \mu_t.$ Similarly,  if we define $B = \E_{\tilde F}
(\hat{B}),$ $B = \kappa.$

\noindent To derive the variance expression, we will use the delta
method below for a ratio, i.e., a function $h(a,b) = a/b$.  For this
ratio, we have 
$$ \frac{d}{da} h(a,b) = \frac{1}{b} \qquad \frac{d}{db} h( a, b ) = -\frac{a}{b^2}.$$ 
Therefore, using the Delta Method for a ratio,
\begin{align*}
  \hat{\mu}_t &= \frac{\frac{1}{n} \sum_{i \in \cS} w_i Y_i(t) }{ \frac{1}{n} \sum_{i \in \cS} w_i }\\
  & = \frac{\hat{A}}{\hat{B}}\\
              & \approx \frac{A}{B} + \frac{1}{B}(\hat{A}-A) - \frac{A}{B^2}(\hat{B}-B) \\
              & = \frac{A}{B} -  \frac{A}{B} +  \frac{A}{B} +   \frac{1}{B}\hat{A} - \frac{A}{B^2}\hat{B} \\
 	& = \mu_t +  \frac{1}{\kappa} \frac{1}{n} \sum_{i \in \cS} w_i Y_i(t)  - \frac{\mu_t}{\kappa} \frac{1}{n} \sum_{i \in \cS} w_i \\
 	& = \mu_t + \frac{1}{n \kappa}  \sum_{i \in \cS} w_i ( Y_i(t) - \mu_t )
\end{align*}
where the first and second equalities follow from the definition of
$\hat{\mu}_t$ and $(\hat{A}, \hat{B}),$ the third from the delta
method, the fourth from simple algebra, the fifth from the definition of $(A, B),$ and the sixth from re-arrangement of the terms. 

\noindent Finally,
\begin{align}
  \var_{\tilde F}\left( \hat{\mu}_t \right) 
  & =  \var_{\tilde F}\left( \hat{\mu}_t - \mu_t \right) \\
  &\approx \frac{1}{n^2 \kappa^2}\cdot \var_{\tilde F}\left( \sum_{i \in \cS} w_i (Y_i(t) - \mu_t) \right) \nonumber \\
&= \frac{1}{n^2 \kappa^2} n \int \frac{\kappa^2}{\pi(\bX_i)^2} (Y_i(t) - \mu_t)^2 d\tilde F(\bX_i,Y_i(t) ) \nonumber \\
	&= \frac{1}{n} \int \frac{1}{\pi(\bX_i)^2} (Y_i(t) - \mu_t)^2 d\tilde F(\bX_i, Y_i(t))\label{eqn:var_finite} 
\end{align}
As such, $\asyvar_{\tilde F}\left( \hat{\mu}_t \right) = \lim_{n \rightarrow \infty}\var(\sqrt{n} \hat\mu_t) = \int \frac{1}{\pi(\bX_i)^2} (Y_i(t) - \mu_t)^2 d\tilde F(\bX_i, Y_i(t))$.
\end{proof}

\begin{lemma}[Weighted Variance] 
Define the weighted variance and the weighted covariance as: 
$$\var_w(A_i) = \int \frac{1}{\pi(\bX_i)^2} (A_i-\bar A)^2 d\tilde F(\bX_i,A_i)$$
$$\cov_w(A_i, B_i) = \int \frac{1}{\pi(\bX_i)^2} (A_i-\bar A)(B_i - \bar B) d\tilde
F(\bX_i,A_i, B_i)$$
Under this definition, common variance and covariance properties apply: 
$$\var_w(A_i + B_i) = \var_w(A_i) + \var_w(B_i) + 2 \cov_w(A_i, B_i)$$
$$\cov_w(A_i + B_i, C_i) = \cov_w(A_i, C_i) + \cov_w(B_i, C_i)$$ 
\end{lemma}

\begin{proof} 
\begin{align*} 
\var_w(A_i + B_i) =& \int \frac{1}{\pi(\bX_i)^2} \left(A_i + B_i - (\bar A + \bar B) \right)^2 d\tilde F(\bX_i, A_i, B_i) \\
=& \int \frac{1}{\pi(\bX_i)^2} \left((A_i - \bar A)^2 + (B_i - \bar B)^2 + 2(A_i - \bar A)(B_i - \bar B) \right) d\tilde F(\bX_i, A_i, B_i) \\
=& \int \frac{1}{\pi(\bX_i)^2} (A_i - \bar A)^2 d \tilde F(\bX_i, A_i, B_i) + \int \frac{1}{\pi(\bX_i)^2} (B_i - \bar B)^2 d \tilde F(\bX_i, A_i, B_i) + \\
&2 \int \frac{1}{\pi(\bX_i)^2} (A_i - \bar A)(B_i - \bar B) d\tilde F(\bX_i, A_i, B_i)  \\
=&\int \frac{1}{\pi(\bX_i)^2} (A_i - \bar A)^2 d \tilde F(\bX_i, A_i) + \int \frac{1}{\pi(\bX_i)^2} (B_i - \bar B)^2 d \tilde F(\bX_i, B_i) + \\
&2 \int \frac{1}{\pi(\bX_i)^2} (A_i - \bar A)(B_i - \bar B) d \tilde F(\bX_i, A_i, B_i) \\
=& \var_w(A_i) + \var_w(B_i) + 2 \cov_w(A_i, B_i)
\end{align*} 
\begin{align*} 
\cov_w&(A_i + B_i, C_i)\\ 
=& \int \frac{1}{\pi(\bX_i)^2} \left(A_i + B_i - (\bar A + \bar B) \right)\left(C_i  - \bar C \right) d \tilde F(\bX_i, A_i, B_i, C_i) \\
=& \int \frac{1}{\pi(\bX_i)^2} \left( (A_i - \bar A)(B_i- \bar B) \right)\left(C_i- \bar C \right) d \tilde F(\bX_i, A_i, B_i, C_i) \\
=& \int \frac{1}{\pi(\bX_i)^2} \left( (A_i - \bar A)(C_i - \bar C) + (B_i - \bar B)(C_i - \bar C) \right) d \tilde F(\bX_i, A_i, B_i, C_i) \\
=& \int \frac{1}{\pi(\bX_i)^2} (A_i - \bar A)(C_i - \bar C) d\tilde F(\bX_i, A_i, B_i, C_i) + \int \frac{1}{\pi(\bX_i)^2} (B_i - \bar B)(C_i - \bar C) d\tilde F(\bX_i, A_i, B_i, C_i) \\
=& \int \frac{1}{\pi(\bX_i)^2} (A_i - \bar A)(C_i - \bar C) d \tilde F(\bX_i, A_i, C_i) + \int \frac{1}{\pi(\bX_i)^2} (B_i - \bar B)(C_i - \bar C) d \tilde F(\bX_i, B_i, C_i) \\
=& \cov_w(A_i, C_i) + \cov_w(B_i, C_i) 
\end{align*} 
\end{proof} 

\begin{corollary}[Asymptotic Variance of a Weighted Estimator] \mbox{}\\ 
The asymptotic variance of a Hájek-style weighted estimator is:
\begin{align*} 
\asyvar_{\tilde F}(\hat \tau_W) & \ = \ \asyvar_{\tilde F}(\hat \mu_1) + \asyvar_{\tilde F}(\hat \mu_0) \\
& \ \approx \ \frac{1}{p} \int \frac{1}{\pi(\bX_i)^2} (Y_i(1)-\mu_1)^2 d\tilde F(\bX_i,Y_i(1))+\frac{1}{1-p} 
\int \frac{1}{\pi(\bX_i)^2} (Y_i(0)-\mu_0)^2 d\tilde{F}(\bX_i,Y_i(0)) \\
& \ = \ \frac{1}{p} \var_w(Y_i(1)) + \frac{1}{1-p} \var_w(Y_i(0)),
\end{align*} 
where $\var_w(\cdot)$ is defined in equation~\eqref{eq:w-var}. 
$p$ is the probability of treatment assignment, i.e., $p =
\Pr_{\tilde{F}}(T_i = 1)$. $\mu_1 = \E_{F}(Y_i(1))$ and $\mu_0 = \E_{F}(Y_i(0))$.
\label{cor:weights} 
\end{corollary} 
\begin{proof} 
Because we are sampling from an infinite super-population, the treatment and control groups can be treated as two separate samples from the infinite super-population. We directly apply Lemma \ref{lem:var_hajeck} to arrive at the final result. \\
\label{cor:var_w} 
\end{proof} 

\begin{corollary}[Asymptotic Variance of Weighted Least Squares Estimator] \mbox{}\\
The asymptotic variance of a weighted least squares estimator is: 
$$\asyvar(\hat \tau_{wLS}) = \frac{1}{p} \var_w(Y_i(1) - \tilde \bX_i^\top \gamma_\ast ) + \frac{1}{1-p} \var_w(Y_i(0)- \tilde \bX_i^\top \gamma_\ast),$$
where 
$\gamma_\ast$ is the vector of \textit{true} coefficients associated with
the pretreatment covariates $\tilde \bX_i$ defined as:
\begin{equation}
  (\tau_{wLS}, \alpha_\ast, \gamma_\ast) =
  \argmin_{\tau, \alpha, \gamma}
  \E_{\tilde F} \left\{\hat{w}_i \left(Y_i - (\tau T_i +
  \alpha + \widetilde{\bX}^\top_i \gamma)\right)^2\right\}
\end{equation}
\label{cor:wLS_var}
\end{corollary} 
\begin{proof} 

To begin, analogous with \cite{lin2013} (Lemma 6), the weighted least squares estimator can be written as:
\begin{equation} 
\hat \tau_{wLS} = \frac{1}{\sum_{i \in \cS} w_i T_i} \sum_{i \in \cS} w_i T_i (Y_i - \tilde \bX_i^\top \hat \gamma) - \frac{1}{\sum_{i \in \cS} w_i (1-T_i)} \sum_{i \in \cS} w_i (1-T_i) (Y_i - \tilde \bX_i^\top \hat \gamma)
\label{eqn:wLS} 
\end{equation} 
Akin with \cite{Ding2021Paradox}, we define $\delta_X$ as: 
$$\delta_X = \frac{1}{\sum_{i \in \cS} w_i T_i} \sum_{i \in \cS} w_i T_i \tilde \bX_i^\top - \frac{1}{\sum_{i \in \cS} w_i (1-T_i)} \sum_{i \in \cS} w_i (1-T_i) \tilde \bX_i^\top $$
$\delta_X$ represents any residual imbalance between the treatment and control groups in the weighted pre-treatment covariates. 
We can re-write Equation~\eqref{eqn:wLS} as: 
\begin{align*} 
\hat \tau_{wLS} =& \frac{1}{\sum_{i \in \cS} w_i T_i} \sum_{i \in \cS} w_i T_i (Y_i - \tilde \bX_i^\top \hat \gamma) - \frac{1}{\sum_{i \in \cS} w_i (1-T_i)} \sum_{i \in \cS} w_i (1-T_i) (Y_i - \tilde \bX_i^\top \hat \gamma) \\
=&\frac{1}{\sum_{i \in \cS} w_i T_i} \sum_{i \in \cS} w_i T_i (Y_i(1) - \tilde \bX_i^\top \hat \gamma) - \frac{1}{\sum_{i \in \cS} w_i (1-T_i)} \sum_{i \in \cS} w_i (1-T_i) (Y_i(0) - \tilde \bX_i^\top \hat \gamma) \\
=& \frac{1}{\sum_{i \in \cS} w_i T_i} \sum_{i \in \cS} w_i T_i (Y_i(1) - \tilde \bX_i^\top \gamma_\ast + \tilde \bX_i^\top  \gamma_\ast - \tilde \bX_i^\top \hat \gamma) - \\
&\frac{1}{\sum_{i \in \cS} w_i (1-T_i)} \sum_{i \in \cS} w_i (1-T_i) (Y_i(0) - \tilde \bX_i^\top \gamma_\ast + \tilde \bX_i^\top  \gamma_\ast - \tilde \bX_i^\top \hat \gamma) \\
=& \frac{1}{\sum_{i \in \cS} w_i T_i} \sum_{i \in \cS} \left( w_i T_i (Y_i(1) - \tilde \bX_i^\top \gamma_\ast) + w_i T_i \tilde \bX_i^\top(\gamma_\ast - \hat \gamma)\right) - \\
&\frac{1}{\sum_{i \in \cS} w_i (1-T_i)} \sum_{i \in \cS} \left(w_i (1-T_i) (Y_i(0) - \tilde \bX_i^\top \gamma_\ast)  +  w_i(1-T_i) \tilde \bX_i^\top  (\gamma_\ast - \hat \gamma) \right) \\
=& \underbrace{\frac{1}{\sum_{i \in \cS} w_i T_i} \sum_{i \in \cS} w_i T_i (Y_i(1) - \tilde \bX_i^\top \gamma_\ast) - \frac{1}{\sum_{i \in \cS} w_i (1-T_i)} \sum_{i \in \cS} w_i (1-T_i) (Y_i(0) - \tilde \bX_i^\top \gamma_\ast)}_{:=\hat \tau^*_{wLS}} + \\
&\underbrace{\frac{1}{\sum_{i \in \cS} w_i T_i} \sum_{i \in \cS} w_i T_i \tilde \bX_i^\top (\gamma_\ast - \hat \gamma) - 
\frac{1}{\sum_{i \in \cS} w_i (1-T_i)} \sum_{i \in \cS} w_i (1-T_i) \tilde \bX_i^\top (\gamma_\ast - \hat \gamma)}_{= \delta_X(\gamma_\ast - \hat \gamma)} \\
=& \hat \tau_{wLS}^{*}+ \delta_X(\gamma_\ast - \hat \gamma),
\end{align*} 
where $\hat \tau_{wLS}^*$ represents the potential outcomes, adjusted
for the pre-treatment covariates using the \textit{true} coefficients
$\gamma_\ast$.

Under standard regularity conditions for least squares, $\gamma_* - \hat \gamma = o_p(1)$ \citep{white1982maximum}. Furthermore, $\sqrt{n} \delta_X = O_p(1)$:
\begin{align*} 
\lim_{n \to \infty} \var_{\tilde F}(\delta_X) &= \lim_{n \to \infty} \left( \frac{1}{n_1} \var_w(\tilde \bX_i) +\frac{1}{n_0} \var_w(\tilde \bX_i)\right) \\
&= \lim_{n \to \infty} \frac{1}{n} \cdot \left( \frac{1}{p} + \frac{1}{1-p} \right) \var_w(\tilde \bX_i) \\
&= \lim_{n \to \infty} \frac{1}{n} \cdot \frac{1}{p(1-p)} \var_w(\tilde \bX_i)
\end{align*} 
Assuming $\var_w(\tilde \bX_i)$ is finite, $\delta_X = O_p(\sqrt{n}^{-1}) \implies \sqrt{n} \delta_X = O_p(1)$. 

Therefore, as $n \to \infty$: 
\begin{align*} 
\sqrt{n}(\hat \tau_{wLS} - \tau) &= \sqrt{n}(\hat \tau^*_{wLS} - \tau) + \underbrace{\sqrt{n}\delta_X(\gamma_\ast - \hat \gamma)}_{\cip 0} \\
&\cid N(0, \var(\hat \tau_{wLS}^*)),
\end{align*} 
where $\var_{\tilde F}(\hat \tau_{wLS}^*) \approx \frac{1}{p} \var_w(Y_i(1) - \tilde \bX_i^\top \gamma_\ast) + \frac{1}{1-p} \var_w(Y_i(0) - \tilde \bX_i^\top \gamma_\ast)$ (this result follows from applying Lemma 1 on the adjusted potential outcomes). 

\end{proof}

\subsection{Proof of Theorem~\ref{thm:consistency_ipw}} 
\mbox{}\\

\noindent\fbox{%
\vspace{2mm}
\parbox{\textwidth}{%
\vspace{2mm}
Suppose Assumption~\ref{assum-ind} 
holds with $\bX_i$, then the
post-residualized weighted estimator, using any predictive model
$g(\bX_i)$ built on the population data, is a consistent estimator for
the PATE: 
$$\hat \tau_{W}^{res} \cip \tau,$$
provided sampling weights $\hat w_i$ are estimated consistently. 
}
} 
\begin{proof} 
We will begin by the proof by showing that $\hat \tau_W^{res}$ can be written as the difference between $\hat \tau_W$, and a weighted estimator computed over the fitted values $\hat Y_i$, which we will define as $\hat \tau_{\hat Y}$. Following the generalization literature, we treat the weights as known, as well as the observed sampled population: 
\begin{align*} 
\hat \tau_{W}^{res} =&  \frac{\sum_{i \in \cS} w_i T_i \cdot \hat e_i}{\sum_{i \in \cS} w_i T_i} - \frac{\sum_{i \in \cS} w_i (1-T_i) \cdot \hat e_i}{\sum_{i \in \cS} w_i (1-T_i)}  \\
=& \frac{\sum_{i \in \cS} w_i T_i \cdot (Y_i- \hat Y_i)}{\sum_{i \in \cS} w_i T_i}
- 
\frac{\sum_{i \in \cS} w_i (1-T_i) \cdot (Y_i - \hat Y_i)}{\sum_{i \in \cS} w_i (1-T_i)} 
\\
=&\underbrace{\frac{\sum_{i \in \cS} w_i T_i \cdot Y_i}{\sum_{i \in \cS} w_i T_i}
- 
\frac{\sum_{i \in \cS} w_i (1-T_i) \cdot Y_i }{\sum_{i \in \cS} w_i (1-T_i)}}_{= \hat \tau_W} - \\
&\underbrace{\left( \frac{\sum_{i \in \cS} w_i T_i \cdot \hat Y_i}{\sum_{i \in \cS} w_i T_i}
- 
\frac{\sum_{i \in \cS} w_i (1-T_i) \cdot \hat Y_i }{\sum_{i \in \cS} w_i (1-T_i)} \right)}_{= \hat \tau_{\hat Y}} \\
=& \hat \tau_{W} - \hat \tau_{\hat Y}
\end{align*} 
We will begin by showing that $\hat \tau_W \cip \tau$. 
To begin: 
\begin{align*} 
\hat \tau_W &= \frac{\sum_{i \in \cS} w_i T_i \cdot Y_i}{\sum_{i \in \cS} w_i T_i}
- 
\frac{\sum_{i \in \cS} w_i (1-T_i) \cdot Y_i }{\sum_{i \in \cS} w_i (1-T_i)}
\intertext{By Law of Large Numbers and the Continuous Mapping Theorem:} 
\hat \tau_W &\cip \underbrace{\frac{\E_{\tilde F} (w_i T_i Y_i)}{\E_{\tilde F} (w_iT_i)}}_{(1)} - \underbrace{\frac{\E_{\tilde F} (w_i(1-T_i) Y_i)}{\E_{\tilde F} (w_i(1-T_i))}}_{(2)}
\end{align*} 
We will now show that the first term (i.e., $(1)$) is equal to $\E_{F}(Y_i(1))$. We first evaluate the expectation in the denominator. 
\begin{align*} 
\E_{\tilde F} (w_i T_i) &= \frac{n_1}{n} \E_{\tilde F} (w_i) \\
&= \frac{n_1}{n} \E_{\tilde F}  \left( \frac{\kappa}{\pi(\bX_i)} \right) \\
&= \frac{n_1}{n} \cdot \kappa \int \frac{1}{\pi(\bX_i)} d \tilde F(\bX_i) \\
&= \frac{n_1}{n} \cdot \kappa \underbrace{\int \frac{1}{\pi(\bX_i)} \pi(\bX_i) d F(\bX_i)}_{=1} \\
&= \frac{n_1}{n} \cdot \kappa 
\end{align*} 
For the numerator: 
\begin{align*} 
\E_{\tilde F} (w_i T_i Y_i) &= \E_{\tilde F} (w_i T_i Y_i(1) )\\
&= \frac{n_1}{n} \E_{\tilde F} (w_i Y_i(1))\\
&= \frac{n_1}{n} \E_{\tilde F}  \left( \frac{\kappa}{\pi(\bX_i)} Y_i(1) \right)\\
&= \frac{n_1}{n} \cdot \kappa \E_{\tilde F}  \left( \frac{1}{\pi(\bX_i)} Y_i(1) \right) \\
&= \frac{n_1}{n} \cdot \kappa \int \frac{Y_i(1)}{\pi(\bX_i)} d\tilde F(\bX_i, Y_i(1)) \\
&= \frac{n_1}{n} \cdot \kappa \int \frac{Y_i(1)}{\pi(\bX_i)} \cdot \pi(\bX_i) d F(\bX_i, Y_i(1)) \\
&= \frac{n_1}{n} \cdot \kappa  \int Y_i(1) dF(\bX_i, Y_i(1)) \\
&= \frac{n_1}{n} \kappa \cdot \E_{F}(Y_i(1))
\end{align*} 
Therefore, re-writing $(1)$:
\begin{align*} 
\frac{\E_{\tilde F} (w_i T_iY_i)}{\E_{\tilde F} (w_i T_i)}&= \frac{p \kappa \cdot \E_F(Y_i(1))}{p \cdot \kappa} \\
&= \E_F(Y_i(1))
\end{align*} 
Similarly, we can show that the second term, $\E_{\tilde F} (w_i(1-T_i)Y_i)/\E_{\tilde F}(w_i(1-T_i))$, is equal to $\E_{F}(Y_i(0))$. Therefore: 
\begin{align*} 
\E_{\tilde F} (\hat \tau_W) &\cip \E_F(Y_i(1)) - \E_F(Y_i(0)) \\
&= \tau
\end{align*} 

Now we will show that $\hat \tau_{\hat Y} \cip 0$. Once again, applying Law of Large Numbers and the Continuous Mapping Theorem: 
\begin{align*} 
\hat \tau_{\hat Y} &= \frac{\sum_{i \in \cS} w_i T_i \hat Y_i}{\sum_{i \in \cS} w_i T_i } -\frac{\sum_{i \in \cS} w_i (1-T_i) \hat Y_i}{\sum_{i \in \cS} w_i (1-T_i)} \\
&\cip \frac{\E_{\tilde F} ( w_i T_i \hat Y_i)}{\E_{\tilde F} ( w_i T_i)} - 
\frac{\E_{\tilde F} (w_i (1-T_i) \hat Y_i)}{\E_{\tilde F} ( w_i (1-T_i))} \\
&= \frac{p \cdot \E_{\tilde F}(w_i \hat Y_i)}{p \E_{\tilde F} (w_i)} - 
\frac{(1-p) \cdot \E_{\tilde F} (w_i \hat Y_i)}{(1-p)\E_{\tilde F} ( w_i)} \\
&= \frac{\E_{\tilde F} (w_i \hat Y_i)}{\E_{\tilde F} (w_i) }- \frac{\E_{\tilde F} (w_i \hat Y_i)}{\E_{\tilde F} (w_i)}  \\
&= 0
\end{align*}
where the third line follows from the fact that treatment assignment is randomized and independent of weights.
Therefore, by the Continuous Mapping Theorem, $\hat \tau_{W}^{res} \cip \tau$. \\
\end{proof}

\subsection{Proof of Theorem~\ref{thm:var_compare}} \mbox{}\\
\noindent\fbox{%
\vspace{2mm}
\parbox{\textwidth}{%
\vspace{2mm}
  The difference between the asymptotic variance of $\hat \tau_W^{res}$ and the asymptotic variance of $\hat \tau_W$ is: 
  \begin{align*} 
\asyvar_{\tilde F}(\hat \tau_W)&-\asyvar_{\tilde F}(\hat \tau_W^{res})  \nonumber \\
              &= -\frac{1}{p(1-p)}\var_w(\hat Y_i) +\frac{2}{p} \cov_w(Y_i(1), \hat Y_i) + \frac{2}{1-p} \cov_w(Y_i(0), \hat Y_i),                
  \end{align*} 
  }
}

\begin{proof} 
From Corollary \ref{cor:var_w}, the asymptotic variance of a weighted estimator is: 
$$\asyvar_{\tilde F}(\hat \tau_{W}) = \frac{1}{p} \var_w(Y_i(1)) + \frac{1}{1-p} \var_w(Y_i(0))$$
Using the residualized potential outcomes $\hat e_i(1)$ and $\hat e_i(0)$, the asymptotic variance of a weighted residualized estimator is: 
$$\asyvar_{\tilde F}(\hat \tau_{W}^{res}) = \frac{1}{p} \var_w(\hat e_i(1)) + \frac{1}{1-p} \var_w(\hat e_i(0)).$$
From the definition of potential residuals, we can write the potential residuals as a function of the original outcome values and the fitted values:
\begin{align} 
\var_w(\hat e_i(0)) &= \var_w(Y_i(0) - \hat Y_i) \nonumber \\
 &= \var_w(Y_i(0)) + \var_w(\hat Y_i) - 2 \cov_w(Y_i(0), \hat Y_i) \label{eqn:e0}\\ 
\nonumber \\
 \var_w(\hat e_i(1)) &= \var_w(Y_i(1) - \hat Y_i) \nonumber \\
 &= \var_w(Y_i(1)) + \var_w(\hat Y_i) - 2 \cov_w(Y_i(1), \hat Y_i) \label{eqn:e1}
\end{align} 
Therefore, the difference in variances of our two estimators is
\begin{align*} 
\asyvar_{\tilde F}(\hat \tau_W) &- \asyvar_{\tilde F}(\hat \tau^{res}_W) \\ 
=& \left\{\frac{1}{p} \var_w(Y_i(1)) + \frac{1}{1-p} \var_w(Y_i(0))) \right\} -
\left\{ \frac{1}{p} \var_w(\hat e_i(1)) + \frac{1}{1-p} \frac{1}{n_0} \var_w(\hat e_i(0))\right\} \\
=& \frac{1}{p} \cdot \big( \var_w(Y_i(1))-\var_w(\hat e_i(1)) \big) +\frac{1}{1-p}\cdot \big(\var_w(Y_i(0)) -\var_w(\hat e_i(0))\big) 
\intertext{Plugging in (\ref{eqn:e0}) and (\ref{eqn:e1}):}
=& -\frac{1}{p} \cdot \bigg\{ \var_w(Y_i(1)) + \var_w(\hat Y_i) - 2 \cov_w(Y_i(1), \hat Y_i) - \var_w(Y_i(1) \bigg\} \\
& - \frac{1}{1-p} \cdot \bigg\{ \var_w(Y_i(0)) + \var_w(\hat Y_i) - 2 \cov_w(Y_i(0), \hat Y_i) - \var_w( Y_i(0)\bigg\} \\
=& - \frac{1}{p(1-p)} \cdot \var_w(\hat Y_i) + \frac{2}{p} \cdot \cov_w(Y_i(1), \hat Y_i) + \frac{2}{1-p} \cdot \cov_w(Y_i(0), \hat Y_i)
\end{align*} 
\end{proof} 

\subsection{Proof of Corollary~\ref{cor:rel_reduc}} \mbox{}\\
\noindent\fbox{%
\vspace{2mm}
\parbox{\textwidth}{%
\vspace{2mm}
  With $R^2_0$ defined as above, define $R^2_1$ as the weighted goodness-of-fit across the treatment units. Let $\xi = R^2_0 - R^2_1$, such that: 
  $$R^2_1 = 1 - \frac{\var_{w}(\hat e_i(1))}{\var_{w}(Y_i(1))} = R_0^2 - \xi.$$
 Furthermore, define the ratio $f = p \var_{w}(Y_i(0)) / (1-p) \var_{w}(Y_i(1))$.
  Then the relative reduction in variance from residualizing first is given by: 
  $$\text{Relative Reduction} = \frac{\asyvar_{\tilde F}(\hat \tau_W) - \asyvar_{\tilde F}(\hat
    \tau_W^{res})}{\asyvar_{\tilde F}(\hat \tau_W)} = R^2_0 - \frac{1}{1+f} \cdot
  \xi$$ 
}
}
\begin{proof} 
Let $C_1 = 1/p$ and $C_0 = 1/1-p$. 
\begin{align*} 
 \frac{\asyvar_{\tilde F}(\hat \tau_W) - \asyvar_{\tilde F}(\hat
    \tau_W^{res})}{\asyvar_{\tilde F}(\hat \tau_W)}&=
\frac{C_1 \var_w(Y_i(1)) + C_0 \var_w(Y_i(0)) - (C_1 \var_w(\hat e_i(1)) + C_0 \var_w(\hat e_i(0)) )}{C_1 \var_w(Y_i(1)) + C_0 \var_w(Y_i(0)) } \\
&= \frac{C_1 \var_w(Y_i(1)) - C_1 \var_w(\hat e_i(1))+ C_0 \var_w(Y_i(0)) - C_0 \var_w(\hat e_i(0)) )}{C_1 \var_w(Y_i(1)) +C_0 \var_w(Y_i(0)) } 
\intertext{Dividing the numerator and denominator by $C_1 \cdot \var(Y_i(1))$, and defining $f = C_0 \var_w(Y_i(0))/C_1 \var_w(Y_i(1))$:}
&= \frac{1 - \var_w(\hat e_i(1))/\var_w(Y_i(1)) + f - f \cdot \var_w(\hat e_i(0))/\var_w(Y_i(0))}{1 + f}\\
&= \frac{1}{1+f} \left( R^2_1 + f R^2_0 \right) 
\intertext{Using the definition of $\xi = R^2_0 - R^2_1$:} 
&= \frac{1}{1+f} \left(R^2_0 - \xi + f R^2_0 \right) \\
&= R^2_0 - \frac{1}{1+f} \cdot \xi 
\end{align*} 

\end{proof} 

\subsection{Proof of Theorem~\ref{thm:wLS_consistency}}
\mbox{}\\

\noindent\fbox{%
\vspace{2mm}
\parbox{\textwidth}{%
\vspace{2mm}
Suppose Assumption~\ref{assum-ind} 
holds with $\bX_i$, the
Post-Residualized Weighted Least Squares Estimator is a consistent
estimator for the PATE:
$$\hat \tau_{wLS}^{res} \cip \tau$$
}
}
\begin{proof} 
The proof of this follows almost identically from Theorem~\ref{thm:consistency_ipw}. 
To begin, we can write $\hat \tau^{res}_{wLS}$ as the above estimator on the residuals of the initial population regression:
\begin{align*} 
\hat \tau^{res}_{wLS} &= \frac{1}{(\sum_{i \in \cS} w_i T_i)} \left(\sum_{i \in \cS} w_i T_i (\hat e_i - \bX_i \hat \gamma^{res} ) \right) - 
\left(\frac{1}{{(\sum_{i \in \cS} w_i (1-T_i))}} \sum_{i \in \cS} w_i(1-T_i) (\hat e_i - \bX_i \hat \gamma^{res})\right) \\
&= \underbrace{\frac{ \sum_{i \in \cS} w_i T_i \hat e_i}{\sum_{i \in \cS} w_i T_i} - \frac{\sum_{i \in \cS} w_i (1-T_i) \hat e_i}{\sum_{i \in \cS} w_i (1-T_i)}}_{= \hat \tau_{W}^{res}} - \underbrace{\left( \frac{\sum_{i \in \cS} w_i T_i\bX_i \hat \gamma^{res}}{\sum_{i \in \cS} w_i T_i} - \frac{\sum_{i \in \cS} w_i (1-T_i) \bX_i \hat \gamma^{res}}{\sum_{i \in \cS} w_i (1-T_i)} \right)}_{(*)},
\end{align*} 
where $\hat \gamma^{res}$ represents the estimated coefficients for the covariates $\bX_i$ in the weighted regression run on the residualized outcomes $\hat e_i$. Note that the above represents two distinct regression steps: $\hat e_i$ is the result of the first population regression. $\hat \gamma^{res}$ is estimated for the covariates $\bX_i$ from the second regression using the residualized sample outcomes, $\hat e_i$.  

From Theorem~\ref{thm:consistency_ipw}, 
$\hat \tau_{W}^{res} \cip \tau$. Looking just at the $(*)$ term: 
\begin{align*} 
\frac{\sum_{i \in \cS} w_i T_i\bX_i \hat \gamma^{res}}{\sum_{i \in \cS} w_i T_i} - \frac{\sum_{i \in \cS} w_i (1-T_i) \bX_i \hat \gamma^{res}}{\sum_{i \in \cS} w_i (1-T_i)} &= \left( \frac{\sum_{i \in \cS} w_i T_i\bX_i }{\sum_{i \in \cS} w_i T_i} - \frac{\sum_{i \in \cS} w_i (1-T_i) \bX_i }{\sum_{i \in \cS} w_i (1-T_i)}\right) \hat \gamma^{res}
\end{align*}
Under standard regularity conditions for least squares, $\hat \gamma^{res}$ converges to $\gamma_{\ast}^{res}$. Furthermore, using Law of Large Numbers and the Continuous Mapping Theorem: 
\begin{align*} 
\frac{\sum_{i \in \cS} w_i T_i\bX_i }{\sum_{i \in \cS} w_i T_i} - \frac{\sum_{i \in \cS} w_i (1-T_i) \bX_i }{\sum_{i \in \cS} w_i (1-T_i)} &\cip \frac{\E_{\tilde F} (w_i T_i \bX_i)}{\E_{\tilde F} (w_i T_i)} - \frac{\E_{\tilde F} (w_i (1-T_i) \bX_i)}{\E_{\tilde F} (w_i (1-T_i))} \\
&= \frac{\E_{\tilde F} (w_i \bX_i)}{\E_{\tilde F} (w_i)} - \frac{\E_{\tilde F} (w_i \bX_i)}{\E_{\tilde F} (w_i)} \\ 
&= 0 
\end{align*} 
As such, we see that the term in $(*)$ will converge in probability to zero. Therefore, $\hat \tau^{res}_{wLS} \cip \tau$. 
\end{proof} 

\subsection{Proof of Theorem~\ref{thm:eff_gain_wls}}
\mbox{}\\

\noindent\fbox{%
\vspace{2mm}
\parbox{\textwidth}{%
\vspace{2mm}
The difference between the asymptotic variance of $\hat \tau_{wLS}$ and the asymptotic variance of $\hat \tau_{wLS}^{res}$ is: 
\begin{align*}
  &\asyvar_{\tilde F}(\hat \tau_{wLS}) - \asyvar_{\tilde F}(\hat \tau_{wLS}^{res}) \\
  \ = \ & \frac{1}{p} \left\{\var_w(Y_i(1) - \tilde \bX_i^\top \gamma_\ast) - \var_w(Y_i(1) - \hat{g}(\bX_i)) \right\} \\
  & \ + \ \frac{1}{1-p} \left\{ \var_w(Y_i(0) - \tilde \bX_i^\top \gamma_\ast)- \var_w(Y_i(0) - \hat{g}(\bX_i)) \right\}\\
 & \ + \ \frac{2}{p} \cov_w(\hat e_i(1), \tilde \bX_i^\top \gamma_\ast^{res}) + \frac{2}{1-p} \cov_w(\hat e_i(0), \tilde \bX_i^\top \gamma_\ast^{res}) - \frac{1}{p(1-p)} \var_w(\tilde \bX_i^\top \gamma_\ast^{res}),
\end{align*} 
where ${\gamma_\ast}$ and $\gamma_\ast^{res}$
  are the true coefficients associated with the pre-treatment covariates, $\widetilde{\bX}_i$
  defined in the weighted least squares regression
  (equation~\eqref{eq:wls}) 
  and the post-residualized
  weighted least squares regression (equation~\eqref{eq:prw-wls}),
  respectively. Formally,
  ${\gamma}_\ast$ and $\gamma^{res}_\ast$ are formally
  defined as the solution to the following optimization problems.
  \begin{equation}
    (\tau_{wLS}, \alpha_\ast, \gamma_\ast) =
    \argmin_{\tau, \alpha, \gamma}
    \E_{\tilde{F}} \left\{\hat{w}_i \left(Y_i - (\tau T_i +
        \alpha + \widetilde{\bX}^\top_i \gamma)\right)^2\right\} 
  \end{equation}
  \begin{equation}
    (\tau^{res}_{wLS}, \alpha^{res}_\ast, \gamma^{res}_\ast) =
    \argmin_{\tau, \alpha, \gamma}
    \E_{\tilde{F}} \left\{\hat{w}_i \left(\widehat{e}_i - (\tau T_i +
        \alpha + \widetilde{\bX}^\top_i \gamma)\right)^2\right\} 
  \end{equation}
}
}

\begin{proof} \mbox{}\\
\begin{align} 
  & \asyvar_{\tilde F}(\hat \tau_{wLS}) -
  \asyvar_{\tilde F}(\hat \tau^{res}_{wLS}) \\
  = & \left\{\frac{1}{p} \var_w(Y_i(1) - \tilde \bX_i^\top \gamma_\ast) + 
\frac{1}{1-p} \var_w(Y_i(0) - \tilde \bX_i^\top \gamma_\ast) \right\} \nonumber \\
& - \left\{\frac{1}{p} \var_w(\hat e_i(1) - \tilde \bX_i^\top \gamma_\ast^{res}) + 
\frac{1}{1-p} \var_w(\hat e_i(0) - \tilde \bX_i^\top \gamma_\ast^{res})\right\}
\label{eqn:wLS_var_1} 
\end{align} 

The adjusted residualized outcomes can be re-written as a function of the residualized outcomes and the fitted values from the regression. First, for the treatment outcomes: 
\begin{align*} 
\var_w(\hat e_i(1) - \tilde \bX_i^\top \gamma_\ast^{res}) &= \var_w(Y_i(1) - \hat{g}(\bX_i) - \tilde \bX_i^\top \gamma_\ast^{res}) \\
&= \var_w(Y_i(1) - \hat{g}(\bX_i)) + \var_w(\tilde \bX_i^\top \gamma_\ast^{res}) - 2 \cov_w(Y_i(1) - \hat{g}(\bX_i), \tilde \bX_i^\top \gamma_\ast^{res})
\end{align*} 
Similarly, 
\begin{align*} 
\var_w(\hat e_i(0) - \tilde \bX_i^\top \gamma_\ast^{res}) &= \var_w(Y_i(0) - \hat{g}(\bX_i) - \tilde \bX_i^\top \gamma_\ast^{res}) \\
&= \var_w(Y_i(0) - \hat{g}(\bX_i)) + \var_w(\tilde \bX_i^\top \gamma_\ast^{res}) - 2 \cov_w(Y_i(0) - \hat{g}(\bX_i), \tilde \bX_i^\top \gamma_\ast^{res})
\end{align*} 
Plugging into Equation \eqref{eqn:wLS_var_1}:
\begin{align*} 
  & \asyvar_{\tilde F}(\hat \tau_{wLS}) -
\asyvar_{\tilde F}(\hat \tau^{res}_{wLS})\\ 
\ = \ & \frac{1}{p} \left\{\var_w(Y_i(1) - \tilde \bX_i^\top \gamma_\ast) - \var_w(Y_i(1) - \hat{g}(\bX_i)) \right\} \\
& \ + \ \frac{1}{1-p} \left\{ \var_w(Y_i(0) - \tilde \bX_i^\top \gamma_\ast)- \var_w(Y_i(0) - \hat{g}(\bX_i)) \right\}\\
& \ - \ \left\{ \frac{1}{p(1-p)} \var_w(\tilde \bX_i^\top \gamma_\ast^{res}) - \frac{2}{p} \cov_w(Y_i(1) - \hat{g}(\bX_i), \tilde \bX_i^\top \gamma_\ast^{res}) - \frac{2}{1-p} \cov_w(Y_i(0) - \hat{g}(\bX_i), \tilde \bX_i^\top \gamma_\ast^{res}) \right\}\\
\ = \ &\frac{1}{p} \left\{\var_w(Y_i(1) - \tilde \bX_i^\top \gamma_\ast) - \var_w(Y_i(1) - \hat{g}(\bX_i)) \right\} 
 + \frac{1}{1-p} \left\{ \var_w(Y_i(0) - \tilde \bX_i^\top \gamma_\ast)- \var_w(Y_i(0) - \hat{g}(\bX_i)) \right\} \\
& \ + \ \left\{ -\frac{1}{p(1-p)} \var_w(\tilde \bX_i^\top \gamma_\ast^{res}) + \frac{2}{p} \cov_w(\hat e_i(1), \tilde \bX_i^\top \gamma_\ast^{res}) + \frac{2}{1-p} \cov_w(\hat e_i(0), \tilde \bX_i^\top \gamma_\ast^{res}) \right\}
\end{align*} 
\end{proof} 

\section{Diagnostic Measure} 
We detail how to estimate the diagnostic measures in this section. 
To estimate the diagnostic for the post-residualized weighted estimator, we compute the estimated weighted variance of both the residuals and the outcomes for the units assigned to control: 
\begin{align} 
\hat R^2_0 &= 1 - \frac{\widehat{\var}_{w,0}(\hat e_i)}{\widehat{\var}_{w,0}(Y_i)} \nonumber \\
&= 1 - \frac{\sum_{i \in \cS} w_i^2 (1-T_i) (\hat e_i - \hat \mu^{res}_0)^2}{\sum_{i \in \cS} w_i^2 (1-T_i) (Y_i - \hat \mu_0)^2}
\label{eqn:r2_diag_est} 
\end{align} 
where $\hat \mu_0$ and $\hat \mu_0^{res}$ are defined as: 
\begin{equation} 
\hat \mu_0 = \frac{\sum_{i \in \cS} w_i (1-T_i) Y_i}{\sum_{i \in \cS} w_i (1-T_i)}, \ \ \ \ \  \ \ \hat \mu^{res}_0 = \frac{\sum_{i \in \cS} w_i (1-T_i) \hat e_i}{\sum_{i \in \cS} w_i (1-T_i)}
\end{equation} 

For the post-residualized weighted least squares estimator, estimating the diagnostic follows similarly, but we now have to account for the covariate adjustment taking place: 
\begin{align} 
\hat R^2_{0, wLS} &= 1 - \frac{\widehat{\var}_{w,0}(\hat e_i - \tilde \bX_i^\top \hat \gamma_0^{res})}{\widehat{\var}_{w,0}(Y_i - \tilde \bX_i^\top \hat \gamma_0)} \nonumber \\
&= 1 - \frac{\sum_{i \in \cS} w_i^2 (1-T_i) (\hat \epsilon^{res}_i - \hat \epsilon^{res}_0)^2}{\sum_{i \in \cS} w_i^2 (1-T_i) (\hat \epsilon_i - \hat \epsilon_0)^2},
\label{eqn:r2_wls_diag_est}
\end{align} 
where $\hat \epsilon_i$ represents the residuals estimated from regressing the outcomes $Y_i$ on the pre-treatment covariates $\tilde \bX_i$, across the subset of units assigned to control (i.e., $Y_i - \tilde \bX_i^\top \hat \gamma_0$, where $\hat \gamma_0$ is estimated by running the regression $Y_i \sim \tilde \bX_i$ across units assigned to control). $\hat\epsilon_i^{res}$ is analogously defined for the residualized outcomes $\hat e_i$. $\hat \epsilon_0$ and $\hat \epsilon_0^{res}$ are the weighted average of both $\hat \epsilon_i$ and $\hat \epsilon_i^{res}$, respectively.

When treating $\hat Y_i$ as a covariate, the diagnostic can be estimated in an analogous way, but by first performing sample splitting. More specifically, the procedure for including $\hat Y_i$ as a covariate for the weighted estimator is as follows: 
\begin{enumerate} 
\item Across the subset of units assigned to control, randomly partition the units into two subsets: $S_1$ and $S_2$. Without loss of generality, we will use $S_1$ as our training sample, and $S_2$ as our testing sample. 
\item Regress $\hat Y_i$ on the outcomes across $S_1$ to obtain a $\hat \beta$ value. 
\item Using $\hat \beta$, estimate the out-of-sample residuals $\hat e_i^{oos}$ across $S_2$, where $\hat e_i^{oos}:= Y_i - \hat \beta \hat Y_i$.
\item Estimate the diagnostic using $\hat e_i^{oos}$ and the outcomes $Y_i$ across $S_2$ using Equation \eqref{eqn:r2_diag_est}. 
\item Cross-fit: repeat steps 1-3, but flipping $S_1$ and $S_2$ (i.e., regress $\hat Y_i$ on the outcomes across $S_2$ to obtain a $\hat \beta$ value, and estimate the diagnostic across $S_1$). 
\item Average the two diagnostic values together. 
\end{enumerate} 

When including $\hat Y_i$ as a covariate for the weighted least squares estimator, researchers can repeat the procedure above; however, when estimating the diagnostic using $\hat e_i^{oos}$, researchers must account for $\tilde \bX_i$. More specifically: 
\begin{enumerate} 
\item Follow Steps 1-3 above to obtain $\hat e_i^{oos}$ across $S_1$. 
\item Regress $\hat e_i^{oos}$ on $\tilde \bX_i$, and regress $Y_i$ on $\tilde \bX_i$ across $S_2$. Use Equation \eqref{eqn:r2_wls_diag_est} to estimate the diagnostic value. 
\item Cross fit, and average the two diagnostic values together. 
\end{enumerate} 

When researchers have relatively small sample sizes, it can be advantageous to perform repeated sample splitting, and take the average of the diagnostic across all the repeated splits (see \cite{jacob2020cross} for more details).

\section{Simulations} \label{sec:appendix_sims} 
This section provides details associated with the simulations described in Section~\ref{sec:sims} 
of the main manuscript.

\subsection{Simulation Set-Up}  
To begin, we randomly generate a set of covariates $\begin{bmatrix} X_1 & X_2 & X_S & X_\tau \end{bmatrix} \sim MVN(\mathbf{0}, \Sigma)$ with the following covariance structure: 
$$\Sigma = \begin{bmatrix} 1 & 0 & 0.45 & 0.5 \\ 
0 & 1 & 0 & 0 \\ 
0.45 & 0 & 1 & 0.9 \\ 
0.5 & 0 & 0.9 & 1 \end{bmatrix} $$

\noindent where, recall, $(X_{1i}, X_{2i})$ are observed pre-treatment covariates, $X_{Si}$ controls the probability of inclusion in the experimental sample, and $X_{\tau i}$ determines the treatment effect.

Unit $i$'s propensity for being included in the experimental sample (recorded as $S_i = 1$) is governed by a logit model on the covariate $X_{Si}$: 
$$P(S_i = 1) \propto \frac{\exp(X_{Si})}{1 + \exp(X_{Si})}.$$ 
At each iteration of the simulation, an experimental sample is drawn using the propensity score, as well as a random sample of the population. The sampled population is used to estimate the residualizing model and sampling weights.

Each specific data generating process for the potential outcome under control is determined by the values of the $\beta$s and $\gamma$s and $\alpha$.  Below, we provide the parameter values and simplified DGP for $Y_i(0)$.
\begin{itemize} 
\item Scenario 1: Linear Data Generating Process, identical population/sample DGP

$\beta_1 = 2, \beta_2 = 1, \beta_3 = 0, \beta_S = 0, \gamma_1 = 0, \gamma_2 = 0, \gamma_3 = 0, \gamma_4 = 0, \alpha = 0$, yielding:
$$Y_i(0) = 2 X_{1i} + X_{2i}  + \varepsilon_i$$
\item Scenario 2: Nonlinear Data Generating Process, identical population/sample DGP

$\beta_1 = 2, \beta_2 = 1, \beta_3 = 0, \beta_S = 2.5, \gamma_1 = 0.5, \gamma_2 = 3, \gamma_3 = 2.5, \gamma_4 = 0, \alpha = 0$, yielding:
$$Y_i(0) = 2X_{1i} + X_{2i} + 0.5X_{1i}^2 + 3 \sqrt{|X_{2i}|} + 2.5 \big(X_{1i} \cdot X_{2i}\big) + \varepsilon_i$$
\item Scenario 3: Linear Data Generating Process, different population/sample DGP

$\beta_1 = 2, \beta_2 = 1, \beta_3 = -1, \beta_S = \beta_S, \gamma_1 = 0, \gamma_2 = 0, \gamma_3 = 0, \gamma_4 = 0, \alpha = 0.5$, yielding:
\begin{align*}
Y_i(0) =& 2 X_{1i} + X_{2i} + \beta_S \cdot (1-S_i) \cdot (0.5 - X_{1i}) + \varepsilon_i,
\end{align*}
\item Scenario 4: Nonlinear Data Generating Process, different population/sample DGP

$\beta_1 = 2, \beta_2 = 1, \beta_3 = -1, \beta_S = \beta_S, \gamma_1 = 0.5, \gamma_2 = 3, \gamma_3 = 2.5, \gamma_4 = 1.5, \alpha = 0.5$, yielding:
\begin{align*} 
Y_i(0) =& 2 X_{1i} + X_{2i} + 0.5 X_{1i}^2 + 3\sqrt{|X_{2i}|} + 2.5 \big(X_{1i} \cdot X_{2i}) \\ 
&\beta_S \cdot (1-S_i) \cdot (0.5- X_{1i} + 1.5 X_{1i} \cdot X_{2i}) + \varepsilon_i,
\end{align*} 
\end{itemize} 
For Scenarios 3 and 4, $\beta_S$ takes on values $\{-5, -2, -1, 0, 1, 2, 5\}$.

\subsection{Supplementary Tables} 

Table~\ref{tbl:summary_sim} presents summary results for estimator performance under Scenarios 1 and 2, including MSE, Bias, and SE.  Column 1 presents the baseline results for the difference-in-means (DiM).  Columns 2-4 present the results for the weighted estimators and columns 5-7 present results for the weighted least squares estimator.  For the weighted and weighted least squares estimators we present the standard estimator without residualizing, the directly residualized estimator and inclusion of $\hat{Y}$ as a covariate.

\begin{table}[t]
\begin{center} 
\textbf{Summary of Estimator Performance (N=10,000)} \\ \vspace{2mm} 
\begin{tabular}{ll|c|ccc|ccccc}
  \hline
   & & &  \multicolumn{3}{c|}{\underline{Weighted}} &  \multicolumn{3}{c}{\underline{Weighted Least Squares}} \\ 
  & & DiM & $\widehat{\tau}_W$ & $\widehat{\tau}^{res}_W $ & $\widehat{\tau}^{cov}_{W} $ & $\widehat{\tau}_{wLS}$ & $\widehat{\tau}^{res}_{wLS}$ & $\widehat{\tau}^{cov}_{wLS}$ \\
  \hline
  \multicolumn{6}{l}{Scenario 1: Linear Outcome Model} \\ \hline
n=100 & MSE & 36.44 & 30.05 & 1.48 & 1.34 & 1.34 & 1.34 & 1.30 \\ 
   & Bias & 3.60 & -0.13 & 0.05 & 0.12 & 0.19 & 0.19 & 0.27 \\ 
   & SE & 4.85 & 5.48 & 1.22 & 1.15 & 1.14 & 1.14 & 1.11 \\ 
  n=1000 & MSE & 16.41 & 2.98 & 0.17 & 0.15 & 0.14 & 0.14 & 0.13 \\ 
   & Bias & 3.74 & 0.00 & -0.01 & 0.00 & 0.00 & 0.00 & 0.01 \\ 
   & SE & 1.56 & 1.73 & 0.41 & 0.38 & 0.38 & 0.38 & 0.36 \\ 
  n=5000 & MSE & 14.39 & 0.64 & 0.04 & 0.03 & 0.03 & 0.03 & 0.03 \\ 
   & Bias & 3.72 & 0.01 & 0.00 & 0.00 & 0.01 & 0.01 & 0.01 \\ 
   & SE & 0.72 & 0.80 & 0.19 & 0.19 & 0.18 & 0.18 & 0.18 \\ \hline
    \multicolumn{7}{l}{Scenario 2: Nonlinear Outcome Model} \\ \hline
   n=100 & MSE & 70.71 & 58.80 & 8.25 & 8.20 & 36.59 & 8.16 & 8.04 \\ 
   & Bias & 3.44 & -0.30 & 0.09 & 0.14 & 0.04 & 0.23 & 0.26 \\ 
   & SE & 7.68 & 7.67 & 2.87 & 2.86 & 6.05 & 2.85 & 2.83 \\ 
  n=1000 & MSE & 20.37 & 5.58 & 0.82 & 0.80 & 3.53 & 0.79 & 0.78 \\ 
   & Bias & 3.78 & 0.05 & -0.00 & -0.00 & 0.05 & 0.00 & 0.01 \\ 
   & SE & 2.46 & 2.36 & 0.91 & 0.90 & 1.88 & 0.89 & 0.89 \\ 
  n=5000 & MSE & 14.80 & 1.17 & 0.18 & 0.18 & 0.83 & 0.17 & 0.17 \\ 
   & Bias & 3.68 & -0.02 & -0.01 & -0.01 & -0.03 & -0.01 & -0.00 \\ 
   & SE & 1.12 & 1.08 & 0.42 & 0.42 & 0.91 & 0.42 & 0.42 \\ 
   \hline
\end{tabular}
\end{center} 
\caption{Summary of estimator performance for Scenarios 1 and 2. The population is fixed at $N=10,000$, and $1,000$ iterations were run for each sample size. MSE is scaled by 100, and the bias and standard error are scaled by 10.} 
\label{tbl:summary_sim} 
\end{table}

Table~\ref{tbl:summary_sim3_4} presents summary results for estimator performance under Scenarios 3 and 4, including MSE and Bias.  In these scenarios we vary the value of $\beta_S$, presented in column 1, which controls the degree of alignment between the experimental sample outcomes and the population outcomes.  We fix the experimental sample size at $n = 1,000$. Columns 2-3 presents the baseline results for the difference-in-means (DiM).  Columns 4-9 present the results for the weighted estimators and columns 10-15 present results for the weighted least squares estimator.  For the weighted and weighted least squares estimators we present the standard estimator without residualizing, the directly residualized estimator and inclusion of $\hat{Y}$ as a covariate.

\begin{table}[ht]
\begin{center} 
\textbf{Summary of Estimator Performance - Scenario 3 and 4 (N = 10,000)} \\ \vspace{2mm} 
\resizebox{\textwidth}{!}{%
\begin{tabular}{l|cc|cccccc|ccccccc} \hline 
& & & \multicolumn{6}{c|}{\underline{Weighted}}& \multicolumn{6}{c}{\underline{Weighted Least Squares}}\\
 & \multicolumn{2}{c}{DiM} & \multicolumn{2}{|c}{$\widehat{\tau}_W$} & \multicolumn{2}{c}{$\widehat{\tau}^{res}_W $} & \multicolumn{2}{c}{$\widehat{\tau}^{cov}_{W} $} & \multicolumn{2}{|c}{$\widehat{\tau}_{wLS}$} & \multicolumn{2}{c}{$\widehat{\tau}^{res}_{wLS}$} & \multicolumn{2}{c}{$\widehat{\tau}^{cov}_{wLS}$} \\ \hline 
 $\beta_S$ & MSE & Bias & MSE & Bias & MSE & Bias & MSE & Bias & MSE & Bias & MSE & Bias & MSE & Bias \\\hline 
 \multicolumn{7}{l}{Scenario 3: Linear Outcome}\\ \hline 
-5 & 16.41 & 3.74 & 2.98 & 0.00 & 10.69 & -0.11 & 0.36 & -0.03 & 0.14 & 0.00 & 0.14 & 0.00 & 0.13 & 0.01 \\ 
  -2.5 & 15.83 & 3.67 & 3.07 & -0.06 & 2.55 & 0.06 & 0.25 & 0.01 & 0.14 & 0.02 & 0.14 & 0.02 & 0.13 & 0.04 \\ 
  -2 & 16.05 & 3.72 & 2.99 & 0.01 & 1.54 & 0.02 & 0.22 & 0.02 & 0.14 & 0.02 & 0.14 & 0.02 & 0.14 & 0.04 \\ 
  -1 & 16.11 & 3.73 & 2.88 & 0.05 & 0.39 & -0.02 & 0.16 & 0.00 & 0.14 & -0.00 & 0.14 & -0.00 & 0.13 & 0.02 \\ 
  -0.5 & 16.37 & 3.75 & 2.89 & 0.07 & 0.17 & -0.02 & 0.14 & -0.00 & 0.13 & 0.00 & 0.13 & 0.00 & 0.13 & 0.02 \\ 
  0 & 16.50 & 3.75 & 3.04 & 0.06 & 0.16 & 0.00 & 0.14 & 0.00 & 0.13 & 0.00 & 0.13 & 0.00 & 0.13 & 0.02 \\ 
  0.5 & 16.38 & 3.74 & 3.19 & 0.04 & 0.41 & 0.01 & 0.21 & 0.01 & 0.13 & 0.01 & 0.13 & 0.01 & 0.12 & 0.02 \\ 
  1 & 16.11 & 3.72 & 3.03 & 0.00 & 0.92 & 0.01 & 0.54 & 0.02 & 0.13 & -0.01 & 0.13 & -0.01 & 0.12 & 0.01 \\ 
  2 & 16.23 & 3.74 & 3.03 & 0.01 & 2.68 & 0.04 & 2.68 & 0.05 & 0.14 & -0.00 & 0.14 & -0.00 & 0.13 & 0.01 \\ 
  2.5 & 16.09 & 3.71 & 3.15 & -0.01 & 3.92 & 0.01 & 3.15 & -0.00 & 0.14 & -0.01 & 0.14 & -0.01 & 0.13 & 0.01 \\ 
  5 & 16.33 & 3.71 & 3.23 & 0.00 & 14.32 & -0.01 & 1.54 & 0.02 & 0.14 & -0.00 & 0.14 & -0.00 & 0.13 & 0.01 \\ 
\hline 
\multicolumn{7}{l}{Scenario 4: Nonlinear Outcome}\\ \hline
    -5 & 20.31 & 3.74 & 5.77 & 0.04 & 37.03 & -0.01 & 5.66 & 0.04 & 3.72 & 0.05 & 26.19 & 0.10 & 1.02 & 0.03 \\ 
  -2.5 & 20.31 & 3.74 & 6.17 & -0.01 & 9.55 & 0.10 & 5.10 & 0.04 & 3.96 & 0.05 & 7.57 & 0.06 & 1.67 & 0.04 \\ 
  -2 & 19.50 & 3.65 & 5.92 & -0.08 & 6.22 & -0.00 & 4.27 & -0.04 & 3.86 & -0.04 & 5.05 & -0.05 & 2.89 & -0.00 \\ 
  -1 & 19.77 & 3.73 & 5.71 & -0.02 & 2.18 & -0.08 & 2.14 & -0.07 & 3.91 & -0.09 & 1.92 & -0.05 & 1.08 & 0.02 \\ 
  -0.5 & 19.75 & 3.68 & 5.70 & -0.06 & 1.10 & -0.03 & 1.09 & -0.03 & 3.96 & -0.12 & 1.06 & -0.01 & 0.83 & 0.05 \\ 
  0 & 19.74 & 3.69 & 5.81 & -0.04 & 0.81 & 0.01 & 0.80 & 0.01 & 3.71 & -0.05 & 0.77 & 0.02 & 0.77 & 0.02 \\ 
  0.5 & 20.49 & 3.83 & 5.40 & 0.09 & 1.42 & 0.03 & 1.30 & 0.03 & 3.65 & 0.04 & 1.09 & 0.01 & 0.75 & 0.02 \\ 
  1 & 20.24 & 3.80 & 5.52 & 0.08 & 2.84 & -0.05 & 2.04 & -0.01 & 3.95 & 0.06 & 1.91 & -0.07 & 0.80 & -0.01 \\ 
  2 & 20.03 & 3.72 & 5.83 & 0.05 & 7.99 & -0.02 & 3.04 & 0.02 & 4.24 & 0.03 & 5.27 & -0.06 & 0.84 & -0.00 \\ 
  2.5 & 20.45 & 3.74 & 6.04 & 0.06 & 12.15 & -0.11 & 3.51 & -0.01 & 4.28 & 0.05 & 8.32 & -0.09 & 0.85 & -0.00 \\ 
  5 & 20.80 & 3.75 & 6.29 & 0.08 & 45.95 & -0.25 & 5.05 & 0.02 & 4.09 & 0.06 & 29.97 & -0.27 & 0.92 & -0.02 \\ 
 \hline 
\end{tabular}}
\end{center} 
\caption{Summary of estimator performance for Scenarios 3 and 4, where $n=1,000$ and $N=10,000$. 1,000 iterations were run for each $\beta_S$ value. The bias is scaled by 10, and the MSE is scaled by 100.} 
\label{tbl:summary_sim3_4} 
\end{table}

In Table~\ref{tab:simulation_tpr_tnr} we summarize the true positive and true negative rates for the diagnostic measures for the post-residualized estimators.\footnote{True positive rates were calculated by taking the total number of true positives (i.e., cases where the diagnostic correctly indicated there would be efficiency gain from residualizing) and dividing by the total number of cases in which residualizing led to efficiency gain. True negatives are similarly defined.} Column 1 presents the value of $\beta_S$. Columns 2-9 present the post-residualized weighted, post-residualized weighted least squares, the post-residualized weighted estimator with $\hat{Y}$ as a covariate, and the post-residualized weighted least squares estimator with $\hat{Y}$ as a covariate, respectively.   We see that in general, the diagnostic measures are able to adequately capture when residualizing results in precision gain. We see that using sample splitting to estimate the pseudo-$R^2$ measure for the case in which we include $\hat Y_i$ as a covariate can sometimes be conservative, which results in a low true positive rate in cases when the divergence between the experimental sample and population are rather large. In cases where residualizing always leads to losses or gains in precision, the total number of true positive or true negative rates is zero (respectively). 

\begin{table}[ht]
\begin{center} 
\textbf{Diagnostic Performance across Simulations} \\ \vspace{2mm} 
\resizebox{\textwidth}{!}{%
\begin{tabular}{l|cc|cc|cc|cc}
  \hline
  & \multicolumn{2}{c|}{$\widehat{\tau}^{res}_W $} & \multicolumn{2}{c|}{$\widehat{\tau}^{cov}_{W} $} &
  \multicolumn{2}{c|}{$\widehat{\tau}^{res}_{wLS}$}&
  \multicolumn{2}{c}{$\widehat{\tau}^{cov}_{wLS}$} \\ 
$\beta_S$ & TPR & TNR & TPR & TNR & TPR & TNR & TPR & TNR \\ 
  \hline
  \multicolumn{6}{l}{Scenario 3: Linear Outcomes} \\ \hline 
-5 & 0/0 & 1000/1000 & 1000/1000 & 0/0 & 207/472 & 329/528 & 338/705 & 166/295 \\ 
  -2.5 & 1/942 & 58/58 & 1000/1000 & 0/0 & 203/499 & 304/501 & 308/694 & 177/306 \\ 
  -2 & 999/1000 & 0/0 & 1000/1000 & 0/0 & 216/514 & 288/486 & 310/689 & 175/311 \\ 
  -1 & 1000/1000 & 0/0 & 1000/1000 & 0/0 & 219/525 & 287/475 & 293/689 & 188/311 \\ 
  -0.5 & 1000/1000 & 0/0 & 1000/1000 & 0/0 & 214/519 & 282/481 & 293/689 & 183/311 \\ 
  0 & 1000/1000 & 0/0 & 1000/1000 & 0/0 & 223/523 & 275/477 & 283/683 & 177/317 \\ 
  0.5 & 1000/1000 & 0/0 & 1000/1000 & 0/0 & 222/536 & 268/464 & 260/666 & 194/334 \\ 
  1 & 1000/1000 & 0/0 & 1000/1000 & 0/0 & 233/519 & 283/481 & 254/669 & 199/331 \\ 
  2 & 999/1000 & 0/0 & 998/1000 & 0/0 & 228/490 & 321/510 & 297/695 & 175/305 \\ 
  2.5 & 0/0 & 999/1000 & 188/490 & 346/510 & 209/466 & 336/534 & 341/705 & 149/295 \\ 
  5 & 0/0 & 1000/1000 & 1000/1000 & 0/0 & 214/486 & 303/514 & 322/699 & 155/301 \\ \hline 
  \multicolumn{6}{l}{Scenario 4: Nonlinear Outcomes} \\ \hline 
  -5 & 0/0 & 1000/1000 & 360/718 & 224/282 & 0/0 & 1000/1000 & 58/1000 & 0/0 \\ 
  -2.5 & 0/0 & 998/1000 & 881/985 & 10/15 & 0/0 & 1000/1000 & 0/1000 & 0/0 \\ 
  -2 & 87/217 & 738/783 & 950/996 & 2/4 & 0/0 & 998/1000 & 0/994 & 5/6 \\ 
  -1 & 1000/1000 & 0/0 & 1000/1000 & 0/0 & 1000/1000 & 0/0 & 1000/1000 & 0/0 \\ 
  -0.5 & 1000/1000 & 0/0 & 1000/1000 & 0/0 & 1000/1000 & 0/0 & 1000/1000 & 0/0 \\ 
  0 & 1000/1000 & 0/0 & 1000/1000 & 0/0 & 1000/1000 & 0/0 & 1000/1000 & 0/0 \\ 
  0.5 & 1000/1000 & 0/0 & 1000/1000 & 0/0 & 1000/1000 & 0/0 & 1000/1000 & 0/0 \\ 
  1 & 1000/1000 & 0/0 & 1000/1000 & 0/0 & 999/1000 & 0/0 & 1000/1000 & 0/0 \\ 
  2 & 13/28 & 907/972 & 1000/1000 & 0/0 & 22/28 & 906/972 & 1000/1000 & 0/0 \\ 
  2.5 & 0/0 & 1000/1000 & 1000/1000 & 0/0 & 0/0 & 1000/1000 & 1000/1000 & 0/0 \\ 
  5 & 0/0 & 1000/1000 & 999/1000 & 0/0 & 0/0 & 1000/1000 & 1000/1000 & 0/0 \\ 
   \hline
\end{tabular}}
\end{center} 
\caption{True positive rates (TPR) and true negative rates (TNR) for the diagnostic measures. \label{tab:simulation_tpr_tnr}}
\end{table}

Finally, in Table~\ref{tab:simulation_coverage} we evaluate the 95\% coverage rates for the proposed post-residualized estimators. We see that in all scenarios, we achieve at least nominal coverage. When the population and sample data generating processes diverge significantly, we showed in the previous sections that there could be a loss in efficiency from using post residualized weighting. However, coverage rates are not affected by residualizing.

\begin{table}[ht]
\begin{center}  
\textbf{Coverage Rates} \\ \vspace{2mm} 
\begin{tabular}{l|ccc|ccc}
\hline
& \multicolumn{3}{c|}{\underline{Weighted}} & \multicolumn{3}{c}{\underline{Weighted Least Squares}} \\
$\beta_S$ & $\widehat{\tau}_W$ & $\widehat{\tau}^{res}_W $ & $\widehat{\tau}^{cov}_{W} $ & $\widehat{\tau}_{wLS}$ & $\widehat{\tau}^{res}_{wLS}$ & $\widehat{\tau}^{cov}_{wLS}$\\ \hline 
\multicolumn{6}{l}{Scenario 3: Linear Outcome} \\   \hline
-5 & 0.95 & 0.95 & 0.97 & 0.99 & 0.99 & 0.99 \\ 
  -2.5 & 0.95 & 0.96 & 0.97 & 0.98 & 0.98 & 0.98 \\ 
  -2 & 0.95 & 0.97 & 0.98 & 0.97 & 0.97 & 0.98 \\ 
  -1 & 0.95 & 0.97 & 0.98 & 0.98 & 0.97 & 0.98 \\ 
  -0.5 & 0.95 & 0.98 & 0.99 & 0.98 & 0.98 & 0.98 \\ 
  0 & 0.95 & 0.99 & 0.98 & 0.99 & 0.99 & 0.99 \\ 
  0.5 & 0.95 & 0.97 & 0.98 & 0.99 & 0.99 & 0.99 \\ 
  1 & 0.96 & 0.95 & 0.95 & 0.98 & 0.98 & 0.98 \\ 
  2 & 0.95 & 0.94 & 0.94 & 0.98 & 0.98 & 0.98 \\ 
  2.5 & 0.94 & 0.94 & 0.94 & 0.98 & 0.98 & 0.98 \\ 
  5 & 0.94 & 0.94 & 0.95 & 0.98 & 0.98 & 0.99 \\ \hline 
  \multicolumn{6}{l}{Scenario 4: Nonlinear Outcome} \\ \hline 
   -5 & 0.95 & 0.96 & 0.95 & 0.96 & 0.96 & 0.96 \\ 
  -2.5 & 0.94 & 0.96 & 0.95 & 0.96 & 0.96 & 0.95 \\ 
  -2 & 0.96 & 0.96 & 0.96 & 0.96 & 0.96 & 0.96 \\ 
  -1 & 0.96 & 0.97 & 0.97 & 0.95 & 0.96 & 0.95 \\ 
  -0.5 & 0.95 & 0.95 & 0.96 & 0.95 & 0.95 & 0.96 \\ 
  0 & 0.94 & 0.96 & 0.96 & 0.95 & 0.96 & 0.96 \\ 
  0.5 & 0.95 & 0.96 & 0.96 & 0.96 & 0.96 & 0.97 \\ 
  1 & 0.95 & 0.94 & 0.96 & 0.95 & 0.96 & 0.96 \\ 
  2 & 0.95 & 0.95 & 0.96 & 0.94 & 0.95 & 0.96 \\ 
  2.5 & 0.96 & 0.95 & 0.96 & 0.94 & 0.95 & 0.96 \\ 
  5 & 0.96 & 0.94 & 0.94 & 0.94 & 0.95 & 0.96 \\ 
   \hline
\end{tabular}
\end{center} 
\caption{95\% coverage rates of Normal approximation confidence intervals across 1000 simulations. \label{tab:simulation_coverage}}
\end{table}

\section{Additional Information for Empirical Application} \label{app:empirical}  

As discussed in Section~\ref{sec:empirical}, 
we construct our target population using a leave-one-out procedure.  Table~\ref{tbl:jtpa_summary} provides a summary of the site specific and target population average treatment effects.  More specifically, the difference-in-means (DiM) columns denote the experimental estimate in the specific site. The target PATE is defined as the average difference-in-means estimate across the other 15 sites.  Standard errors are presented in parentheses. Certain sites, such as MT (Butte, MT) contain only 38 experimental units, and the point estimate of the experimental site DiM is vastly different from the target PATE. Thus, we expect the task of generalizing to be more difficult for these sites.

\begin{table}[ht]
\begin{center} 
\textbf{Summary of Experimental Sites and Target Population} \\ \vspace{2mm} 
\resizebox{\textwidth}{!}{%
\begin{tabular}{cllcc|cc|cc}
\hline
& & & & & \multicolumn{2}{c|}{Earnings} & \multicolumn{2}{c}{Employment}  \\
 &  & Expt. & Target Pop & Prob. of & \multicolumn{2}{c|}{(in \$1000)} & \multicolumn{2}{c}{(Percentage)} \\ \cline{6-9}
Site & Location  & Size ($n$) & Size ($N$) & Treatment & DiM & Target PATE & DiM & Target PATE \\   \hline
CC & Corpus Christi, TX & 524 & 5578 & 0.65 & -0.21 (1.16) & 1.37 (1.16) & -0.28 (3.2) & 1.8 (3.2) \\ 
  CI & Cedar Rapids, IA & 190 & 5912 & 0.63 & 1.35 (1.89) & 1.24 (1.89) & -0.77 (5.07) & 1.71 (5.07) \\ 
  CV & Coosa Valley, GA & 788 & 5314 & 0.66 & 1.63 (0.95) & 1.18 (0.95) & 5.95 (2.63) & 0.98 (2.63) \\ 
  HF & Heartland, FL & 234 & 5868 & 0.73 & 0.95 (1.38) & 1.28 (1.38) & 6.8 (5.07) & 1.42 (5.07) \\ 
  IN & Fort Wayne, IN & 1392 & 4710 & 0.67 & 1.73 (0.83) & 1.1 (0.83) & -0.4 (1.58) & 2.23 (1.58) \\ 
  JC & Jersey City, NJ &  81 & 6021 & 0.64 & -0.53 (3.01) & 1.27 (3.01) & -2.39 (9.66) & 1.67 (9.66) \\ 
  JK & Jackson, MO & 353 & 5749 & 0.67 & 2.16 (1.22) & 1.19 (1.22) & 5.66 (4.16) & 1.38 (4.16) \\ 
  LC & Larimer County, CO & 485 & 5617 & 0.69 & 1.61 (1.32) & 1.21 (1.32) & -1.97 (3.24) & 1.93 (3.24) \\ 
  MD & Decatur, IL & 177 & 5925 & 0.70 & 1.24 (2.5) & 1.23 (2.5) & 0.03 (5.24) & 1.67 (5.24) \\ 
  MN & Northwest MN & 179 & 5923 & 0.67 & -1.43 (2.3) & 1.32 (2.3) & -0.52 (6.26) & 1.69 (6.26) \\ 
  MT & Butte, MT &  38 & 6064 & 0.71 & -5.21 (4.1) & 1.27 (4.1) & -7.41 (5.14) & 1.67 (5.14) \\ 
  NE & Omaha, NE & 636 & 5466 & 0.66 & 1.11 (0.98) & 1.25 (0.98) & -1.15 (2.56) & 1.98 (2.56) \\ 
  OH & Marion, OH &  74 & 6028 & 0.70 & -2.99 (2.71) & 1.3 (2.71) & -6.82 (10.37) & 1.74 (10.37) \\ 
  OK & Oakland, CA &  87 & 6015 & 0.64 & 1.83 (3.48) & 1.24 (3.48) & 3.34 (10.77) & 1.57 (10.77) \\ 
  PR & Providence, RI & 463 & 5639 & 0.69 & 3.03 (1.34) & 1.12 (1.34) & 6.78 (4.58) & 1.34 (4.58) \\ 
  SM & Springfield, MO & 401 & 5701 & 0.67 & 0.6 (1.31) & 1.29 (1.31) & 5.44 (3.34) & 1.36 (3.34) \\
   \hline
\end{tabular}}
\end{center} 
\caption{Summary of the JTPA study.} 
\label{tbl:jtpa_summary} 
\end{table}

\subsection{Estimating the Residualizing Model} 
Pre-treatment covariates were taken from the baseline survey conducted at the beginning of the original JTPA experiment, to assess whether or not individuals were eligible for JTPA services. A full list of the covariates included in the residualizing model is provided in Table \ref{tbl:baseline_cov}. In addition to the pre-treatment covariates, we also include normalized measures of previous earnings. Specifically, we include the $z$-score of an individual's previous earnings, relative to the experimental site, as well as the $z$-score of an individual's previous earnings, relative to the entire population. 

\begin{landscape} 
\begin{table}[ht]
\begin{center} 
\textbf{Baseline Covariates included in Residualizing Models} \\ \vspace{2mm} 
\begin{tabular}{llllllll}
\toprule
\textbf{Ethnicity} &  & \textbf{Weeks Worked$^\dagger$} &  & \textbf{Public Assistance History} &  & \textbf{Family Income$^\dagger$} &  \\
White &  & Zero &  & Food Stamps &  & Less than \$3,000 &  \\
Black &  & 1-26 weeks &  & Cash Welfare, other than AFDC &  & \$3,000-\$6,000 &  \\
Hispanic &  & 27-52 weeks &  & Unemployment Benefits &  & More than \$6,000 &  \\
AAPI &  &  &  &  &  &  &  \\
 &  & \textbf{Earnings} &  & \textbf{AFDC Histories} &  & \textbf{Accessibility} &  \\
\textbf{Education} &  & Previous Earnings$^\ddag$ & * & Ever AFDC case head &  & Driver's License &  \\
ABE/ESL &  & Weekly Pay & * & Case head anytime$^\dagger$ &  & Car available for regular use &  \\
High school diploma &  & Quantile within Site & * & Received AFDC$^\dagger$ &  & Telephone at home &  \\
GED certificate &  & \hspace{3mm}$<$ 25\% & * & Years as AFDC case head: &  &  &  \\
Some college &  & \hspace{3mm}$>$ 50\% & * & \hspace{3mm}Less than 2 years &  & \textbf{Household Composition} &  \\
Occupational Training &  & \hspace{3mm}$>$ 90\% & * & \hspace{3mm}2-5 years &  & Marital Status &  \\
Technical certificate &  & Quantile across Experiment & * & \hspace{3mm}More than 5 years &  & Spouse present &  \\
Job search assistance &  & \hspace{3mm}$<$ 25\% & * &  &  & Household Size &  \\
Years of Education$^\ddag$ & * & \hspace{3mm}$>$ 50\% & * & \textbf{Age} &  & Number of children present &  \\
 &  & \hspace{3mm}$>$ 90\% & * & Age$^\ddag$ & * & Child under 6 present &  \\
\textbf{Work History} &  & Non-Zero Previous Earnings & * & Age Buckets &  &  &  \\
Ever employed &  & UI Reported Earnings &  & \hspace{3mm}20-21 &  & \textbf{Geographic Region} &  \\
Employed upon application &  &  &  & \hspace{3mm}22-29 &  & West & * \\
Total earnings$^\dagger$ &  & \textbf{Living in Public Housing} &  & \hspace{3mm}30-44 &  & Midwest & * \\
Hourly earnings &  & Yes &  & \hspace{3mm}45-54 &  & South & * \\
Hours worked & * &  &  & \hspace{3mm}55 or older &  & North & * \\ \bottomrule
\end{tabular}
\end{center} 
\caption{We provide a list of all of the covariates included in the Super Learner. Many of these variables were included in the original JTPA study's regression model. Any variable denoted with an asterisk $(*)$ was not included in the original JTPA study's regression model. $^\dagger$ indicates that the measure is from the past 12 months prior to the baseline survey, $^\ddag$ indicates higher order terms included of that variable.} 
\label{tbl:baseline_cov} 
\end{table}
\end{landscape} 

\subsection{Numerical Results for Empirical Application}

Table~\ref{tbl:jtpa_coverage} provides numerical results for the mean absolute error across all 16 experimental sites for the six different estimators. We note that the mean absolute error of the point estimates do not vary substantially from using post-residualized weighting.  This supports the results in Section~\ref{sec:empirical_bias}.

\begin{table}[t]
  \begin{center} 
  \scalebox{1}{
\begin{tabular}{lccccccc}
  \hline
  & \multicolumn{3}{c}{Weighted Estimator} && \multicolumn{3}{c}{Weighted Least Squares} \\ 
  \cline{2-4} \cline{6-8}
& $\widehat{\tau}_W$ & $\widehat{\tau}^{res}_W $& $\widehat{\tau}^{cov}_{W} $ && $\widehat{\tau}_{wLS}$ & $\widehat{\tau}^{res}_{wLS}$ & $\widehat{\tau}^{cov}_{wLS}$  \\ 
  \hline
  Earnings  & 2.37 & 2.07 & 2.19 && 2.46 & 2.28 & 2.24 \\ 
Employment ($\times$ 100) & 8.53 & 8.06 & 8.21 && 7.95 & 7.65 & 8.01 \\ 
   \hline
\end{tabular}}
\end{center} 
\caption{Mean absolute error across sites.}
\label{tbl:jtpa_coverage}
\end{table}

Table~\ref{tab:good_fit} reports the estimated standard errors (columns 3-5 for weighted estimators and columns 8-10 for weighted least squares estimators) for each site, along with the estimated diagnostics (columns 6-7 for weighted estimators and columns 11-12 for weighted least squares estimators). In general, the diagnostics are able to adequately determine whether or not we expect there to be improvements in standard error for accounting for the population outcome information, as discussed in Section~\ref{sec:empirical_diag}.

\begin{table}[!ht]
\begin{center} 
\textbf{Standard Errors and Diagnostics for Residualizing Models for Residualizing Models} \\
\vspace{2mm} 
\resizebox{\textwidth}{!}{%
\begin{tabular}{lc|ccccc|ccccc}
\hline
 & & \multicolumn{5}{c|}{Weighted} &  \multicolumn{5}{c}{Weighted Least Squares} \\\cmidrule(lr){3-7} \cmidrule(lr){8-12}
 Site & $n$ & $\widehat{\tau}_W$ & $\widehat{\tau}^{res}_W $ & $\widehat{\tau}^{cov}_{W} $ & $\hat R^2_0$ & $\hat R^2_{0,cov}$ & $\widehat{\tau}_{wLS}$ & $\widehat{\tau}^{res}_{wLS}$ & $\widehat{\tau}^{cov}_{wLS}$ & $\hat R_{0,wLS}^2$ & $\hat R^2_{0, wLS, cov}$ \\ \hline
\multicolumn{12}{l}{Outcome: Earnings} \\ \hline
NE & 636 & 1.70 & 1.53 & 1.53 & 0.23 & 0.22 & 1.58 & 1.53 & 1.53 & 0.08 & 0.06 \\ 
  LC & 485 & 2.46 & 2.02 & 2.11 & 0.42 & 0.32 & 2.40 & 2.08 & 2.14 & 0.38 & 0.26 \\ 
  HF & 234 & 1.88 & 1.63 & 1.66 & 0.36 & 0.19 & 1.87 & 1.66 & 1.69 & 0.42 & 0.18 \\ 
  IN & 1392 & 1.03 & 0.93 & 0.92 & 0.25 & 0.26 & 1.00 & 0.91 & 0.91 & 0.22 & 0.21 \\ 
  CV & 788 & 1.40 & 1.25 & 1.22 & 0.04 & 0.01 & 1.36 & 1.23 & 1.20 & 0.08 & 0.07 \\ 
  CC & 524 & 2.51 & 2.52 & 2.48 & -0.06 & -0.18 & 2.42 & 2.42 & 2.39 & -0.13 & -0.23 \\ 
  JK & 353 & 2.29 & 2.28 & 2.25 & 0.19 & -0.25 & 2.19 & 2.18 & 2.16 & 0.30 & 0.10 \\ 
  MT &  38 & 6.44 & 7.04 & 8.40 & -0.36 & -9.31 & 4.64 & 4.83 & 6.09 & 0.40 & -3.90 \\ 
  PR & 463 & 2.69 & 2.61 & 2.60 & 0.08 & -0.16 & 2.82 & 2.75 & 2.71 & 0.03 & -0.17 \\ 
  MN & 179 & 4.79 & 4.70 & 4.80 & -0.03 & -0.35 & 3.72 & 4.26 & 4.20 & -0.31 & -0.56 \\ 
  MD & 177 & 2.87 & 2.46 & 2.48 & 0.33 & 0.24 & 2.67 & 2.30 & 2.32 & 0.30 & 0.13 \\ 
  SM & 401 & 2.07 & 2.28 & 2.12 & -0.30 & -0.13 & 2.13 & 2.23 & 2.11 & -0.14 & -0.09 \\ 
  OH &  74 & 3.97 & 3.27 & 3.42 & 0.33 & -0.22 & 3.94 & 3.75 & 3.77 & 0.29 & -0.37 \\ 
  CI & 190 & 3.84 & 3.33 & 3.07 & 0.41 & 0.31 & 3.47 & 3.15 & 2.94 & 0.28 & -0.18 \\ 
  OK &  87 & 4.69 & 5.07 & 4.64 & -0.05 & -0.43 & 4.61 & 4.39 & 4.22 & 0.14 & -0.19 \\ 
  JC &  81 & 7.24 & 8.81 & 8.50 & -0.75 & -1.15 & 6.14 & 7.51 & 6.56 & -0.19 & -0.49 \\ 
  \hline
      \multicolumn{12}{l}{Outcome: Employment} \\ 
  \hline
  NE & 636 & 0.04 & 0.04 & 0.04 & 0.03 & -0.01 & 0.04 & 0.04 & 0.04 & 0.02 & -0.00 \\ 
  LC & 485 & 0.06 & 0.06 & 0.05 & 0.19 & 0.20 & 0.06 & 0.06 & 0.05 & 0.11 & 0.09 \\ 
  HF & 234 & 0.06 & 0.06 & 0.06 & 0.04 & -0.03 & 0.06 & 0.06 & 0.06 & 0.02 & -0.03 \\ 
  IN & 1392 & 0.02 & 0.02 & 0.02 & -0.15 & -0.04 & 0.02 & 0.02 & 0.02 & -0.21 & -0.08 \\ 
  CV & 788 & 0.03 & 0.03 & 0.03 & -0.01 & -0.01 & 0.03 & 0.03 & 0.03 & -0.03 & -0.01 \\ 
  CC & 524 & 0.06 & 0.06 & 0.06 & -0.09 & -0.10 & 0.06 & 0.06 & 0.06 & -0.10 & -0.08 \\ 
  JK & 353 & 0.10 & 0.09 & 0.09 & 0.13 & -1.53 & 0.09 & 0.09 & 0.09 & 0.08 & -0.85 \\ 
  MT &  38 & 0.13 & 0.13 & 0.13 & --- & ---  & 0.13 & 0.15 & 0.14 & --- &  --- \\ 
  PR & 463 & 0.06 & 0.06 & 0.06 & 0.03 & -0.05 & 0.07 & 0.06 & 0.07 & 0.03 & -0.03 \\ 
  MN & 179 & 0.13 & 0.12 & 0.11 & 0.23 & -3.09 & 0.12 & 0.11 & 0.11 & 0.21 & -0.89 \\ 
  MD & 177 & 0.09 & 0.08 & 0.08 & 0.19 & -0.06 & 0.08 & 0.08 & 0.08 & 0.20 & -4.0e28 \\ 
  SM & 401 & 0.06 & 0.06 & 0.06 & 0.07 & -0.15 & 0.06 & 0.06 & 0.06 & 0.05 & -0.03 \\ 
  OH &  74 & 0.07 & 0.06 & 0.07 & 0.09 & -1.78 & 0.08 & 0.07 & 0.08 & 0.08 & -1.8e28 \\ 
  CI & 190 & 0.04 & 0.04 & 0.05 & 0.19 & -0.07 & 0.05 & 0.05 & 0.05 & 0.27 & -4.4e28 \\ 
  OK &  87 & 0.21 & 0.19 & 0.19 & 0.12 & -2.03 & 0.17 & 0.17 & 0.18 & 0.21 & -1.35 \\ 
  JC &  81 & 0.16 & 0.14 & 0.14 & -0.88 & -3.88 & 0.12 & 0.13 & 0.13 & -0.92 & -0.58 \\ 
   \hline
\end{tabular}}
\end{center}
\caption{Standard error and diagnostic values for post-residualized weighting across the 16 experimental sites for two primary outcomes--earnings and employment. The diagnostic values for the site of Butte, Montana (MT) are null when outcome is employment, because all units in the control group were unemployed.}
\label{tab:good_fit} 
\end{table}

Finally, Table~\ref{tab:diag} presents the true positive rate and false positive rate for our diagnostics across the sites where the diagnostic indicated residualizing would increase precision (or not).  We present these counts for both outcomes, separately.

\begin{table}[ht]
  \begin{center} 
    \begin{tabular}{lccccccc}
      \hline
      & \multicolumn{2}{c}{Weighted Estimator} & & \multicolumn{2}{c}{Weighted Least Squares} \\ 
      \cline{2-3} \cline{5-6}
      & $\widehat{\tau}^{res}_W $& $\widehat{\tau}^{cov}_{W} $ & & $\widehat{\tau}^{res}_{wLS}$ & $\widehat{\tau}^{cov}_{wLS}$\\ 
      \hline
      Earnings & \\
      \quad True Positive Rate & 10/11 & 7/12 & & 11/11 & 7/13 \\       
      \quad True Negative Rate & 5/5  &4/4  & & 4/5 & 3/3\\ \midrule
      Employment  \\
      \quad True Positive Rate & 11/13 & 1/12 & & 10/10 & 1/7 \\ 
      \quad True Negative Rate & 3/3 & 4/4 & & 5/6& 8/9 \\  
      \hline
    \end{tabular}
    \end{center} 
  \caption{Performance of proposed diagnostic measures, as measured through the true positive rate and false positive rate.}\label{tab:diag}
\end{table}

\subsection{Using Proxy Outcomes}
To illustrate use of a proxy outcome, we run the same analysis as in Section~\ref{sec:empirical}, 
except we use employment as a proxy for earnings, and vice versa when building the residualizing model.  This mimics a situation in which we have access to a related, but different outcome measure in our target population.  Because employment is binary while earnings are continuous, we expect that direct residualizing may not result in substantial efficiency gains, and thus that our diagnostic measures would indicate not to residualize. However, treating $\hat Y_i$ as a covariate should still result in efficiency gain, as earnings and employment are correlated and the model can adjust for the scaling differences.

\subsubsection{Bias} Table~\ref{tbl:jtpa_coverage_proxy} presents the mean absolute error of the different estimation methods. When earnings is the outcome, both directly residualizing and using $\hat Y_i$ as a covariate result in relatively stable performance. However, when employment is the outcome, the scaling differences between earnings (in $\$1000$) and the binary employment measure lead to large residuals.  We see a loss to precision from direct residualizing, and exacerbated finite sample bias.  However, when including $\hat Y_i$ as a covariate, we are able to account for the scaling differences, and the mean absolute error is lower.

\begin{table}[ht]
\begin{center} 
\textbf{Estimator Performance Summary with Proxy Outcomes}\\ \vspace{2mm}
\begin{tabular}{lccccccc}
  \hline
  & \multicolumn{3}{c}{Weighted} && \multicolumn{3}{c}{Weighted Least Squares} \\ 
  \cline{2-4} \cline{6-8}
& $\widehat{\tau}_W$ & $\widehat{\tau}^{res}_W $ & $\widehat{\tau}^{cov}_{W} $ && $\widehat{\tau}_{wLS}$ & $\widehat{\tau}^{res}_{wLS}$ & $\widehat{\tau}^{cov}_{wLS}$ \\ 
  \hline
  Earnings & 2.37 & 2.35 & 2.14 && 2.46 & 2.44 & 2.21 \\ 
Employment ($\times$ 100) & 8.53 & 66.15 & 7.85 && 7.95 & 65.33 & 7.45 \\ 
   \hline
\end{tabular}
\end{center} 
\caption{Mean absolute errors for each estimator, across all experimental sites when using proxy outcomes.} 
\label{tbl:jtpa_coverage_proxy}
\end{table}

\subsubsection{Diagnostics} We estimate the same diagnostics as in Section~\ref{sec:empirical_bias} 
to determine when to expect precision gains from performing post-residualized weighting. We summarize the true positive and true negative rates of the diagnostic in Table \ref{tab:diag_proxy}.  We see that the performance of the diagnostic is good for direct residualization.  However, we see that the diagnostic for including $\hat Y_i$ as a covariate is relatively conservative, and fails to identify all the cases in which it is beneficial to residualize. However, the true negative rate of the diagnostic for including $\hat Y_i$ as a covariate is very high (almost 100\%), which indicates that the diagnostic is very effective at identifying when residualizing fails to lead to precision gain.

\begin{table}[!ht]
  \begin{center} 
    \begin{tabular}{lccccccc}
      \hline
      & \multicolumn{2}{c}{Weighted Estimator} & & \multicolumn{2}{c}{Weighted Least Squares} \\ 
      \cline{2-3} \cline{5-6}
      & $\widehat{\tau}^{res}_W $& $\widehat{\tau}^{cov}_{W} $ & & $\widehat{\tau}^{res}_{wLS}$ & $\widehat{\tau}^{cov}_{wLS}$\\ 
      \hline
      Earnings & \\
      \quad True Positive Rate & 13/14 & 7/12 & & 12/13 & 6/13  \\       
      \quad True Negative Rate &  2/2  &4/4  & & 2/3 & 3/3\\ \midrule
      Employment  \\
      \quad True Positive Rate & -- & 2/12 & & -- & 3/10 \\ 
      \quad True Negative Rate & 16/16 & 4/4 & & 16/16& 5/6 \\  
      \hline
    \end{tabular}
    \end{center}
  \caption{Performance of proposed diagnostic measures using proxy outcomes, as measured through the true positive rate and false positive rate.}\label{tab:diag_proxy}
\end{table}

Table \ref{tab:good_fit_proxy} provides the standard errors and diagnostic measures for each site and estimator.  Within the ``Weighted'' and ``Weighted Least Squares'' sections, the left three columns present the standard error for the corresponding estimator for each site, and the right two columns present the diagnostic measure.  One key takeaway is that, when employment is the outcome, using earnings as a proxy outcome results in large scaling differences between our residualizing model, captured by $\hat Y_i$, and the true outcome measure.  This is unsurprising since earnings is continuous and employment is binary.  As a result, the $\hat R^2_0$ measures for the estimators that use direct residualizing (i.e., $\hat \tau_W^{res}$ and $\hat \tau_{wLS}^{res}$) are all negative, indicating that we should not use direct residualizing in that setting.  However, even in this scenario, the diagnostic for using $\hat Y_i$ as a covariate does not indicate significant gains.  When using employment as a proxy for earnings, the diagnostics indicate small gains to direct residualizing across most sites, and gains from including $\hat Y_i$ as a covariate across about half of sites.

\begin{table}[!ht]
\begin{center} 
\textbf{Standard Errors and Diagnostics for Residualizing Models with Proxy Outcomes} \\
\vspace{2mm} 
\resizebox{\textwidth}{!}{%
\begin{tabular}{lc|ccccc|ccccc}
\hline
 & & \multicolumn{5}{c|}{Weighted} &  \multicolumn{5}{c}{Weighted Least Squares} \\\cmidrule(lr){3-7} \cmidrule(lr){8-12}
 Site & $n$ & $\widehat{\tau}_W $ & $\widehat{\tau}^{res}_W $ & $\widehat{\tau}^{cov}_{W} $ & $\hat R^2_0$ & $\hat R^2_{0,cov}$ & $\widehat{\tau}_{wLS}$ & $\widehat{\tau}^{res}_{wLS}$ & $\widehat{\tau}^{cov}_{wLS}$ & $\hat R_{0,wLS}^2$ & $\hat R^2_{0, wLS, cov}$ \\ \hline
\multicolumn{12}{l}{Outcome: Earnings} \\
  \hline
NE & 636 & 1.70 & 1.70 & 1.62 & 0.00 & 0.09 & 1.58 & 1.58 & 1.53 & 0.00 & 0.03 \\ 
  LC & 485 & 2.46 & 2.45 & 2.39 & 0.00 & 0.15 & 2.40 & 2.40 & 2.37 & 0.00 & 0.05 \\ 
  HF & 234 & 1.88 & 1.88 & 1.76 & 0.01 & 0.15 & 1.87 & 1.86 & 1.78 & 0.01 & 0.08 \\ 
  IN & 1392 & 1.03 & 1.03 & 0.96 & 0.01 & 0.27 & 1.00 & 0.99 & 0.95 & 0.01 & 0.21 \\ 
  CV & 788 & 1.40 & 1.40 & 1.39 & 0.00 & -0.06 & 1.36 & 1.36 & 1.37 & 0.00 & -0.03 \\ 
  CC & 524 & 2.51 & 2.51 & 2.46 & 0.01 & -0.03 & 2.42 & 2.41 & 2.40 & 0.00 & -0.10 \\ 
  JK & 353 & 2.29 & 2.28 & 2.10 & 0.01 & 0.07 & 2.19 & 2.18 & 2.07 & 0.01 & 0.04 \\ 
  MT &  38 & 6.44 & 6.44 & 9.60 & -0.00 & -9.23 & 4.64 & 4.65 & 7.32 & 0.01 & -5.86 \\ 
  PR & 463 & 2.69 & 2.69 & 2.70 & 0.00 & -0.16 & 2.82 & 2.82 & 2.82 & -0.00 & -0.15 \\ 
  MN & 179 & 4.79 & 4.78 & 4.13 & 0.00 & 0.13 & 3.72 & 3.71 & 3.71 & 0.00 & -0.14 \\ 
  MD & 177 & 2.87 & 2.87 & 2.61 & 0.01 & 0.14 & 2.67 & 2.66 & 2.43 & 0.01 & 0.16 \\ 
  SM & 401 & 2.07 & 2.07 & 2.04 & 0.00 & -0.07 & 2.13 & 2.12 & 2.06 & 0.00 & -0.02 \\ 
  OH &  74 & 3.97 & 3.97 & 4.00 & 0.00 & -0.44 & 3.94 & 3.93 & 3.75 & 0.00 & -0.50 \\ 
  CI & 190 & 3.84 & 3.84 & 3.40 & 0.00 & -0.03 & 3.47 & 3.47 & 3.07 & 0.00 & -0.22 \\ 
  OK &  87 & 4.69 & 4.71 & 4.51 & -0.01 & -0.88 & 4.61 & 4.61 & 4.06 & -0.01 & -0.67 \\ 
  JC &  81 & 7.24 & 7.26 & 8.52 & -0.01 & -0.82 & 6.14 & 6.17 & 6.75 & -0.01 & -0.83 \\ 
  \hline 
    \multicolumn{12}{l}{Outcome: Employment} \\ \hline
   NE & 636 & 0.04 & 0.56 & 0.04 & -352.40 & -0.00 & 0.04 & 0.49 & 0.04 & -248.43 & -0.01 \\ 
  LC & 485 & 0.06 & 0.70 & 0.05 & -220.79 & 0.13 & 0.06 & 0.56 & 0.05 & -193.43 & 0.02 \\ 
  HF & 234 & 0.06 & 0.94 & 0.06 & -260.76 & -0.02 & 0.06 & 0.90 & 0.06 & -282.80 & 0.06 \\ 
  IN & 1392 & 0.02 & 0.34 & 0.02 & -391.67 & -0.04 & 0.02 & 0.32 & 0.02 & -354.59 & -0.05 \\ 
  CV & 788 & 0.03 & 0.42 & 0.03 & -151.95 & 0.02 & 0.03 & 0.38 & 0.03 & -129.68 & 0.03 \\ 
  CC & 524 & 0.06 & 1.00 & 0.06 & -284.99 & -0.21 & 0.06 & 0.89 & 0.06 & -236.05 & -0.18 \\ 
  JK & 353 & 0.10 & 1.10 & 0.08 & -104.99 & -3.17 & 0.09 & 0.92 & 0.08 & -88.67 & -2.13 \\ 
  MT &  38 & 0.13 & 2.42 & 0.12 & --- & --- & 0.13 & 2.49 & 0.13 & --- & --- \\ 
  PR & 463 & 0.06 & 0.93 & 0.06 & -228.05 & -0.05 & 0.07 & 0.80 & 0.07 & -200.67 & -0.06 \\ 
  MN & 179 & 0.13 & 1.57 & 0.13 & -207.13 & -14.37 & 0.12 & 1.61 & 0.12 & -189.16 & -3.35 \\ 
  MD & 177 & 0.09 & 0.93 & 0.08 & -66.00 & -0.14 & 0.08 & 0.81 & 0.08 & -77.66 & -4.7e28 \\ 
  SM & 401 & 0.06 & 0.76 & 0.06 & -95.56 & -0.12 & 0.06 & 0.68 & 0.06 & -84.67 & -0.10 \\ 
  OH &  74 & 0.07 & 1.77 & 0.07 & -1202.77 & -0.56 & 0.08 & 1.60 & 0.08 & -985.02 & -2.4e28 \\ 
  CI & 190 & 0.04 & 1.40 & 0.05 & -1312.75 & -0.47 & 0.05 & 1.41 & 0.05 & -1241.70 & -1.1e28 \\ 
  OK &  87 & 0.21 & 3.24 & 0.16 & -249.20 & -1.29 & 0.17 & 2.44 & 0.16 & -65.65 & -0.28 \\ 
  JC &  81 & 0.16 & 3.20 & 0.17 & -6487.60 & -4.51 & 0.12 & 1.70 & 0.14 & -300.75 & -0.24 \\ 
   \hline
\end{tabular}}
\end{center} 
\caption{Standard error and diagnostic values for post-residualized weighting using proxy outcomes across the 16 experimental sites for two primary outcomes–earnings and employment. Once again, the diagnostics for MT are null when employment is the outcome, because all the units in the control group are unemployed.}
\label{tab:good_fit_proxy} 
\end{table}

\subsubsection{Efficiency Gain} Table~\ref{tab:good_fit_proxy} presents the standard errors of each weighting method, with and without post-residualizing, for each site. Table~\ref{tbl:jtpa_estimator_perf_proxy} presents the average standard error across sites for post-residualized weighting using proxy outcomes, where we restrict our attention to the sites identified by the diagnostic measures for when we expect precision gains.  When using employment as a proxy for earnings, direct residualizing indicates small gains in 13/16 sites, and including $\hat Y_i$ as a covariate indicates gains in just under half of sites. 
The relative improvement in variance is small due to the differences in magnitude between $\hat Y_i$ and $Y_i$.  In particular, we see around a 0.3-0.4\% reduction in variance from performing direct residualizing. However, when including $\hat Y_i$ as a covariate, which accounts for the scaling difference, the improvements are more substantial. In particular, when using $\hat Y_i$ as a covariate in the weighted estimator, there is a 14\% reduction in variance. Using weighted least squares, there is a 9\% reduction in variance from including $\hat Y_i$ as a covariate. The primary takeaway to highlight is that using $\hat Y_i$ as a covariate to perform post-residualized weighting can allow us to leverage proxy outcomes that exist on different scales than the outcome of interest, where we expect greater gains the more closely related the outcome and proxy outcome are. 

For employment, we do not consider direct residualizing because the diagnostic measure did not identify any experimental sites in which directly residualizing would lead to precision gains.  When including $\hat Y_i$ as a covariate the diagnostic indicated 5 sites that indicate gains from post-residualized weighting; among these we see a 5\% reduction in variance when using $\hat Y_i$ as a covariate in the weighted estimator, and a 1\% reduction in variance in the weighted least squares estimator.  Finally, we emphasize that estimating the PATE results in variance inflation relative to the within-sample difference-in-means, as expected.  However, we see that post-residualized weighting can offset some of this loss in precision.

This exercise shows how a proxy outcome can be used for building the residualizing model.  When the two variables are on very different scales, we expect that direct residualizing would not be beneficial, as evidenced here and captured by our diagnostic measures.  Including $\hat Y_i$ as a covariate addresses scaling concerns, although as we see when using earnings as a proxy for employment, does not always allow for gains.  We see that even using proxy outcomes, our diagnostic measures can accurately capture when there is potential for precision gains, and our post-residualized weighting method can lead to precision gains in estimation of the target PATE.

\begin{table}[!ht]
\centering 
\textbf{Summary of Standard Errors across Experimental Sites Subset by Diagnostic, using Proxy Outcomes} \\   \vspace{2mm} 
\resizebox{\textwidth}{!}{%
\begin{tabular}{l|cccc|cccc} \toprule
& \multicolumn{4}{c|}{\underline{Earnings}} & \multicolumn{4}{c}{\underline{Employment}} \\ 
& \begin{tabular}{@{}c@{}}Number\\of Sites\end{tabular} & DiM & Standard& \begin{tabular}{@{}c@{}}Post Resid. \\ Weighting\end{tabular}
& \begin{tabular}{@{}c@{}}Number\\of Sites\end{tabular} & DiM & Standard& \begin{tabular}{@{}c@{}}Post Resid. \\ Weighting\end{tabular}\\\midrule
\hspace{3mm} Weighted  & & & & & &\\
$\qquad$ Direct Residualizing & 13& 1.53 & 2.58 & 2.57 & 0 & -- & -- & --  \\ 
$\qquad$ $\hat Y_i$ as Covariate & 7  & 1.50 & 2.43 & 2.23 & 2 & 2.93 & 4.01 & 4.00\\ 
\hspace{3mm} Weighted Least Squares  & & & & & &\\ 
$\qquad$ Direct Residualizing & 13 & 1.74 & 2.57 & 2.57  & 0 & -- & -- & -- \\ 
$\qquad$ $\hat Y_i$ as Covariate & 6 & 1.37& 1.95 & 1.85 & 3 & 3.65 & 4.68 & 4.63 \\ \midrule
\end{tabular}}
\caption{Summary of standard errors across the 16 experimental sites identified by the diagnostic measures.} 
\label{tbl:jtpa_estimator_perf_proxy}
\end{table}

\newpage
\printbibliography

\end{document}